\documentclass[10pt]{article}

\usepackage{amsthm}
\usepackage[cmex10]{amsmath}
\usepackage{amssymb}
\usepackage{dblfloatfix}
\usepackage{lineno,hyperref}
\usepackage[colorinlistoftodos]{todonotes}
\usepackage{stackrel}
\usepackage{pdfpages}
\usepackage{graphics}
\usepackage{epsfig}
\usepackage{times}
\usepackage{amsmath}
\usepackage{amssymb}
\usepackage{hyperref}
\usepackage{color}
\usepackage{arydshln}
\usepackage{arydshln}
\usepackage{gensymb}

\newcommand{\cF}{{\cal F}}
\newcommand{\bR}{{\bf R}}

\newcommand{\bC}{{\bf C}}
\newcommand{\bZ}{{\bf Z}}
\newcommand{\bX}{{\bf X}}

\newcommand{\cT}{{\cal T}}
\newcommand{\cP}{{\cal P}}

\newcommand{\bx}{{\bf x}}
\newcommand{\by}{{\bf y}}

\newcommand{\ba}{{\bf a}}

\newcommand{\cS}{{\cal S}}
\newcommand{\cD}{{\cal D}}

\newcommand{\bitem}{\begin{itemize}}
\newcommand{\eitem}{\end{itemize}}
\newcommand{\goto}{\rightarrow}
\newcommand{\beq}{\begin{equation}}
\newcommand{\eeq}{\end{equation}}

\newcommand{\bi}{ {\bf i} }

\newtheorem{thm}{Theorem}[section]
\newtheorem{cor}[thm]{Corollary}
\newtheorem{lem}[thm]{Lemma}

\newtheorem{dfn}[subsection]{Definition}

\DeclareMathOperator*{\argmin}{arg\,min}

\newcounter{comment}

\usepackage{color}
\usepackage[normalem]{ulem}

\newcommand{\epreg}{\eps^*_{\mbox{\small bd}}}

\newcommand{\eps}{\epsilon}
\newcommand{\easy}{\eps^*_{\small\mbox{asy}}}
\newcommand{\emb}{\eps^*_{\small\mbox{mb}}}
\newcommand{\esb}{\eps^*_{\small\mbox{sb}}}

\newcommand{\cFu}{{\cal F}_{\aus}}
\newcommand{\cK}{{\cal K}}

\begin{document}

\pagenumbering{arabic}

\title{Sparsity/Undersampling 
Tradeoffs in Anisotropic Undersampling, 
with Applications in MR Imaging/Spectroscopy}
%\title{Compressed Sensing Phase Transitions in Multi-dimensional NMR}

\author{Hatef Monajemi$^{*}$,
David L. Donoho\thanks{Department of Statistics, Stanford University, Stanford, CA}
}

%\contributor{Submitted to Proceedings of the National Academy of Sciences
%f the United States of America}
% \date{Original: August 2016;\\
% This Version: March 9, 2017}
\date{March 16, 2018}
\maketitle

%%%%%%%%%%%%%%%%%%%%%%%%%%%%%%%%%%%%%%%%%%%%%%%%%%%%%%%%%%%%%%%%

\begin{abstract}
%\comment[HM]{This document will be synced to Chapter 6 of Hatef Monajemi's Thesis}

\noindent %The free-induction decays in a $d$-dimensional NMR experiment belong to the set of $2^d$-dimensional hypercomplex numbers. In a partial-component sampling one acquires at each time pixel a random subset of the $2^d$ hypercomplex components. 
We study {\it anisotropic} {undersampling} schemes like those used
in multi-dimensional NMR spectroscopy and MR imaging, 
which sample exhaustively in certain time dimensions and randomly in others. 

Our analysis shows that anisotropic undersampling schemes are equivalent 
to certain block-diagonal measurement systems.  
We develop novel exact formulas for the sparsity/undersampling tradeoffs
in such measurement systems, assuming uniform sparsity fractions in each column.
Our formulas predict  finite-$N$ phase transition behavior 
differing substantially from the well-known asymptotic phase transitions for classical Gaussian undersampling.
Extensive empirical work shows that our formulas accurately describe observed finite-$N$ behavior,
while the usual formulas based on universality are substantially inaccurate at the moderate $N$ involved
in realistic applications.

We also vary the anisotropy, keeping the total number of samples fixed, and for each variation 
we determine the precise sparsity/undersampling tradeoff (phase transition). 
We show that, other things being equal, 
the ability to recover a sparse object decreases with 
an increasing number of exhaustively sampled dimensions. 

\end{abstract}

{\bf keywords:} Sparse Recovery, Compressed Sensing, Block Diagonal Measurement Matrix.\footnote{It is a pleasure to acknowledge discussions with the pioneers of anisotropic undersampling: Jeffrey Hoch and Adam Schuyler (U. Conn. Health Sciences), Michael Lustig (UC Berkeley), John Pauly (Stanford). This research was
partially supported by NSF-DMS 1418362 and NSF-DMS 1407813. We would also like to thank the Stanford Research Computing Center for providing computational resources and support that were essential to these research results. Thanks also to the anonymous referees for many thoughtful comments.} 

\section{Introduction}
% Compressed Sensing (CS) is a modern signal processing paradigm that allows acquisition of \emph{sparse} signals using fewer observations than what traditionally thought to be needed. Since its introduction almost a decade ago \cite{Donoho1,CandesTao}, CS
% has been successfully applied to a wide range of problems in science and engineering from
% communications networks to MR Imaging, linear classification, genotyping, astronomy and
% more recently multi-dimensional NMR, electronic spectroscopy, and super-resolution microscopy\cite{SparseMRI, Calder09, Erlich10, Huang13, Sanders12, Maciejewski11, zhu12}.

\subsection{Background}
In Compressed Sensing (CS), one wishes to reconstruct an 
$N$-dimensional discrete signal $x_0$
using $n < N$  measurements.  Theory shows that
if $x_0$ is sufficiently sparse, and the $n \times N$ sensing matrix
$W$ is an \emph{i.i.d} Gaussian random matrix, then $x_0$ can be reconstructed accurately and reliably
from measurements $y= W x_0$ using convex optimization; 
see many papers and books, such as
\cite{Donoho1,CandesTao,DonohoHuo,CSMRI,DoTa10,Calder08,Tropp07,Applebaum}. %{\textcolor{red}{ADD OTHER REFS}}

In general, at a given fixed level of undersampling, 
the chance of successful recovery depends on the
sparsity of the underlying object,  in an almost-binary fashion.
Namely, suppose that the object $x_0$ is $k$-sparse---has
at most $k$ nonzero entries---and consider the situation
where $k \sim \epsilon N $ and $n \sim \delta N $.
Then, as a function of undersampling fraction $\delta = n/N$, 
there is, asymptotically for large $N$ 
a definite interval for the sparsity fraction $\epsilon = k/N$ that
permits successful recovery, while outside this range,
recovery is unsuccessful.

Figure \ref{figure-phase-transition} depicts the situation 
for Gaussian measurement matrices $W$.
It shows a so-called phase diagram $(\epsilon,\delta) \in (0,1)^2$
and a curve $\easy(\delta)$ 
separating a `success' phase from a `failure' phase. Namely,
if $\epsilon < \easy(\delta)$, 
then, with overwhelming probability for large $N$,
convex optimization will recover $x_0$ exactly; while on the other hand,
if $\epsilon > \easy(\delta)$, then, with overwhelming probability, 
convex optimization will fail. 
%Moreover, for finite-$N$ problems one observes a transition zone that narrows down with increasing problem size. This phenomenon is known as \emph{phase transition} in compressed sensing \cite{DoTa10}.

\begin{figure}[h]
\centering
\includegraphics[height=3in]{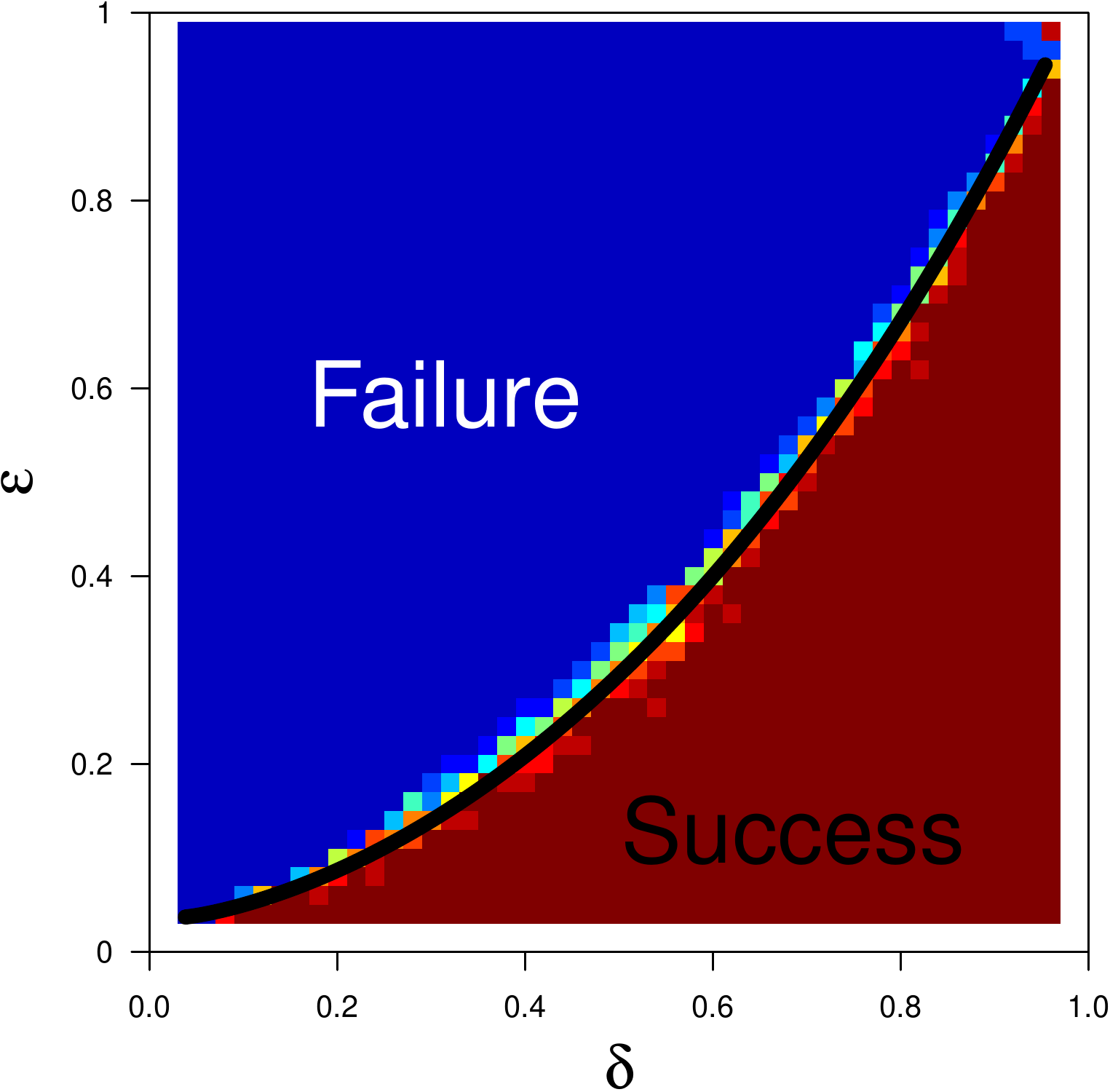}
\caption{Success and failure regions for `classical' compressed sensing with Gaussian measurement matrices; assumes object to recover is sparse and nonnegative. The asymptotic phase transition curve $\easy(\delta)$ (solid black line) separates the two regions. Shaded attribute gives fraction of successful reconstructions. Red, 100\%; blue, 0\%. In this experiment, $n=250$. 
%Empirical ($\epsilon$-$\delta$)-phase diagram for the cyclic ensemble. Vertical axis: Graphic. Horizontal axis: Graphic. . Dashed line, asymptotic Gaussian prediction  Graphic. In this experiment, Graphic, Graphic.
}
\label{figure-phase-transition}
\end{figure}

%CS with its original flavor requires the use of \emph{random}
%sensing matrices. In real life applications, however, such requirement is
%often hard to satisfy as a result of limited memory resources and/or
%constraints in sampling procedures. Therefore, \emph{Deterministic} sensing
%matrices are often preferred.   Moreover, certain matrices,
%based on fast Fourier and fast Hadamard transforms,
%lead to fast and practical iterative reconstruction
%algorithms where sensing matrices are simply applied implicitly
%and never need to be stored explicitly; the implicit operations often can be carried
%out in order $N \log(N)$ time rather than the naive order $N^2$ time
%typical with random dense matrices. Also, certain deterministic  systems
%\cite{HowardCalderbankSearle} have special structures that
%enable especially fast reconstruction in especially large problems.

Exact expressions for the boundary $\easy(\delta)$ separating 
success from failure
were derived in \cite{Donoho04,DoTa05} assuming the measurement matrix is Gaussian  and the problem
size $N$ is large.  In \cite{DoTa08,DoTa09a} those same expressions were 
experimentally observed to also describe accurately many 
non-Gaussian random measurement 
schemes. Thorough mathematical analysis now fully
supports all these findings across very large classes of random matrices \cite{Amelunxen14,Oymak15}.

The same boundary $\easy(\delta)$ even
applies to an important class of non-random measurement 
matrices, see \cite{MoJaGaDo2013}, 
so the `universality' of the compressed sensing phase transition is quite broad.

\newcommand{\aus}{\mbox{\sc aus}}
\newcommand{\ius}{\mbox{\sc ius}}
\newcommand{\us}{\mbox{\sc us}}
\newcommand{\fs}{\mbox{\sc fs}}
\newcommand{\xus}{\hat{x}_{\us}}
\newcommand{\xfs}{\hat{x}_{\fs}}
\newcommand{\xaus}{\hat{x}_{\aus}}
\newcommand{\xius}{\hat{x}_{\ius}}

As an important example, consider Fourier undersampling
in a stylized model of $2$-dimensional imaging. 
The underlying object is a two-dimensional array $x_0 = (x_0({t_0,t_1}),\ 0 \leq t_i < M)$,
whose two-dimensional discrete Fourier transform $\hat{x}_0 = \cF_{2}(x_0)$
is also an $M$ by $M$ array $(\hat{x}_0({k_0,k_1}),\ 0 \leq k_i < M)$.
The traditional experiment gathers the {\it fully-sampled} array  
$\hat{x} = (\hat{x}_0({k_0,k_1}),\ 0 \leq k_i < M)$, 
systematically evaluating the 2D Fourier transform
at each distinct 2D frequency index $(k_0,k_1)$ in the range $0 \leq k_i < M$.   

A {\it randomly-undersampled $k$-space experiment} first 
selects $n$ distinct pairs $(k_{0,i},k_{1,i})$ 
uniformly at random from among all such pairs, and then evaluates 
the Fourier transform just at those $n < N=M^2$ points. Letting
$\cK_2 = \{ (k_{0,i},k_{1,i}), i=1,\dots, n\}$ denote the list of 
sampled $k$-space pairs,
the undersampled Fourier transform operator 
$\cF_{\us} \equiv \cF_{\us}(\cdot; \cK_2)$
produces as output $\cF_{\us}(x_0) = (\hat{x}_0({k_{0,i},k_{1,i}}),\ i=1,\dots, n)$.
%Assuming the object $x_0$ is sparse, 
Mimicking the Gaussian measurements case,
one attempts to reconstruct by $\ell_1$ minimization:
\[
  (P_1^{\us}) \qquad \arg\min_{x} \| x \|_1 \quad \mbox{subject to} \quad 
  \cF_{\us}(x) = \cF_{\us}(x_0).
\]
Depending on the details of the sampling schedule $(k_{0,i},k_{1,i})$
and the sparsity level in $x_0$, this strategy might be successful or unsuccessful.
The random undersampling 
situation has been studied  carefully empirically, and a 
phase transition from success to failure for this sampling scheme
has been observed  \cite{DoTa09,MoJaGaDo2013}, and shown
to agree with the phase transition
curve for Gaussian measurements. So we observe
another instance of the `universality' of Figure \ref{figure-phase-transition} .
However, clearly not {\it every} measurement scheme 
can behave equivalently to Gaussian measurements,

\begin{figure}[t]
\centering
\includegraphics[width=3in]{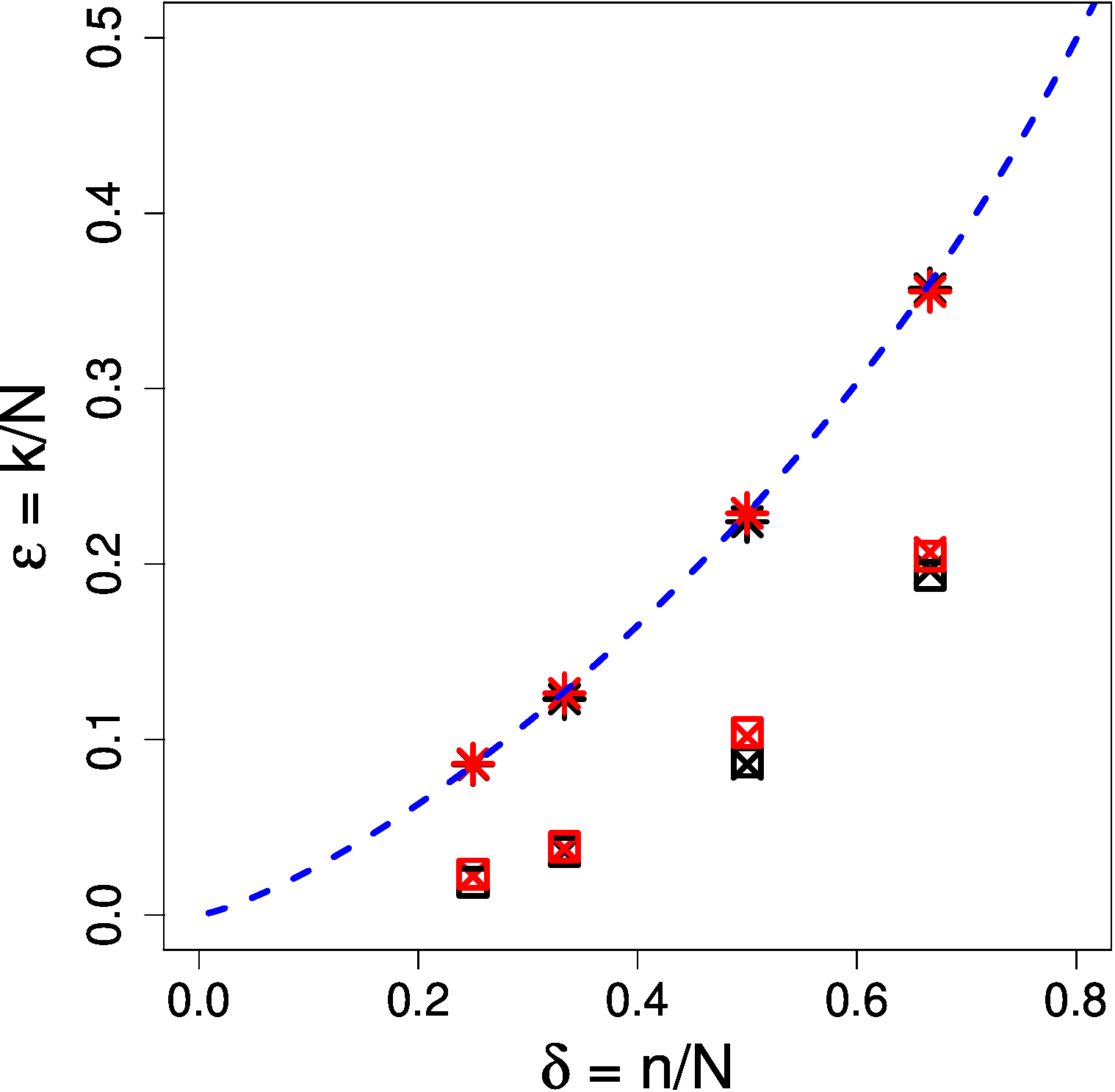}
\caption{Observed Finite-$N$ phase transitions
for anisotropic undersampling and for 
block diagonal measurements studied in this paper. 
The $\times$ symbol indicates finite-$N$ phase transition 
for anisotropic undersampling experiments using partial 2D FT 
where a specific fraction $\delta$ 
of rows are selected uniformly at random,  
and then each selected row is sampled exhaustively. 
The $*$ symbol indicates isotropic undersampling experiments 
using a partial 2D FT 
where a certain fraction of $k$-space samples are selected uniformly at random.
The $\Box$ symbol indicates experimental data from a block 
diagonal measurement matrix with a single repeated Gaussian random matrix block. 
The colors indicate  different problem sizes. Black stands for a $24 \times 24$ grid and {red} for  
a $48\times48$ grid. The dashed blue line gives the asymptotic phase transition location 
for complex-valued Gaussian measurement ensembles. The isotropic undersampling
data lie close to the dashed blue line, while the anisotropic undersampling
data are substantially displaced.}
\label{figure-2d-fft-vs-tiuse}
\end{figure}

\begin{figure}[t]
\centering
\includegraphics[width=6in]{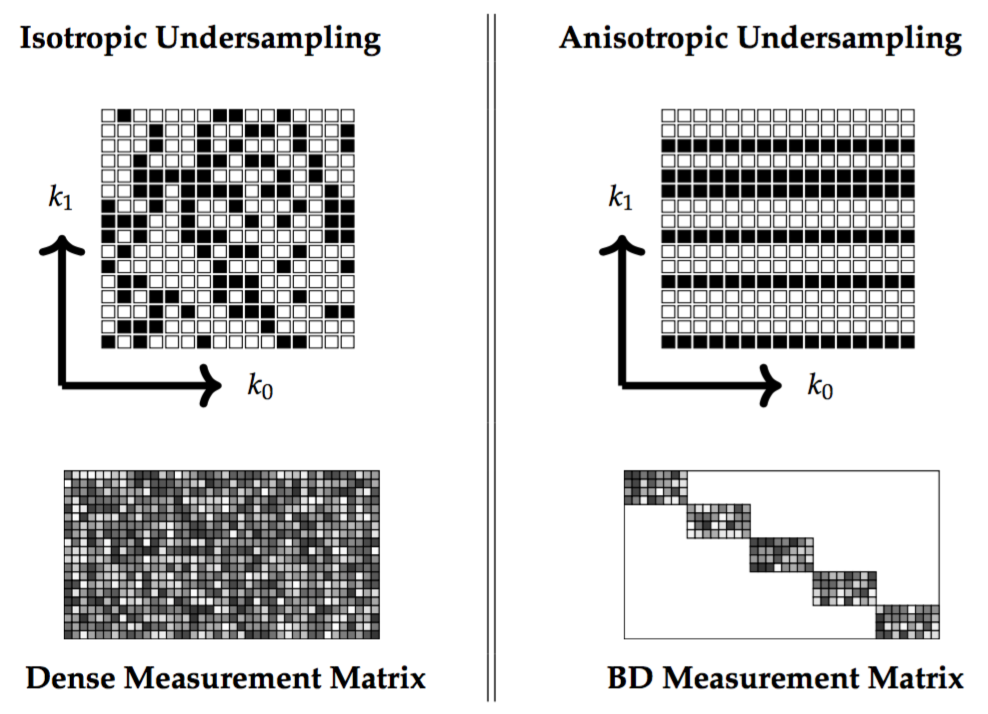}
\caption{Emblems of isotropic (upper left) and anisotropic (upper right) $k$-space 
undersampling discussed in this paper. 
Traditional compressed sensing literature studies isotropic $k$-space undersampling, 
whereas anisotropic undersampling is similar to practical schemes 
commonly used in MR imaging/spectroscopy.
Lower left panel depicts a dense measurement matrix 
of the type produced by isotropic $k$-space undersampling,
while lower right panel depicts a block diagonal measurement matrix associated 
to anisotropic $k$-space undersampling. 
The $k$-space sampling plans of the top row
are equivalent under an appropriate isometry to
undersampling matrices depicted in  the bottom row.
The figure serves merely as a `cartoon illustration' of the 
concept of equivalence established in this paper 
and so the size and number of blocks should not be taken literally.}
\label{figure-equivalance}
\end{figure}

\subsection{Anisotropic Undersampling}

This paper studies an important class of anisotropic undersampling schemes 
that exhibit novel theoretical behavior and arise naturally in 
MR imaging and NMR spectroscopy.
%Anisotropic undersampling exhibits new phase transition phenomena which we will characterize precisely. 
%Examples of sampling schemes outside the universality class of Gaussian ensembles are tensorial measurements such as Random Under-sampling in the Indirect Dimension (RUID), and Random Phase Detection (RPD)
%in MR imaging and spectroscopy \cite{SparseMRI, schmieder93, Maciejewski11}.

% In addition, the measured quantity at each time pixel is a hypercomplex number. The traditional non-uniform sampling (NUS) schemes in NMR spectroscopy measure all the components of the hypercomplex number at each sampled time pixel, whereas more recent undersampling schemes such as Random Phase Detection \cite{Maciejewski11} and its generalization Partial-component sampling \cite{Schuyler13} acquire only a random low-dimensional subset of the high-dimesnional hypercomplex number.

Let's give a concrete example of anisotropic undersampling.
As earlier, the underlying object is a 2D $M \times M$ array $x_0$
with $\hat{x}_0$ a full $M \times M$ array  
%$x_0 = (x_0({t_0,t_1}),\ 0 \leq t_i < M)$.
%There is a 2D array  
%$(\hat{x}_0({k_0,k_1}): 0 \leq k_i < M)$ 
of potential Fourier measurements.
We undersample anisotropically by randomly selecting
$m < M$ {rows} 
%columns 
of the array,
and then sampling {\it everything} within each selected {row},
%column, 
producing $n = m \cdot M$ samples overall.
%where $m = \lfloor \delta M \rfloor$. 
We implement this  concretely by sampling  uniformly at random
$m$ distinct integers $k_{1,i}$ from the set 
$0 \leq  k_1 < M$, forming a list $\cK_1$ 
of $m$ {row} indices (See Figure \ref{figure-equivalance}). The operator 
$\cF_{\aus} = \cF_{\aus}(\cdot; \cK_1)$
yields the partial measurements {
$\cF_{\aus}(x_0)  \equiv (\hat{x}_0(k_{0},k_{1,i}),\ i=1,\dots, m,\ 0 \leq  k_0 < M)$.} 
No other samples are collected. 
The subscript $\aus$ reminds us of the \underline{a}nisotropic 
\underline{u}nder\underline{s}ampling\footnote{
Under the notation we are using, anisotropic undersampling
could also be represented using the general undersampling operator \--
$\cF_{\aus}(\cdot ; \cK_1)  \equiv \cF_{\us}(\cdot ; \{0,\dots,M-1\} \times \cK_1)$ \-- however it simplifies discussion to have a dedicated notation.
}.
Again assuming sparsity of the object to be recovered,
attempt to reconstruct  using  $\ell_1$ minimization:
%\sidecomment[R1-]{"add superscript "aniso" on y". We used different notation}
\[
  (P_1^{\aus}) \qquad \argmin_{x} \| x \|_1 \quad \mbox{subject to} \quad  \cF_{\aus}(x) = \cF_{\aus}(x_0).
\]
Once again, one can observe experimentally that
the sparsity level determines success or failure.

Figure \ref{figure-2d-fft-vs-tiuse} shows results from an 
empirical study of the sparsity-undersampling tradeoff
for anisotropic undersampling.
It displays the { location of the } 
empirical phase transition from
success to failure for our reconstruction from anisotropic undersampling, 
as a function of underlying sparsity { fraction $\eps$}; 
when the object's { sparsity fraction $\eps=k/N$ falls below} 
the depicted transition point, success is the predicted outcome, 
whereas when the object { sparsity fraction $\eps$ exceeds that level,}
we predict failure. 

Let's call the earlier random-$k$-space
undersampling scheme {\it isotropic}.
The terminology reminds us
that for isotropic sampling, $k_0$ and $k_1$ 
are scattered randomly with no directional preference, while for  
anisotropic sampling, { the sampling scheme is 
exhaustive in the $k_0$ coordinate and random in $k_1$ as depicted schematically in Figure \ref{figure-equivalance}. Empirical phase transitions
for the isotropic sampling scheme are also shown in Figure \ref{figure-2d-fft-vs-tiuse}.} 

The striking comparison is that, while there are definite phase transitions in each case,
the anisotropic ones don't occur at the same place as the isotropic
ones; instead the transitions for the anisotropic sampling 
are shifted downwards substantially from the phase transitions
for the isotropic scheme.  
Formulas for the precise amount of shift 
are presented below.\footnote{It will still be the case that at sufficiently large $N$,
the shift goes away; however, it will become clear that the required
$N$ are unreasonably large, so such schemes in practice will always exhibit
a noticeable shift, by an amount we here quantify  precisely.}

Figure \ref{figure-2d-fft-vs-tiuse} also shows a curve giving 
the  location of the phase transition in the Gaussian case with comparable
sparsity and $n$ and $N$. This curve goes quite near the
empirical phase transitions for isotropic sampling,
confirming the results of \cite{DoTa09,MoJaGaDo2013}.
Consequently, we can also say that the anisotropic undersampling
results differ substantially from the Gaussian undersampling results.

\subsection{Block Diagonal Undersampling}

Figure \ref{figure-2d-fft-vs-tiuse}
actually displays empirical results for {\it three} seemingly 
very different situations. The first two were mentioned:
isotropic and anisotropic undersampling in 2D Fourier imaging, respectively.

The third situation is seemingly unrelated to 2D Fourier imaging: 
{\it block-diagonal Gaussian undersampling}. In that setting,
the object is a 1D vector of length $M^2$, partitioned into
$M$ blocks of size $M$ each.  For a given undersampling parameter 
$m  < M$,
the measurement matrix $A$ is
$n = m \cdot M $ by $N = M \cdot M$, and has 
a block-diagonal form made of $M$ blocks, each of size $m$ by $M$.
The off-diagonal blocks are all zero, and the diagonal blocks 
are random, filled with  i.i.d.\ Gaussian entries. The object to
recover $\bx_0$ is a sparse vector of length $N=M^2$. 
The measurements are $\by= A \bx_0$.

Figure \ref{figure-2d-fft-vs-tiuse}  shows the
the finite-$N$ phase transitions of block-diagonal undersampling
schemes, but they are hard to discern; 
the locations are visually quite close
to those of anisotropic undersampling, given the same values for the underlying 
undersampling fraction $\delta=n/N$ and sparsity $k/N$.
%This is not surprising if we consider the independent operations of 2D partial Fourier transform 
%on columns and rows of a two dimensional object. 

%have slower convergence rates or even different asymptotic phase transitions location. In these cases, one must be cautious when using asymptotic results for finite-$N$ studies as they may no longer be applicable. As an example, suppose for a certain sampling procedure the rate of convergence to asymptotic phase transition follows a $1/\sqrt{N}$ law. This means that for a specified level of accuracy, one needs a much larger problem size, say $N = 2^{16}$, compared to $N = 2^8$ that was needed in the case of sampling with Gaussian matrices.

%Randomness of the sensing matrix is, in general, not required for the success of CS
%theory. Indeed, extensive empirical evidence reported in \cite{MoJaGaDo2013, DCS-PURL} show that
%the CS framework can be applied to numerous reasonable deterministic matrices arising in engineering applications, as
%long as the object to be recovered is \emph{typical}\footnote{In a typical object, the positions of the
%non zeros are chosen purely at random}. Also, \cite{MoJaGaDo2013} shows that the phase transition of
%certain deterministic matrices matches that of the Gaussian ensembles.

\subsection{This paper's contribution}

This paper shows that the observed equality of phase transition 
between anisotropic undersampling and block-diagonal undersampling
is no coincidence. { It demonstrates 
the formal theoretical equivalence of
anisotropic  undersampling  with appropriate block-diagonal undersampling. 
It then exploits this equivalence, by  deriving precise formulas for the 
finite-$N$ phase transitions of sparse reconstruction 
from block-diagonal undersampling. These formulas are rigorously proven in one class of situations,
but the formulas arising from our study are shown empirically to accurately predict phase transitions
observed in anisotropic undersampling 
in Figure \ref{figure-2d-fft-vs-tiuse}
and several other situations.
}  

We have so far presented just the example
of anisotropic undersampling in 2D Fourier imaging, but this is only a special case of our general
results, which apply to fully  general anisotropic undersampling of $d$-dimensional Fourier imaging,
in which some dimensions are sampled uniformly at random and others are sampled
exhaustively. In dimension $d=2$ there is only one type of anisotropic sampling;
but in higher dimensions one can have $d_r$ dimensions sampled randomly and $d_e$ dimensions
sampled exhaustively, covering  all $d = d_r+d_e$ dimensions. The case described so far is simply 
$d_r=1$, $d_e=1$, $d=2$; but our results and methods are far more general.

An important conclusion from our study --
see Corollary \ref{Lemma-PT-MultiBlock-General} -- 
will be that, for a given number $n$ of observations,
the best sparsity-undersampling tradeoffs are obtained when $d_e$ is as small as possible
and $d_r$ is as large as possible, in a way that we can quantify precisely.
While that would seem to suggest always using $d_e=0$ and hence using isotropic undersampling,
in certain applications, randomness can only be implemented in a subset of the dimensions. 

\subsection{Application Areas}

Our results have stylized applications to two important practical fields: 
MR imaging and NMR spectroscopy. In either setting,
 the experiment produces a sequence of {\it free induction decays} (FIDs);  these are  individual time series  output
 by radio-frequency  receivers. They are variously called  repetitions, interleaves, or  phase-encodes in MRI.  In such a sequence the acquired data may be labeled  $(k_0; k_1,\dots,k_{d-1})$,
 with $k_0$ indexing the time samples of the FID and $(k_1,\dots,k_{d-1})$ 
 indexing the FID itself. Under complete acquisition,
 we would acquire a complete collection of FID's, 
 and thereby obtain a complete Cartesian sampling
 spanning $0 \leq k_i <  M$, $i=0,\dots , d-1$, 
 while under anisotropic undersampling,
 we would acquire only a subset of FID's.  In more detail:
\bitem
\item {\it Multi-dimensional MR imaging.} 
Ordinary MR imaging, producing a single $2D$ image,
 is effectively a case of $d=2$-dimensional Fourier imaging. 
 Higher-dimensional MR imaging can be 
 either $3D$ $(x,y,z)$ or
 dynamic $(t,x,y)$ or $3D$ dynamic  $(t,x,y,z)$. 
 
 In MR imaging experiments, each FID 
 $( \hat{x}({k_0,\dots,k_{d-1}}),\ 0 \leq k_0 < M)$ { can be viewed as
 the sequence of samples of the traditional 
 complex-valued $d$-dimensional Fourier transform along an 
 axis-oriented line $ u \mapsto (u,k_1,\dots, k_{d-1})$  for $u=0 , 1, \dots M-1$ }
 in the $d$-dimensional data hypercube.  
 
Anisotropic undersampling has been used in 
some way in MRI for many years 
in some cases for the purpose of accelerating image acquisitions \cite{SparseMRI,mcgibney1993quantitative,marseille1996nonuniform,sodickson1997simultaneous,pruessmann1999sense,madore1999unaliasing,greiser2003efficient}.  It has been called `random undersampling in the indirect dimensions’ or { `random sampling in the phase-encodes'}; see the article by Michael Lustig et al.\ \cite{SparseMRI}.

 %\cite{LustigSparseMRI2007}.
 
 \item {\it Multi-dimensional NMR spectroscopy.}
NMR spectroscopy experiments are more abstract and flexible 
than MR imaging, and can in principle  be designed to encompass arbitrary-dimensional experiments; 
however, high-dimensional experiments
take longer than low-dimensional ones, and practical limitations  can intervene: these include
denaturing of the sample material, and lack of exclusive
access to a spectrometer for the days or weeks that might be required.  
In practice, experiments at  higher dimensions
than $3$ are rarely attempted. 

The FID in a spectroscopy experiment also can be viewed as the sequence of 
samples along a line in a  $d$-dimensional data hypercube of `Fourier' coefficients.
However, the notion of Fourier transform differs in spectroscopy 
because each coefficient is hypercomplex-valued, 
so each sampled value $\hat{x}(k_0,\dots,k_{d-1})$ is  $2^d$-dimensional;
for an explanation of this point see \cite{Monajemi2016ACHA,Monajemi_thesis_2016}.

In NMR spectroscopy, anisotropic undersampling has been applied for
decades by Jeffrey Hoch and collaborators; see \cite{schmieder93,schmieder1994improved,mobli2012sparse,hoch2014nonuniform}.
%In the NMR spectroscopy literature, one starts to see discussions of the freedom that exists to
%sample differently in some dimensions than others; see for example \cite{Monajemi2016ACHA} and the articles of Schuyler et al.~\cite{Schuyler13,Schuyler15}.
\eitem

Mathematical scientists who study compressed sensing often mention MR imaging or NMR spectroscopy as
applied settings where undersampled Fourier imaging is indeed applied successfully today; they rarely if ever mention
that in either applied setting, the sampling is {\it always} anisotropic; it {\it never} makes sense to sample
isotropically, because one {\it always} gets exhaustive samples along { one of the dimensions (a.k.a direct dimension)} inherently as part of 
the physical experiment; it makes no sense to throw away measurements  that were already mandatorily taken.
Mathematical scientists often speak as if isotropic undersampling were an option in these applied settings,
and reference theories involving isotropic undersampling. However, in either setting, isotropic sampling is not a sensible option,
and  the referenced theories do not offer accurate predictions of what happens in real experiments.

In contrast, our results describe  anisotropic sampling of the type actually used in these applied fields
and give accurate predictions of 
the sparsity/undersampling relation
in undersampled imaging/spectroscopy.

\subsection{Relation to Previous Work on Block Diagonal Undersampling}
This paper considers precise finite-N phase transition properties
of a { setting seemingly unrelated to block-diagonal undersampling:} 
anisotropic undersampling in $d$-dimensional (hypercomplex-) Fourier imaging. 
It identifies block-diagonal measurement systems as an analysis tool that allows
us to make accurate predictions of the behavior of anisotropic undersampling.

Actually, block-diagonal undersampling has been discussed previously as
an approach to compressed sensing of interest in its own right. 
For example, block-diagonal undersampling has been used for compressed image acquisition 
in \cite{lu07,Fowler12} and was studied as part of a more general category of compressed sensing, 
namely tensor compressed sensing, in \cite{Friedland14, Li13}. 

%\comment[HM]{Below, we have added citation to Chun-Adcock 2016 paper as suggested by Reviewer\#1 and modified the statement.

% OLD TEXT:\\
% Also, Eftekhari et al. [16] study the Restricted Isometry Property (RIP) for block-diagonal matrices. Their study suggests that block-diagonal
% undersampling is qualitatively similar to dense Gaussian undersampling for compressed sensing.
% In contrast, we consider precise finite-N phase transition properties and show that anisotropic undersampling has
% finite-N phase transitions different from the phase transitions for Gaussian undersampling. Our extensive computational
% results document the accuracy of our finite-N prediction formulas and show substantially worse phase transitions
% for anisotropic undersampling than for Gaussian undersampling. We do find that, asymptotically as N grows large, the
% finite-N phase transitions of anisotropic undersampling schemes converge to the asymptotic phase transition of Gaussian
% undersampling; however, this convergence occurs much more slowly than under isotropic sampling.3 Finally, this
% paper shows that certain instances of block-diagonal undersampling are equivalent to anisotropic undersampling in
% d-dimensional (hypercomplex-) Fourier imaging.}
{ 
There are also interesting  theoretical papers 
on block-diagonal measurement matrices,
the main emphasis has been on the 
Restricted Isometry Property (RIP) for such matrices. 
For example, see work by Eftekhari et al. \cite{Eftekhari2015}
and by Adcock and Chun \cite{Chun16}. 
The RIP offers qualitative insights,
and allows these earlier authors to propose the interpretation that, 
under favorable assumptions,}\footnote{{ These papers assume that } 
signals have sparse representation by incoherent dictionaries such as Fourier or cosine basis}.
block-diagonal matrices {\it  asymptotically} ``perform nearly as well as dense Gaussian random matrices'' for compressed sensing. 

Our interpretation of the results we obtain in this paper is quite different.
Motivated by anisotropic undersampling in practical $d$-dimensional (hypercomplex-) 
Fourier imaging, we study the finite-$N$ 
phase transitions to learn about the precise sparsity level needed for exact recovery.
We then show { rigorously that for a special analytically tractable set of situations,}
the Finite-$N$ phase transition is conspicuously different from the phase transitions for dense Gaussian undersampling.
{  Moreover, we derive formulas that predict accurately even 
outside  cases where we can do rigorous mathematical analysis.
By extensive computations we} document the accuracy of our finite-$N$ prediction formulas
and thereby show quite generally that there are substantially 
worse finite-$N$ phase transitions for anisotropic/block-diagonal undersampling 
than for Gaussian undersampling.  We do find that, asymptotically as $N$ grows {\it  very} large,
the finite-$N$ phase transitions of anisotropic undersampling schemes converge 
to the asymptotic phase transition of Gaussian undersampling; however, this convergence occurs much more slowly  than
under isotropic sampling.\footnote{As an example, for an $M\times M$ grid where one of the 
  dimensions is measured exhaustively, the rate of convergence is $M^{-1/2}$, 
  while for { isotropic random sampling} in both dimensions, this rate would be $M^{-2}$.}   
\newcommand{\bH}{{\bf H}}

\section{Finite-N Phase Transitions for Block Diagonal Measurements}

In this section, we discuss the problem of recovering 
a sparse $N$-vector $\bx_0$ from $n$ measurements
 $\by = A \bx_0$. Here the $n \times N$ measurement matrix $A$ 
has a special block structure and we use a particular
 convex optimization in our attempt to recover $\bx_0$.
 In one special case, we derive the exact
 finite-$N$ phase transition properties and show that 
 block-diagonal measurement matrices underperform dense
 \emph{i.i.d} Gaussian matrices by a substantial amount.
 In later sections, the {\it ansatz} provided by the explicit formulas
derived in this special case are generalized to successfully predict experimental
results across all other cases, with similar conclusions.

\newcommand{\rk}{{ k}}
\subsection{The convex optimization problem}
The data vector $\by  $ is assumed to arise by applying the
measurement matrix $A$ to the unknown object $\bx_0$.
To reconstruct $\bx_0$, we solve the following convex optimization problem:
$$
(P_{1,\bX})\qquad  \min \  \| \bx \|_{1,\bX}  \quad  \text{subject to} \quad A \bx = \by, \qquad \bx \in \bX^N.
$$
Here each coefficient $x(i)$ is supposed to belong
to a convex subset $\bX \subset \bR^\rk$ with nonempty interior, and 
$\| \cdot \|_{1,\bX}$ denotes the (appropriately defined) ``$\ell_1$ norm'' on $(\bR^d)^N$. 
The coefficient set $\bX$ might for example be $[0,1]$, $[0,\infty)$, or $\bR$,
in which case $d=1$ and  $\| \cdot \|_{1,\bX}$ denotes the usual $\ell_1$-norm on $\bR^N$. But we could also have
$\bX$ be the set of complex numbers $\bC$, or the hypercomplex set $\bH^{d}$ \cite{Monajemi2016ACHA,Monajemi_thesis_2016}.
{ When $\bX$ is $\bC$ (resp. $\bH_d$), 
the ambient dimension $\rk$ is $2$ (resp. $2^d$),}
and $\| \cdot \|_{1,\bX}$ denotes what is more usually called the
mixed $\ell_{2,1}$ norm: $\| x \|_{2,1} = \sum_{i=1}^N \| x(i) \|_{\ell_2(\bR^\rk)}$ (resp. hypercomplex one-norm $\| x \|_{\bH,1}$ \cite{Monajemi2016ACHA}).

Below, we often write $(P_1)$ rather than $(P_{1,\bX})$, making the coefficient domain explicit only where necessary.

\subsection{Block-Diagonal Measurement Matrices}
In this section, the measurement matrices $A$ will always be in block form:
\[
A=
\left[
\begin{array}{c|c|c|c|c}
A^{(1)}& 0 & 0 & \dots & 0 \\
\hline
0 & A^{(2)} & 0 & \dots & 0 \\
\hline
\multicolumn{5}{c}{\dots } \\
\hline
0& \dots &  0 &   A^{(B-1)}  & 0 \\
\hline
0& \dots &  0 & 0 &  A^{(B)} \\
\end{array}
\right],
\]
where each block $A^{(b)}$ is $m \times M$.   The $0$'s here also denote $m \times M$ blocks,
filled with zero entries. The whole matrix is of 
size $n =m  B$ by  $N =  M B$, and only the blocks on the diagonal can be nonzero.

We can construct such block diagonal matrices in two ways:
\begin{itemize}
\item \textit{Repeated-Block Ensembles (RB)}:  Our blocks are simply $B$ identical copies of the same 
 $m \times M$ block.
%Note that this setup corresponds to a 2D lattice with $N$ rows and $B$ columns in the case of partial row sampling of 2D Fourier
%transform.
\item \textit{Distinct-Block Ensembles (DB)}: 
There are $B$ distinct blocks of size
 $m \times M$. 
 \eitem
 %In either case, the blocks themselves might be either random or deterministic.
 
To obtain the individual blocks, we often consider drawing them at random. 
A standard construction involves Gaussian \emph{i.i.d} entries $A_{i,j} \sim N(0,\frac{1}{m})$.
We often start with such a matrix but then  normalize its
columns to unit length; formalizing this:

\begin{dfn}
A matrix $A$ is said to (have columns sampled from, be sampled from) the 
{\em Uniform Spherical Ensemble} ({\bf USE})
if its  columns  $\ba_i$ are sampled \emph{i.i.d} from the 
uniform distribution on the unit sphere $S^{m-1} \subset \bX^m$,
where $\bX$ is either $\bR$ or $\bC$. 
\end{dfn}

Using random blocks from USE, the Distinct/Repeated blocks distinction gives us
 two kinds of matrix ensembles:
\bitem
\item {\it Repeated-Block USE (RBUSE)}. We draw a {\it single} block  $A^{(1)}$ from  USE.
We generate a  block-diagonal matrix 
$A = \text{diag}(A^{(1)}, \dots, A^{(1)})$ having all $B$ blocks be identical copies of $A^{(1)}$.
Equivalently, the full measurement matrix $A$ is a  \emph{Kronecker} product: 
$A = I_B \otimes A^{(1)}$.
\item {\it Distinct-Block USE (DBUSE)}.
 We draw $B$ independently-sampled $m \times M$ blocks $A^{(b)}$, $ b=1,\dots, B$, from USE. 
We generate a  block-diagonal matrix $A = \text{diag}(A^{(1)}, \dots, A^{(B)})$.
This may equivalently be written as
a direct sum: $A = {\bigoplus}_{b = 1}^{B}A^{(b)}$.
\end{itemize}

%Here, we consider only the case of underdetermined system of equations where $n<N$. 
%$$
%G_{Bn, BN} =    \left[\begin{array}{cccc}
%   A_{m.M}^1& & &\\
%   		   &A_{m.M}^2& & \\
%    			& &\ddots &\\
%    			& & &A_{m.M}^B\\
%   \end{array}\right],
%$$
%$$
%x =    \left[\begin{array}{c}
%	   x_1\\
%   	   x_2 \\
%    	   \vdots\\
%    	   x_B\\
% 	\end{array}\right],
% \qquad
% y =    \left[\begin{array}{c}
%   	y_1\\
%   	y_2 \\
%    	\vdots\\
%    	y_B\\
%   	\end{array}\right],
%$$

\subsection{Separability}

%Problem  $(P_1)$ is \emph{separable}  sense that the unique solution $x^*$, when exists, can be found by solving for individual blocks independently, i.e.,  
\newcommand{\cX}{{\cal X}}
Since our measurement matrix $A$ has the  block-diagonal form  
 $A = \text{diag}(A^{(1)}, A^{(2)}, \cdots, A^{(B)})$,
it makes sense
%\sidecomment[R2-]{makes since $\goto$ makes sense}
to partition the vectors $\by$ and $\bx$
 involved in the relation $\by=A\bx$ consistently with block structure of $A$: 
\[
\bx =  [ x^{(1)} \mid x^{(2)} \mid \cdots \mid x^{(B)}], \quad \by =  [ y^{(1)}\mid y^{(2)} \mid \cdots \mid y^{(B)}],
\]
where the subvectors $x^{(b)}$ are $M \times 1$, while the $y^{(b)}$ are $m \times 1$.
The equation $\by = A \bx$ is then precisely equivalent to the $B$ different relations
\[
y^{(b)} = A^{(b)} x^{(b)}, \qquad b = 1, \dots , B.
\]
Define now the $b$-th {\it block subproblem}:
$$
(P_1^{(b)}) \quad  \min \  \| x \|_{1,\bX}  \quad \text{subject to} \quad A^{(b)} x = y^{(b)} , \quad x \in \bX^M.
$$
The key consequence of block-diagonality of $A$ is that the optimization problem $(P_1)$ becomes
separable into its pieces $(P_1^{(b)})$.
\begin{lem} {\bf (Separability of $(P_1)$)}
We have
\[
    val(P_1) = \sum_{b=1}^B val(P_1^{(b)}) .
\]
Let $\cX_1^{(b)} \subset \bX^M$ denote the set of optimal solutions of $(P_1^{(b)})$ and let
$\cX_1 \subset \bX^N$ denote the set of solutions of $(P_1)$. Then
\[
   \cX_1 = \bigoplus_{b=1}^B \cX_1^{(b)}.
\]
In particular, suppose that each subproblem $(P_1^{(b)})$ has a unique solution $x_1^{(b)}$. 
The combined vector $\bx_1 = [x_1^{(1)}\mid x_1^{(2)}| \dots \mid x_1^{(B)}]$ is then the unique solution to $(P_1)$. 
Suppose that $(P_1)$ has a unique solution $\bx_1$. Then the $b$-th block of $x_1$, say $x_1^{(b)}$, %%\sidecomment[R1-][MS]{??}
is the unique solution of 
$(P_1^{(b)})$.

\end{lem}
%\comment[HM]{Missing proof?}
%\noindent {\bf Claim 1:}  {\it Any optimal solution to $(P_1^S)$ is an optimal solution for $(P_1)$, and vice versa.}  \\
%\noindent {\bf Proof:} The KKT optimality conditions of $(P_1)$ and $(P_1^S)$ are identical, hence any optimal solution for one 
%is an optimal solution for the other. The KKT optimality conditions are:
%\begin{eqnarray*}
%A^i {x_*^{(i)}}  = y^{(i)}, & &  {A_I^{(i)}}^T {z_*^{(i)}}  = \text{sign}({x_{*,I}^{(i)}})   \\
%\| {A_{{I^c}}^{(i)}}^T {z_*^{(i)}} \|_\infty \le 1, &&  z_*^{(i)} \ge 0 \\
%\end{eqnarray*}
%Here $z_*^{(i)}$ represents the dual solution for block $i \in \{1,2,\cdots, B\}$, $I := \text{supp}\big(x_*^{(i)}\big) = \{ j \mid x_*^{(i)}(j) \neq 0\}$, 
%and $I^c = \{ j \mid x_*^{(i)}(j) = 0 \} $.
% 
%\noindent {\bf Claim 2.}  {\it The solutions to $(P_1^S)$ and $(P_1)$ are unique and identical if Condition 1 below is satisfied.}  \\
%\noindent {\bf Proof.} 
%Condition 1 is sufficient and necessary for uniqueness of solution to $(P_1^S)$ \cite{zhang12}. Since any optimal solution of  $(P_1)$ is an optimal 
%solution of $(P_1^S)$, and  $(P_1^S)$ admits only one optimal solution, the optimal solution of $(P_1)$ is unique and the same one as $(P_1^S)$. 
%
%\noindent {\bf Condition 1.}
%For an optimal solution to $(P_1^S)$, say $x_*^{(i)}$, let $I := \text{supp}(x_*^{(i)})$. 
%\bitem
%\item Submatrix $A_{I}^{(i)} $ has full column rank, and 
%\item There is a $z_* \in R^n$ obeying ${A_I^{(i)}}^T z_* = \text{sign}(x_{i,I}^*)$, and $\|{A_{I^c}^{(i)}}^T z^*\|_\infty < 1 $   
%\eitem

\begin{cor} {\bf (Product Rule for Success Probabilities) }
Suppose that the block matrices $A^{(b)}$, $b=1,\dots,B$ are sampled \emph{i.i.d} from a common distribution,
and the subvectors $x_0^{(b)}$ are sampled \emph{i.i.d} from a common distribution.
Define the events
\[
      \Omega^{(b)} \equiv \{ (P_1^{(b)}) \mbox{ has an unique solution $x_1^{(b)}$, and $x_1^{(b)} = x_0^{(b)}$} \}  \}
\]
(i.e. $\Omega^{(b)} = \{  \bx_1^{(b)} =  \{ x_0^{(b)} \}  \}$). Correspondingly, let 
\[
      \Omega \equiv \{ (P_1) \mbox{ has an unique solution $\bx_1$, and $\bx_1 = \bx_0$}  \}.
\]
Then
\[
    \Omega = \cap_{b=1}^B \Omega^{(b)} ,
\]
and
\[
   \Pr(\Omega) = \Pr(\Omega^{(1)})^B.
\]
\end{cor}

For clarity we point out that the matrices $A^{(b)}$ are not assumed 
by the Corollary to have any specific properties themselves, e.g. they do not have
to have \emph{i.i.d} elements $A^{(b)}_{ij}$; instead $A^{(i)}$ 
is simply assumed to be stochastically independent of $A^{(j)}$.  

In a sense, this corollary reduces the task of computing the probability of exact recovery
to the task of computing $P(\Omega^{(1)})$. 

\subsection{Exact Finite-$N$ Success Probabilities for $(P_{1,[0,1]})$}

In one very special case, it is possible to evaluate $P(\Omega^{(1)})$
exactly at each $M$ and $\ell$.  We study this case carefully for clues about the general situation.
Consider  the (single-block) convex optimization problem:
\[
(P_{1,[0,1]}^{(1)}) \quad  \min \  \| x \|_1  \quad \text{subject to} \quad A^{(1)} x = A^{(1)} x_0, \quad 0 \leq x(i) \leq 1.
\]
This is an instance of what we earlier called $(P_{1,\bX})$ with the specific coefficient set $\bX = [0,1]$.
In this problem only,  when we say that $x_0$ has {\it at most $\ell$ { non-constrained} elements},
we mean that at most $\ell$ coefficients $x_0(i)$ do not belong to the boundary $\{ 0, 1 \}$ of $\bX = [0,1]$.

To proceed further, we need two notions:
\bitem
\item {\it Exchangeability.}
The random variables $Z_1,\dots Z_M$ are {\it exchangeable} if, for any permutation $P$
%\sidecomment[R1-]{"any P" to "any permutation P"} 
on the set $\{1,\dots,M\}$,
the joint probability distribution of $(Z_{P(1)},\dots, Z_{P(M)})$ is the same as that for $(Z_{1},\dots, Z_{M})$.
\item {\it General position.}
The vectors $\ba_1,\dots \ba_M$ are in {\it general position} in $\bR^m$ 
if no subcollection of at most $m$ vectors is linearly dependent.
\eitem
We now describe two conditions, either of which allows exact evaluation of 
success probabilities.
\begin{description}
\item[$(C_A)$] $A$ is any fixed $m \times M$ matrix  with its $M$ columns
in general position in $\bR^m$. 
$x_0$ is a random $M$-vector  in $[0,1]^M$ 
surely having $\ell$  entries different than $0$ or $1$,  
and the joint distribution of  $(x_0(i): 1\leq i \leq M)$
is exchangeable. 
\item[$(C_x)$] $x_0$ is any fixed vector in $[0,1]^M$ having $\ell$ entries different than $0$ or $1$. 
$A$ is a random $m \times M$ matrix whose
columns $(\ba^{(1)}, \ba^{(2)}, \dots, \ba^{(M)})$ are { almost } surely in general position and
 have an exchangeable joint distribution.
\end{description}

\noindent
\begin{thm} \label{thm-Donoho-Bnd}
\cite{DoTa10a} ({\bf Exact Success probabilities in the Single-Block Problem, $\bX = [0,1]$.})
Assume either of assumptions $(C_A)$, $(C_x)$ for the joint 
distribution of $(A,x_0$). 
Let $\Omega$ denote the 
event that  $(P_{1,[0,1]}^{(1)})$ has a unique solution, and that solution is precisely $x_0$. Then
$\Pr(\Omega)$ depends only on $\ell$, $m$, $M$, and not on any other details
of the joint distribution of $(A,x_0)$.  In fact, $\Pr(\Omega) = Q_{sb}(\ell,m,M; [0,1])$,
where
%\begin{eqnarray*}
\begin{align}
 \nonumber   Q_{sb}(\ell,m,M ; [0,1])  &= 1 - 2^{-(M-\ell-1)} \sum_{j=0}^{M-m-1} { M-\ell-1 \choose j } \\
        &= 1- P_{M-m,M-\ell}, \mbox{ say }.
\label{bnd-formula}        
\end{align}
%\end{eqnarray*}
\end{thm}
 
%This theorem is the foundation of our study as it will provide
%the exact probability of reconstruction for an unknown object $x_0 \in [0,1]^{NB}$ 
%that is acquired by a block diagonal sensing matrix. Furthermore, this probability is independent of 
%the sensing matrix $A$ ranging through an open dense set 
%in the space of $n\times N$ matrices.  This motivates the hypothesis that results derived for random matrices may be equally applicable to  
%deterministic sequences. Such equivalence for dense matrices was first reported in \cite{MoJaGaDo2013}. Here, we see such equivalence holds true for certain block diagonal matrices. See Figure(\ref{figure-2d-fft-vs-tiuse}) for an example. 
%
%\subsubsection{ Analysis of offset, regular objects}
%For the case of regular objects, the success of reconstruction in each block is independent
%of other blocks, and so,

We now remind the reader that in a sequence of $B$ independent  Bernoulli trials with common success probability $q$,
the chance of $B$ consecutive successes is $q^B$. As a result, we can infer general multiblock success probabilities
from single-block ones (under appropriate assumptions).

\begin{lem}  ({\bf Exact Success probabilities in the Multiblock Problem.}) \label{cor:ExactProbMultiBlock}
Consider  a  random instance of the multiblock problem $(P_{1,\bX})$,
%\sidecomment[R1-]{$(P_{1,\bX}^{B}) \goto (P_{1,\bX})$ }
where the individual components $(A^{(b)},x_0^{(b)})$ are \emph{i.i.d} according to a specific distribution $\nu$.

Let $Q_{sb} = Q_{sb}(\ell,m,M; \nu,\bX)$ denote the success probability
for  the single-block problem $(P_{1,\bX}^{(1)})$: namely, let $\Omega^{(1)}$ denote the 
event that  $(P_{1,\bX}^{(1)})$ has a unique solution, and that solution is precisely $x_0^{(1)}$,
and set
\[
   Q_{sb} = \Pr(\Omega^{(1)}).
\]
Let $Q_{mb} = Q_{mb}(B \cdot \ell,B \cdot m, B \cdot M; \nu,\bX)$ denote the success probability
for  the multiblock problem $(P_{1,\bX})$
%\sidecomment[R1-]{$(P_{1,\bX}^{B}) \goto (P_{1,\bX})$ }
: i.e., with $\Omega$ denoting the 
event that  $(P_{1,\bX})$
%\sidecomment[R1-]{$(P_{1,\bX}^{B}) \goto (P_{1,\bX})$ }
has a unique solution, and that solution is precisely $x_0$,
we have
%\sidecomment[R1-]{the notation $\Omega^{(B)}$ would seem to conflict with the meaning of $\Omega^{(b)}$ (when b=B) from the previous page}
\[
   Q_{mb} = \Pr(\Omega).
\]
Then
\[
    Q_{mb} = (Q_{sb})^B.
\]
\end{lem}

Turn now to  the corresponding  multiblock problem
\[
(P_{1,[0,1]}) \quad  \min \  \sum_{b=1}^B \| x^{(b)}  \|_1  \quad \text{subject to} \quad A^{(b)} x^{(b)} = y^{(b)}, \quad 0 \leq x^{(b)}(i) \leq 1, \quad 1 \leq b \leq B.
\]
%\sidecomment[R1-]{$(P_{1,[0,1]}^{B}) \goto (P_{1,[0,1]})$ }

\begin{cor} \label{cor:MBfromSB}
Consider  a  random instance of the multiblock problem $(P_{1,[0,1]})$%\sidecomment[R1-]{$(P_{1,[0,1]}^{B}) \goto (P_{1,[0,1]})$ },  
where the individual components $(A^{(b)},x_0^{(b)})$ are \emph{i.i.d} according to a specific distribution $\nu$
that almost surely obeys $(C_A)$ conditionally on $A$, or almost surely obeys $(C_x)$ conditional on $x_0$.
Then when $N = B \cdot M$, $n = B \cdot m$, and $k = B \cdot \ell$,
\beq \label{eq:MBfromSB}
  Q_{mb}(k,n, N; \nu,[0,1]) = Q_{sb}(\ell,m,M;[0,1])^B .
\eeq
\end{cor}

Note that the RHS  of (\ref{eq:MBfromSB}) does not depend on any further details of the joint distribution
$\nu$. It is in this sense {\it universal}.

\subsection{Finite-$N$ Phase transition Location for $(P_{1,[0,1]})$}
%\sidecomment[R1-]{$(P_{1,[0,1]}^{B}) \goto (P_{1,[0,1]})$ }

Here is how we obtain estimates of the Finite-$N$ phase transition.

\begin{dfn}\label{dfn:q} Assume that we have experimental data for the  frequency 
of successful reconstruction at a fixed undersampling ratio $\delta$ and 
varying sparsity ratio $\epsilon$, Assume that we  fit a generalized linear model
\begin{equation*}
\Pr \{ \textit{Success}|\delta,\eps \}  = \pi(a + b \eps| \delta),
\end{equation*}
to the empirical success frequencies  $\hat{Q}(k,n,N)$, where 
$\pi( \cdot  | \delta)$ denotes the fitted distribution.
\bitem 
\item In the single-block case,  we use the Normal distribution (Probit link).
\item In the multi-block case, we use the Gumbel distribution (CLL link).
\eitem
% where $\pi( y| \delta)$ is some parametric Cumulative distribution function
% with $\pi( 0 | \delta)$.
% We define the probability $q^*$ associated with the phase transition location as
% $$
% q^* = \pi(0 | \delta).
% $$
% For example, if the success probability follows a Gumbel distribution given by
% \begin{equation*}
% \pi(\eps| \delta) = 1 - \exp \left\{ - \exp (a(\delta) + b(\delta) \cdot \epsilon) \right\},
% \end{equation*}
In the multiblock case we define the special constant
$q^* = 1-1/e$ and in the single block case, we set $q^* = 1/2$. 
\end{dfn}

\begin{dfn}
Consider a random instance of an optimization problem $(P)$ with problem sizes $(k,n,N)$,
where $n$ and $N$ are the extent of the matrix $A$ and $k$ is the number of nonzeros in $x_0$.
Let $Q(k,n,N)$ denote the probability of success with given size parameters.
Let $q^*$ be the probability defined in Definition \ref{dfn:q}.
Let $k^*$ denote the smallest integer closest to achieving success probability $q^*$: %$q^* \equiv 1 -1/e \approx 0.632$:
\[
   Q(k^*,n,N)  \approx  q^*.
\]
The Finite-$N$ phase transition location is the ratio
\[
    \epsilon^*(n,N;(P)) = \frac{k^*}{N} .
\]
\end{dfn}

We now apply this concept using the formulas
of the last section, in two ways.
Once, on a `classical' single-block problem, and once on 
a multiblock problem of equivalent size.

\bitem
\item {\it Single-Block Problem}.  
Consider a single-block problem of size $N = B M$,
 $n = B m$, $k = B \ell$, which is equivalent to the problem size
 of a multiblock problem to be considered next. We emphasize that
 this is not the main case for analysis in this section, but we study it for
 comparison purposes. It corresponds to the case $N=M$, $B=1$ in our notation,
 which is not our usual case. Using the preceding Theorem, the critical number
 of nonzeros
 $k^*_{sb}(n,N)$ solves
 \[
     Q_{sb}(k_{sb}^*,n,N) \approx q^*,
 \]
 and  we define the single-block Finite-$N$ phase transition by
 \[
    \epsilon^*_{sb} = \frac{k^*_{sb}(n,N)}{N}.
 \]
 \item {\it Multi-Block Problem.} Again in the multiblock setting $B \gg 1$,
the preceding corollary shows that 
the probability of success   is a function of $\ell,m,M,B$. 
The  critical number of nonzeros $\ell_{mb}^* = \ell_{mb}^*(m,M,B)$, yielding
\begin{eqnarray*}
Q_{sb}(\ell_{mb}^*,m,M)^B &\approx & q^*.
\end{eqnarray*}
Setting  $k_{mb}^* = B \cdot \ell^*$ for the equivalent total number of nonzeros and
 the total problem sizes $n = B m $, $N  =B M$, 
  the phase transition location is
\[
     \emb(m,M,B) =  \frac{k^*_{mb}}{N} =  \frac{B \cdot \ell^*}{B \cdot M} = \frac{\ell_{mb}^*(m,M,B)}{M}.
\]
\eitem

To be more concrete, we need specific assumptions about $m$, $M$, and $B$.

\begin{lem}\label{Lemma-PT-MultiBlock}
Consider a sequence of problem sizes where $B=M$,
$M \goto \infty$,  and 
$m/M \goto \delta \in (1/2,1)$.
%\sidecomment[R1-]{is it necessary to assume $\delta >1/2$ here? That will ensure $2\delta - 1 > 0$}
With $N = B \cdot M$ and $n = B \cdot m$
we have $n \sim \delta N$. Define the asymptotic phase transition 
\[
   \easy(\delta ; [0,1] ) = (2 \delta-1)_+.
\]  
For the single-block finite-$N$ phase transition we have:
\[
   \esb(m,M; [0,1] )  = \easy(\delta)  + O(\frac{1}{M}) .
\]
Define $\gamma_M = \sqrt{\frac{2 \log(M)}{M}}$.
For the multi-block finite-$N$ phase transition we have
\[
   \emb(m,M,B; [0,1] ) =  \easy(\delta) -  \sqrt{2(1-\delta)} \cdot \gamma_M + o(\gamma_M). 
\]
\end{lem}
\begin{proof}
See Appendix \ref{app:PT-MultiBlock}.
\end{proof}

In particular, this lemma shows that as $B = M \goto \infty$ with $m \sim \delta M$,
\[
  \esb( m ,M)  -  \emb(m,M,B)  = \sqrt{2(1-\delta)} \cdot \gamma_M \cdot (1 + o(1)).
\]
Because $\gamma_M \goto 0$ as $M \goto \infty$, 
 this shift in phase transition\sout{s}
locations is asymptotically negligible. However, our experimental 
observations---given above and also below---show it to be quite substantial
in the intended applications. The mismatch between the single-block prediction
and the observed behavior in the multiblock case is quite substantial unless
$M$ (not $N$) is large.  In applications it is much harder to make $M$ large
than to make $N$ large.
Note that in the above lemma the system size is $N = B M = M^2$.  Hence we may equivalently write
\[
  \esb -  \emb  \sim \frac{\sqrt{2(1-\delta) \log(N)}}{N^{1/4}}  , \qquad  N \goto \infty.
\]
The denominator shows that the gap between the two phase transitions closes
very slowly with increasing problem size $N$.

\subsection{Nonidentical Subproblems?}
{
Corollary \ref{cor:MBfromSB} showed us that in case the different 
subproblems $(A^{(b)},x^{(b)})$ are i.i.d. from
a common distribution, simple formulas for the multiblock 
success probability become available.
In applications, as we will discuss later,
the different subproblems might not be identical
in structure. However, the above formulas 
provide ample clues for those cases as well,
as we will discuss further below. 

For example, we can see that, if among the $B$ subproblems,
if there were one `outlier subproblem' with dramatically
higher fraction of nonzeros $\eps^{(b)} = k_b/M$,
then that subproblem would likely be the one whose success or failure
determined the success or failure of the whole reconstructruction.
That subproblem would be in a sense the `weakest link'.

Following down this path, we see that having identical
sparsity fractions and iid matrices $A^{(b)}$ is a
kind of extremal situation; in other situations the
finite-$N$ phase transition is likely to be worse. 
We call this situation the regular situation, and
because we document a sizeable offset below $\easy$
in this situation, one easily sees that other cases will
show even larger effects than documented here in Lemma \ref{Lemma-PT-MultiBlock}.

As an example, consider
a situation where the vector $\bx$ has $\eps = k/N$ nonzeros
at randomly chosen positions. 
In particular the different partitions of the block would
have different numbers of nonzeros, according to the usual multinomial 
distribution. Below we call this situation the multinomial
situation. We have worked 
out the offset of the finite-$N$ phase transition
below $\easy$, and indeed the offset is even larger in the multinomial case
than in the case with equal numbers of nonzeros per block.
We leave detailed discussion of the 
multinomial case for future work.

Below we focus on the regular case, keeping in mind its extremal
nature as the block-diagonal situation somehow closest to 
the fully dense situation.
}
\section{Equivalence with Anisotropic Undersampling}

We now discuss the precise equivalence between 
anisotropic undersampling and block-diagonal undersampling,
considering for now only the case of 2D Fourier imaging.
We wish to recover an unknown object 
$x_0 = (x_0(t_0,t_1):  0 \leq t_i <M)$ with complex-valued entries,
defined on a 2D grid of size $M \times M$.  Our observations
are of the form { $\xaus(k_0,i) = \hat{x}(k_0,k_{1,i})$ for some specific 
choices $ \{ k_{1,i}$, $i=1,m \}$, 
and for each $k_0$ satisfying $0 \leq k_0 < M$.
 Let $\bC^{M \times M}$
denote the collection of arrays $x(t_0,t_1)$ with $0 \leq t_0,t_1 < M$,
while $\bC^{M^2}$ denotes the collection of arrays $\bx = (x(i))_{i=1}^{M^2}$.
}

{ We  think of these measurements $\xaus \in C^{M \times m}$ 
as arising from a linear operator $\cFu$ applied to $x_0$:   
$\xaus = \cFu (x_0)$}. The operator $\cFu$ is representable as a pipeline 
{ $\cFu = \cS_2 \circ \cF_2$} of two linear operators.
The first, $\cF_{2}$, say, is simply the usual 
complex-valued 2D discrete Fourier transform that
maps arrays in $\bC^{M \times M}$ to their 2D DFT's, also in $\bC^{M \times M}$.
The second,  $\cS_{2,M,\cK}$, 
is a selection operator 
that  takes as input an $M \times M$ array, and extracts from it
the $m$ {rows}
%columns  
with indices in $\cK = (k_i: 0 \leq i <  m)$; 
here $0 \leq k_i < M$ and the $k_i$ are all distinct).
Within each selected {row}, it exhaustively samples all $M$ elements.
The composition $\cFu =  \cS_{2,M,\cK} \circ \cF_2$ performs anisotropic sampling in 2D-Fourier imaging.

For comparison, let $A^{(1)}$ denote  an $m \times M$ block 
matrix representing the pipeline of two linear operators.
The first, $\cF_1$, performs the usual {\it one-dimensional}  discrete Fourier transform
of a vector $ v \in \bC^M$ delivering a transformed vector $\hat{v} \in \bC^M$.  
The second, $\cS_{1,\cK}$, takes as input an $M$-vector  ($\hat{v}$, say) and
selects the $m$ entries  $(\hat{v}_{k_i}: 1 \leq i \leq m)$ out of the $M$ entries available, where $\cK = (k_i)_{i=0}^{m-1}$. 
Further, let $A$ denote the block-diagonal matrix made by repeating 
the block matrix $A^{(1)}$ along the diagonal $M$ times. 
Then $A \in \bC^{m M \times M^2}$. 

Let $vec(): \bC^{M \times M} \mapsto \bC^{M^2}$ 
denote the operator of stacking all the rows of a matrix
one by one in one tall vector. Let $\bx_0 = vec(x_0)$.
and let $\by = A \cdot \bx_0$, so that 
$\by \in \bC^{m\cdot M}$.
As $m \cdot M < M^2$,
$\by$ is an undersampling of $\bx_0 \in \bC^{M \times M}$.

The problems of recovering $\bx_0$ from $\by= A \cdot \bx_0$ and
from $\by=\cFu(x_0)$ are not obviously related.  One involves a 2D Fourier
transformation that is then subsampled, the other involves a
stack of separate 1D Fourier transforms.

To connect the two,  we need for
the element indices selected by $\cS_{1,\cK}$ 
in the construction of $A$ to be {\it identical} 
to the {row} indices selected by 
the anisotropic selection operator $\cS_{2,M,\cK}$ 
in the construction of $\cFu$.

\begin{thm}{\bf (Anisotropic undersampling models 2D Fourier Imaging.)} \label{thm:RUID} 
In the construction of $A$ and $\cFu$, suppose the underlying indices $(k_i)_{i=1}^m$ used by
$\cS_{1,\cK}$ in the specification of $A^{(1)}$ are the same as the indices
$(k_{1,i})_{i=1}^m$ used by $\cS_{2,M,\cK}$ in the specification of
$\cFu(\cdot ; \cK)$.
Let $ x_0 $ be an array
in $\bC^{M \times M}$,
and $\bx_0 = vec(x_0)$ the corresponding array in $\bC^{M^2}$.
The following two problems have identical values and isomorphic solution sets:
$$
\mathrm{(P_{1,\bC}^{\aus})} \quad \min \| x \|_{1,\bC} \quad \textrm{subject to} \quad  \cFu(x)= \cFu(x_0), \qquad x \in \bC^{M \times M},
$$
$$
(P_{1,\bC}) \quad \min \| \bx \|_{1,\bC} \quad \textrm{subject to} \quad A \bx=  A \bx_0,  \qquad \bx \in \bC^{M^2}.
$$
Namely, $val(P_{1,\bC}^{\aus}) = val(P_{1,\bC})$, and every solution
of the first problem is converted into a solution of the second problem
by $vec()$.
\end{thm}
%\sidecomment[R1-]{The equivalence is not immediate and must be proved.}
%\newenvironment{proof}{\paragraph{Proof:}}{\hfill$\square$}
\begin{proof}
We give two proofs. Appendix \ref{sec:AltProofRUID} gives a direct proof.
Appendix \ref{app:RUID} gives a much more general result of this kind, 
which is adapted later to prove
further results.
\end{proof}

We point out a very special variant that connects to  earlier results.

%\comment[HM]{The following Cor needs correction. It says 'C', shows [0,1] and positive real}
\begin{cor}{\bf (Anisotropic undersampling in 2D Fourier Imaging, bounded coefficients.)} \label{cor:RUID} 
%\sidecomment[R1-]{positive real $\goto$ bounded}
In the construction of $A$ and $\cFu$, let the underlying indices $k_i$,
$i=1,\dots, m$ in $\cS_1$  be the same as the indices { $k_{1,i}$} used
by $\cS_2$. Let $ x_0 $ be an array
in $[0,1]^{M \times M}$ and 
$\bx_0 =  vec(x_0)$.
The following two problems have identical values and isomorphic solution sets:
$$
\mathrm{(P_{1,[0,1]}^{\aus})} \quad \min \| x \|_{1,\bR} \quad \text{subject to} \quad  \cFu( x)= \cFu(x_0), \qquad x \in {[0,1]}^{M \times M},
$$
$$
(P_{1,[0,1]}) \quad \min \| \bx \|_{1,\bR} \quad \text{subject to} \quad A \bx= A \bx_0, \qquad \bx \in {[0,1]}^{M^2}.
$$
\end{cor}
\begin{lem} ({\bf T. Tao} \cite{tao2003uncertainty})
Suppose that $M$ is prime. Then the $m \times M$ matrix 
$A^{(1)} = \cS_{1} \circ \cF_1$ 
constructed above has its columns in general position in $\bC^m$.
\end{lem}

\newcommand{\bw}{{\bf w}}
\begin{cor}
Let $M$ be prime.
Let $\bw_0 \in [0,1]^N$ be a random vector of length $N=M^2$ with exactly $\ell$ 
entries not equal to $0$ or $1$ in each $M$-block.
Let $\bx_0$ be a random vector created by randomly permuting the entries of $\bw_0$
%\sidecomment[R1-]{"the entries in each" $\goto$ "the entries of $\bw_0$ in each"}
in each $M$-block,
via uniformly-distributed random permutations that are stochastically 
independent from block to block.  

With $A$ the fixed block matrix created above, and $\bx_0$ the random vector described in this Corollary,
the assumptions $(C_A)$ and general position of Corollary \ref{cor:ExactProbMultiBlock} apply. Hence
the probability that the solution $\bx_1$  of the multiblock problem $(P_{1,[0,1]})$ is identical to $\bx_0$
is precisely given by the formula
\[
   \Pr(\{ \bx_0 = \bx_1 \} ) =  Q_{sb}(\ell,m,M; [0,1])^M .
\]
\end{cor}

In consequence, our earlier results for 
block-diagonal undersampling give exact results for success probabilities in anisotropic 
undersampling. Namely, consider $M \times M$ 
images $\ell$ nonzeros thrown down at random
within each column.
Let  $\epsilon_\textit{\aus}(m,M)$ denote the associated 
finite-$N$ phase transition for  exact recovery in
anisotropic undersampling of the object $x_0$
in 2D-Fourier imaging. This is identical to $\eps_{mb}(m,M;M)$. We have:

 \begin{cor} \label{cor:GaussianLemma}
Let $\bx_0$ be the random object constructed in the previous corollary. 
Let $W$ denote  an \emph{i.i.d} Gaussian sensing matrix of size $n \times N$ and
let $\by_0 = W \bx_0$ denote Gaussian undersampled measurements. Define
$$
{\rm ({P}_{1,[0,1]}^W)} \quad \min \| \bx \|_{1,\bR} \quad \text{subject to} \quad W \bx= W \bx_0, \qquad \bx \in [0,1]^{N} .
$$
Let $\epsilon_{W}(n,N)$ denote the associated finite-$N$ phase transition for exact recovery from 
Gaussian undersampling.
            Then, as $M$ increases, 
            the offset between Gaussian and anisotropic undersampling 
phase transitions has the following behavior:
\[
  \epsilon_{W}(mM,M^2) - \epsilon_\textit{\aus}(m,M) = \sqrt{2(1-\delta)} \cdot  \gamma_M + o(\gamma_M) ,
\]
where, as above, $\gamma_M = \sqrt{\frac{2 \log(M)}{M}}$.
\end{cor}

\section{Experimental Approach}

The preceding section precisely locates 
the finite-N phase transition from anisotropic undersampling
in one specific case.  The finite-$N$ phase transition 
was shown theoretically to be displaced downwards 
from the asymptotic Gaussian phase transition 
by a definite amount, which depends on $\delta$ and $M$.

This formula can be generalized to predict behavior
of finite-N phase transitions across a wide range 
of situations,  including general $d$-dimensional anisotropic sampling and
encompassing coefficients that are real, complex and  hypercomplex.
In all these cases,
the formula predicts that the 
phase transition for anisotropic undersampling is substantially displaced
from the phase transition for Gaussian undersampling, by an amount
that matters in practically-important problem sizes.
The scaling of this offset with $M$ and $B$ 
is the same in these cases, and the dependence 
on $\delta$ involves in a very particular way
the underlying coefficient set $\bX$. 

To evaluate the accuracy of these predictions, we developed a framework for massive
empirical simulation, which ultimately involved millions of computational
experiments. Empirical results are more informative for applications
than mathematical proofs would be, as they concern 
behavior in situations of the scale
and type that one might actually encounter, instead of the very large problem sizes 
typically assumed by
asymptotic mathematical analysis, which happen
beyond the reach of modern computers and modern NMR experimentation.
Our computational framework is consistent with the approach
developed in \cite{DoTa09a, MoJaGaDo2013}.

 Though our computational setup allows for an arbitrary number of blocks, 
 in this paper we present results  only for the case of $B = M$,
 which as we have seen corresponds to undersampled 2D-Fourier imaging.
 
\subsection{Predictions of Phase Transition Location}

Our formulas for the finite-$N$ phase transition 
location in block-diagonal undersampling
 will be stated in terms of deviation from
the asymptotic phase transition for Gaussian undersampling.  We first make clear what this means,
and then we state our formulas.

\smallskip

\noindent
{\it Formulas for Gaussian Phase Transition.}  
We extend the discussion of Gaussian undersampling 
from Corollary  \ref{cor:GaussianLemma}, to
cover situations of greater generality. Let  the  $n \times N$ random measurement
matrix  $W$  have \emph{i.i.d} $N(0,1)$ entries \footnote{Exactly what this means can be spelled out more precisely in the case
of quaternionic or hypercomplex entries, although we do not pause to do so here.}. For an object $x_0 \in \bX^N$, 
we obtain $n$  measurements $\by_0 = W \bx_0$. We attempt reconstruction via
\[
(P^W_{1,\bX})\qquad  \min \  \| \bx \|_{1,\bX}  \quad  \text{subject to} \quad W \bx = \by_0, \quad \bx \in \bX^N.
\]
To predict success or failure, we take an asymptotic approach.
Consider a sequence of  problems indexed by $N \goto \infty$ with $n/N \goto \delta \in (0,1)$,
and in each  problem instance let 
 $x_0$ be $k_N$-sparse, where $k_N/N \goto  \epsilon \in (0,1)$.
Let $\bx_1$ denote the solution of $({P}^W_{1,\bX})$ with problem instance $(W,\bx_0)$.
The existing literature on compressed sensing gives formulas for the critical 
sparsity level { $\easy(\delta ; \bX)$ 
such that, as $N \goto \infty$,
\[
  \Pr \{ \bx_1 = \bx_0 \} \goto  \left \{ 
   \begin{array}{ l l }
       1 & \epsilon < \easy(\delta ; \bX) \\
       0 & \epsilon > \easy(\delta ; \bX) 
   \end{array} \right . .
\]
For different choices of $\bX$ one can find such formulas in \cite{DoTa05,AMP,donoho-accurate,Donoho1,Amelunxen14}.
For example we have already used above the formula $\easy(\delta; [0,1]) = (2 \delta -1)_+$.
}
\smallskip

\noindent
{\it Formula for regular sparsity.}
Now return to the block-diagonal undersampling case, where $N=MB$ and 
the measurement matrix $A$ is block-diagonal, made from $B$ different $m \times M$
blocks. We can partition the underlying vector $x_0 \in \bX^N$ into $B$ blocks 
of size $M$ consistent with those of $A$.
We say that $x_0$ has {\it regular sparsity} if it has the same number, $\ell$ say,
of nonzeros in each block. We further assume that $x_0$ is random, with
a block-exchangeable distribution.
In this  setting our  formula states that observed solution to $(P_{1,\bX})$ will exhibit, as a function of 
$(\bX,m,M,B)$\footnote{When $B=M$, we use $\epreg(m,M;\bX)$ for the purpose of brevity.}, a finite-$N$ phase transition $\epreg(m,M,B;{\bX})$. Under the assumption that 
$M \goto \infty$ and $m/M \goto \delta \in (0,1)$, the predicted offset 
of the anisotropic undersampling phase transition 
$\epreg(m,M,B; \bX)$ `below'  the asymptotic 
transition $\easy(m/M)$ obeys
%\comment[HM]{This equation was garbled ($\gamma$ is misspelled) in the original submission. Here Corrected!}
\begin{eqnarray}
\frac{\easy(\delta) - \epreg(m,M,B)}{\easy(\delta)} = \alpha \cdot \eta(\delta)   \cdot \gamma + O(\gamma^2),
\label{eq:PTforTensor}
\end{eqnarray}
where $\delta = n/N$,  $\gamma = \gamma_{M,B}  = \sqrt{2 \log(B)/M}$, 
$\alpha = \alpha_\bX$ is a constant given in Table \ref{table-alpha-beta} below, 
and 
$$ \label{eq:DisplacePT}
\eta(\delta ; \bX) = 
\left\{
{\arraycolsep=1.4pt\def\arraystretch{2.2}
\begin{array}{lll}
{\easy(\delta)}^{-1} (1 - {\easy}(\delta))^{\frac{1}{2}},  & \bX = [0,1] & \delta \in (\frac{1}{2},1]\\
\delta^{-\frac{1}{2}}  (1 - {\easy}(\delta))^{\frac{1}{2}}, & \bX = \bR_+ & \delta \in (0,1] \\
\delta^{-\frac{1}{2}} ,& \bX \in \{\bR,\bC\} & \delta \in (0,1]
\end{array}
}
\right. .
$$

%Here ${\epsilon_{\bX}^*} = {\rho_{\bX}^*} \delta$ denotes the asymptotic phase 
%transition curves for the Gaussian ensembles derived in \cite{DoTa05,DoTa08}. 
%$\alpha_\bX$ in equation (\ref{eq:PTforTensor})  is a constant that depends on the 
%coefficient field $\bX$. Table \ref{table-alpha} summarizes the values of $\alpha(\bX)$. 
%Figure (\ref{figure-fdelta}) depicts function $\eta_\bX(\delta)$ for various coefficient fields. These functions 
%can be calculated numerically by the formulas of phase transition presented in
%the supporting information to \cite{AMP}. As an example for $\bX=\bR_+$, one can compute $\eta_{\bR_+}$ in the following way: For $z = [0, \infty)$ let
%$$
%\delta(z)  = \frac{\phi(z)}{\phi(z) + z \Phi(z)},   
%$$
%and 
%$$
%\rho(z)  = 1 -  \frac{z (1-\Phi(z))}{\phi(z)}  
%$$
%then, 
%$$
%\eta_{\bR_+}\big( \delta(z) \big)  = \sqrt{{\delta(z)}^{-1} -  \rho(z) } = \left(\frac{z}{\phi(z)}  \right)^{1/2}.
%$$
%Here $\Phi(z)$ and $\phi(z)$ denote the standard normal distribution and density respectively.  
%

In the above formulas, $\easy(\delta) = \easy(\delta; \bX)$
denotes the vertical location of the 
asymptotic Gaussian phase transition for 
the indicated coefficient set $\bX$. 
The specific forms of the offset shapes
$\eta$ used here have some precedent\footnote{
{
More specifically, Donoho and Tanner
in \cite{DoTa10} proved finite-$N$ bounds on the probability 
of failure, and their bounds involve a vertical offset in the $(\delta,\eps)$
plane of finite-$N$ iso-probability contours away from
the corresponding large-$N$ phase-transition location.
In the case {$\bX=\bR$}, their offset is proportional to our
offset function  $\eta(\delta; \bR) = \delta^{-1/2}$. 
See Appendix \ref{app:exp-bound} for more details.}} in \cite{DoTa10}.

\noindent
{\it Modeling the second order effect.}
When problem sizes are very small (e.g., $M=B=100$), we go beyond
equation (\ref{eq:PTforTensor}) by
including a second-order term:
\begin{eqnarray}
\frac{\easy(\delta) - \epreg(m,M,B)}{\easy(\delta)} = \alpha \cdot \eta(\delta)   \cdot \gamma + \beta \cdot \zeta(\delta) \cdot \gamma^{2} + o(\gamma^2),
\label{eq:PTforTensor-SO}
\end{eqnarray}
where $\beta = \beta_{\bX}$ is a constant given in Table \ref{table-alpha-beta}, and 
$$ \label{eq:DisplacePT-SO}
\zeta(\delta ; \bX) = 
\left\{
{\arraycolsep=1.4pt\def\arraystretch{2.2}
\begin{array}{lll}
1,  & \bX = [0,1] & \delta \in (\frac{1}{2},1]\\
\eta(\delta;\bX),& \bX \in \{\bR_+,\bR,\bC\}  & \delta \in (0,1] \\
\end{array}
}
\right. .
$$
The additional term, quadratic in $\gamma$, leads to improved accuracy in phase transition
locations, as will be evident from the plots of Section \ref{sec:Verify}.

\begin{figure}[t]
\centering
\includegraphics[height=2.054in]{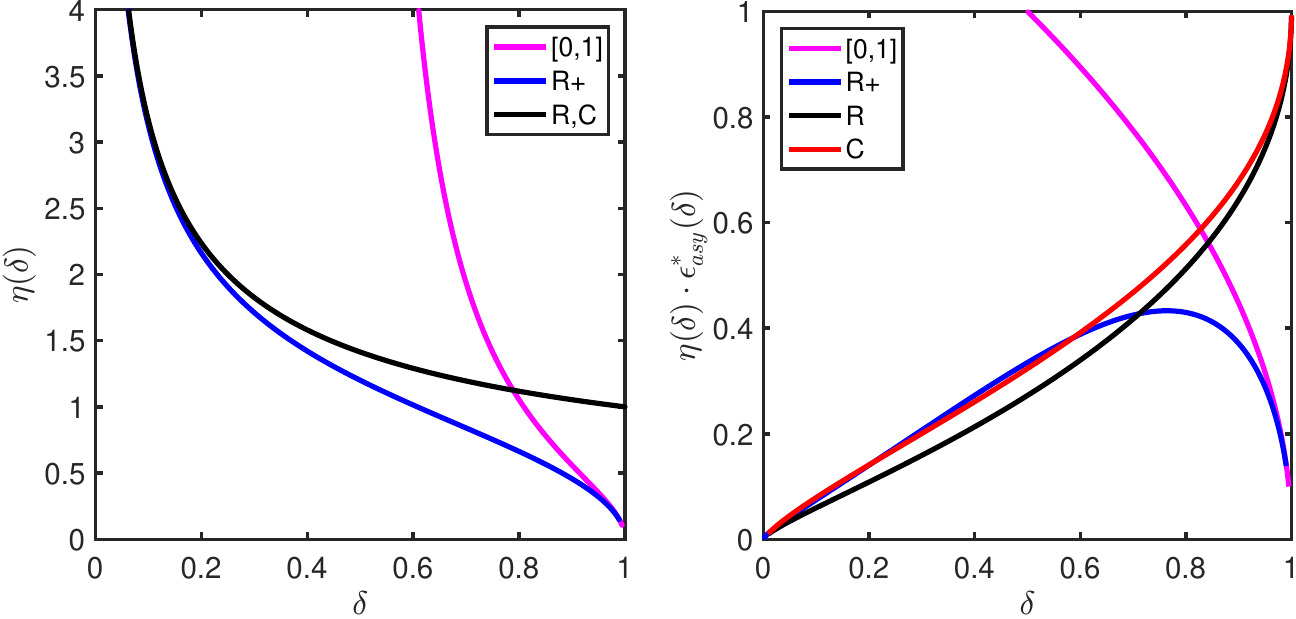}
\caption{Slope functions in the finite-N prediction formula (\ref{eq:PTforTensor}): (left) $\eta(\delta)$, (right) { $\eta(\delta) \cdot  \easy(\delta)$}}
\label{figure-fdelta}
\end{figure}

\begin{table*}[htbp]
\begin{center}
\caption{Values for $\alpha_{\bX}$ and $\beta_{\bX}$ used in (\ref{eq:PTforTensor-SO})}
\label{table-alpha-beta}
%The starred column gives the exact constant from combinatorial geometry;
%others are determined based on empirical evidence.}

\smallskip
\begin{tabular*}{.5\linewidth}{@{\extracolsep{\fill}} c | cccc}
\hline
$\bX$	 	& $[0, 1]$ & $\bR_+$ & $\bR$ & $\bC$ \\
\hline
$\alpha_\bX$	& $1$ & $1$       & $1$& $2/3$ \\ 
%$\alpha(\text{irreg}|\bX)$	& $-$ & $3/2 $       & $3/2$ & $3\sqrt{2}/2$ \\ 
$\beta_{\bX}$	& 1/2 & $-1/3$		& $-1/2$   &  $-1/3$ 
\end{tabular*}
\end{center}
\end{table*}

%\item[+]{\textit{\textbf{Main prediction (irregular objects).}} Consider a block diagonal matrix G with $B$ blocks which are either independent or repeated $n$ by {N} matrices sampled from a ``well-behaved" probability distribution. Further, let $x_0$ be a vector of length $NB$ sampled from the set of $kB$ sparse vectors with coefficients constrained to set $\bX$. The observed solution to $(P_1)$ will exhibit, as a function of $(\delta,\bX, N,B)$, finite-$N$ phase transition behavior that is, to the first order, given by,}
%
%
%$$\epsilon_{\bX}^{ir}  \ge \epsilon_{\bX}^{r} - \frac{\gamma^2}{2} \left( \sqrt{1+\frac{4\epsilon_{\bX}^{r}}{\gamma^2}  } -1 \right) $$
%for $\gamma^{-2} \gg 1$, we can make the following approximation 
%$$\epsilon_{\bX}^{ir} \ge  \epsilon_{\bX}^{r} - \gamma \sqrt{ \epsilon_{\bX}^{r} }$$
%
%
%
%{\bf Derivation:}
%For the irregular case, the nonzero are approximately distributed as $Poi(k,k)$. Therefore, we will have successful recovery with overwhelming probability (i.e., $1-1/N$), if the following condition is satisfied 
%$$\epsilon_{N} + \gamma_N \sqrt{\epsilon_{N}}  \le \epsilon_N^{r}$$ 
%or
%$$\epsilon_{N}  \le  \epsilon_{N}^{r} - \frac{\gamma_N^2}{2} \left( \sqrt{1+\frac{4\epsilon_{N}^{r}}{\gamma_N^2}  } -1 \right) \quad  $$. 
%So, the equivalence break point $\epsilon_N^{ir}$ is lower bounded by the right hand side of the above. i.e, 
%$$\epsilon_{N}^{ir} \ge  \epsilon_{N}^{r} - \gamma_N \sqrt{ \epsilon_{N}^{r} } \quad \text{Q.E.D}$$
%\eitem

\subsection{Experimental Procedure} 
For each quadruple $(k, m, M, B)$, and each relevant coefficient ground set 
$\bX$ %and regularity constraint ${\tt rc} \in \{{\tt Reg}, {\tt Irreg}\}$, 
we run $S$ Monte Carlo trials.  
In each experiment, we generate a pseudo-random $k$-sparse object $\bx_0 \in \bX^{N}$ according to
the regularity constraint {\tt rc}. We take undersampled linear measurements, $\by_0 = A \bx_0$ where the $B$ blocks of matrix $A$
are each of size $m \times M$ and generated according to a certain random or deterministic sequence. 
$(y, A)$ provides an instance of $(P_{1,\bX})$ that we supply to 
a convex optimization solver to obtain solution $\bx_1$. We then compare $\bx_0$ with $\bx_1$. If the relative error $\| \bx_0 - \bx_1\|_2 / \| \bx_0\|_2 < 0.001$,
we declare the reconstruction a success;
otherwise we declare it a failure. 
We thus obtain $S$ binary measurements $Y_i$ indicating success or failure of reconstruction. 
The empirical success probability is then calculated as
$$
\hat{\pi}(k | A, \bX  ) = \frac{\#\text{successes}}{\#\text{trials}} = S^{-1} \sum_{i=1}^{S} Y_i .
$$
Our raw dataset contains these empirical success fractions, at each combination of %$(k, m, M, B,{\tt rc}, S)$
$(k, m, M, B, S)$ we explored.

\subsection{Modeling the Quantal Response Function}
In biological assessment, the quantal response measures the probability of organism 
failure (e.g., death) as a function of drug dose. In the context of compressed sensing, 
the quantal response gives the probability of failure in reconstruction as a function of 
the `complexity dose', i.e. the number of nonzeros in the vector $\bx_0$.
This of course is measured by sparsity ratio $\epsilon$. 
It is shown in \cite{DoTa09a} that a Probit model adequately describes the 
quantal response for Gaussian measurement matrices.   

For  block-diagonal matrices with block-regular sparsity, 
the failure probability is expected to follow the
generalized extreme value distribution, as it involves the product of failure probabilities of individual blocks.
Extreme value theory shows that for large $B$, the 
Complementary Log Log (CLL) distribution  is an appropriate model for quantal response. 
Given certain problem size $(M,B)$, 
that theory states that the expected fractional success rate 
can be approximated by:
\begin{equation}
\pi(\eps| \delta) = \Pr \{ \textit{Success}|\delta,\eps \}  = 1 - \exp \left\{ - \exp (a(\delta) + b(\delta) \cdot \epsilon) \right\},
\end{equation} 
for certain underlying parameters $a=a(\delta)$, and $b=b(\delta)$.
We then define the empirical phase transition location, at each fixed $\delta$, as the sparsity level 
$\epsilon$ at which the success probability $\pi = 1-1/e$ (i.e., 63.2\%).

\subsection{Studying Very Large Problem Sizes}

In the results section we compare models 
(\ref{eq:PTforTensor}) and (\ref{eq:PTforTensor-SO})
to data.
clear understanding of models explaining 
offsets  of order $\gamma$ and $\gamma^2$. This 
required data from experiments conducted at a
range of problem sizes \-- in particular large problem sizes.
Actually, plausible sizes can easily led to computational difficulties.
In a  2D anisotropic undersampling problem on a $768 \times 768$ Fourier grid,
we would be considering block-diagonal  undersampling with parameters
 $M=B = 768$, in which case $N = 768^2 = 589824$.
General-purpose convex optimization solvers such as CVX 
are not really appropriate for solving such large problems.

Nevertheless, we have been able to get precise information about
the behavior of $(P_{1,\bX})$ on block-diagonal problems of such large sizes.
The key comes in applying  Lemma \ref{cor:ExactProbMultiBlock} ,
which allows us to
infer success probabilities for problems of size $N = B \cdot M$, once we
 know them for problems of size $M$. In the cases we are studying, $B=M$, so  $M = \sqrt{N}$ and
 we can use computationally modest resources (denominated in terms
 of $\sqrt{N}$)  to study very large-$N$ problems that 
would ordinarily require massive investments of computational resources.
 
%Instead, we have designed a statistical significance test to solve this problem. Instead of solving large optimization problems corresponding to triples $(m,M,B)$, we solve smaller optimization problems corresponding to single individual blocks of size $m \times M$ for the entire range of plausible sparsity. This is possible thanks to the \emph{separability} of the optimization problem $(P_{1,\bX})$. The runs for small problem sizes are embarrassingly parallel and can be mapped to many nodes on available computational clusters. This task is done  by the software package Clusterjob \cite{clusterjob}. We then use the following test to find the phase transition location for large block diagonal matrix. 

\newcommand{\cQ}{{\cal Q}}
%\subsubsection{Design of test and sampling resolution}

Let $Q_{mb}(k, n, N ; \nu, \bX) $ denote the probability of success in the multiblock optimization problem
$(P_{1,\bX})$ at given $k = B \cdot \ell$, $n = B m$ and $N = B M$, where the component
subproblems are \emph{i.i.d} according to a fixed distribution $\nu$.
Let $Q_{sb} = Q_{sb}(\ell,m,M; \nu , \bX)$ denote the probability of success in a component single-block problem.
 Lemma \ref{cor:ExactProbMultiBlock} gives us the equivalence:
\[
Q_{mb}(k,n,N)  \leq q^*  \Leftrightarrow  Q_{sb}(\ell,m,M)  \leq  (q^*)^{1/B}.
\]
At first blush, a hypothesis on $Q_{mb}$---such as the finite-$N$ phase transition---would
seem to require evidence from trials 
in which the multiblock problem $(P_{1,\bX})$ of total size $N = B\cdot M$ gets solved.
But we have just shown that such a hypothesis on $Q_{mb}$ is equivalent to one on $Q_{sb}$.
We get information about $Q_{sb}$ by solving random instances of
a single-block problem  of size $M$.
 Suppose $k = B \cdot \ell$ and $N = B M$.
 Then the hypothesis that  $\epsilon_{mb}^* <  k/N $ is equivalent to
$Q_{mb}( k , n , N) < q^*$, which is equivalent to 
 $Q_{sb}( \ell , m , M) < (q^*)^{1/B}$. So we can indeed use single-block
problem realizations to shed light on $\epsilon_{mb}^*$.

Generate $S$  independent problem realizations $(A^{(s)},x^{(s)}_0)$, 
each one a single-block
problem instance with size parameters $\ell, m , M$.
Solve each realization in turn and record the binary success indicators
 $X_s = 1_{\{x_1^{(s)} = x_0^{(s)}\}}$. 
 These are Bernoulli random variables at some common but
 unknown success probability, $\pi$, say.
Let $Y_s = 1-X_s$ denote the indicator of failure.
Calculate  the mean failure rate
 $\bar{Y}  = S^{-1} \sum_{s=1}^S Y_s$.

We propose the following statistical test of
$H_0:  (1-\pi)^B \leq  q^*$ against $H_1:  (1-\pi)^B > q^*$.
Fix $\alpha > 0$ small (e.g. $\alpha = 1/20$), and let $z_{1-\alpha/2}$ denote the usual
$1 -\alpha/2$ quantile of the Normal distribution, 
so that $z_{.975} \approx 1.96$.
Define $\mu = \mu_B = \log(1/q^*)/B$.
Reject the hypothesis $H_0$ if the failure fraction is high:
\[
      \bar{Y} >   \mu + z_{1-\alpha/2}  \sqrt{\frac{\mu}{S}}.
\]
Accept $H_0$  if the fraction of failures is  low:
\[
      \bar{Y} <  \mu - z_{1-\alpha/2}  \sqrt{\frac{\mu}{S}}.
\]
Make no decision otherwise.

{\bf Derivation:}
Let $q_B = 1- (q^*)^{1/B}$, and suppose our variables were distributed as 
$X_s \sim \textrm{Ber}((q^*)^{1/B})$, i.e., just on the sharp edge of the asymptotic
phase transition at problem size $B$. Then
$Y_s \sim \textrm{Ber}(q_B)$. Let $T = \sum_{s=1}^S Y_s$;
then $T \sim_\textrm{approx} \textrm{Poi}(\lambda)$, where $\lambda = S \cdot \mu$.
By normal approximation to the binomial,
 when $\lambda$ is large, 
$T \sim_\textrm{approx} N(\lambda,\lambda)$.
Consequently,
\[
    \Pr \{ T \in [ \lambda - z_{1-\alpha/2} \sqrt{\lambda} , \lambda + z_{1-\alpha/2} \sqrt{\lambda} ] \} \approx 1-\alpha ,
\]
where the approximation gets increasingly good as $\lambda \rightarrow \infty$. The rule we proposed above then follows.

%%\sidecomment[R1-][MS]{Text missing}
Another way to write the rule sets  $T = S \cdot \bar{Y}$. Then we can decide to reject/accept just in case
\[
   \bar{Y}  \not\in  \mu \cdot  \left( 1 \pm  \frac{z_{1-\alpha/2}}{\sqrt{S \cdot  \mu}} \right).
\]
The probability of mistaken rejection is approximately $\alpha$.

% We can use this rule to understand sampling resolution.  In particular, if $\alpha = 0.05$ then
% we should pick $\lambda = 16$ in order that the above interval becomes
% \[
%    \bar{Y}  \not\in \frac{\log(1/q^*)}{B} \left( 1/2,3/2 \right).
% \]
% A Monte Carlo sample size $S \geq \frac{16 \cdot  B}{\log(1/q^*)} $ allows us to resolve which side of the phase transition
% we are on, except in a swath around $ \frac{\log(1/q^*)}{B} $ of relative width about $50\%$. 
% We used $S =32 B$ in our computations. 

\section{Results}

\subsection{Data collection}
To efficiently generate the quantal response data for various ensembles, we have developed and used
software package Clusterjob (CJ) \cite{clusterjob} \-- a collection of Perl scripts for automating reproducibility and 
hassle-free submission of massive computational jobs to clusters. 
Our computational jobs have mainly run on three different clusters at Stanford, namely \verb+sherlock+, \verb+solomon+, and \verb+proclus+. The optimization solvers used include \verb+ASP+\cite{ASP,BPdual}, \verb+CVX+ \cite{CVX}, and {\tt MOSEK}\cite{Mosek}. It is worth mentioning that software package CVX uses SDPT3 and SEDUMI as its main optimization solvers. 
Our dataset currently includes $29$ million rows, which are the results of nearly 35 million Monte Carlo runs for various problem sizes, and ensembles including \verb+RBUSE+, \verb+DBUSE+, \verb+RBPFT+, etc. For experiments involving smaller problem sizes, one row of data contains information such as the probability of successful reconstruction and error in reconstruction for a particular quadruple $(\ell,m,M,B)$ in the phase space. For data of larger problem sizes, one row contains information such as error in reconstruction and a binary number indicating success or failure for a particular triple $(\ell,m,M)$.

\subsection{Verifying predictions} \label{sec:Verify}

Figures \ref{fig-RBUSE-bnd} through \ref{fig-RBUSE-C-error} show the comparison of experimental phase transition data against the first-order and second-order predictions for the four different coefficient sets $\bX \in \{[0,1], \bR_+, \bR, \bC \}$. As an example, Figure \ref{fig-RBUSE-bnd} shows the empirical offset from the asymptotic phase transition 
location and the corresponding predictions for the case $\bX = [0,1]$, 
for which precise and mathematically rigorous results 
were derived in Section 4. 
In all these cases, the match between the predictions and data is quite good. 
The figures also show that our second-order correction terms improve the predictions of the phase transition location \-- especially for smaller problem sizes.

%\subsubsection*{+ Bounded Coefficients}

% We first discuss the case of recovering
% a bounded object, for which exact rigorous results 
% were presented in Section 4. 
% Figure \ref{fig-RBUSE-bnd} shows the empirical offset from the asymptotic phase transition 
% location and the corresponding predictions under our formalism. 
% The match between the theoretical prediction and 
% data is clear. The figure also shows that
% the second-order correction term improves the predictions of the phase 
% transition location \-- especially for smaller problem sizes. 

\begin{figure}[htbp]
\centering
\includegraphics[width=2.5in,angle=0]%{/Users/hatef/Dropbox/TensorCS/ANALYSIS/analysis_for_proposal_defense/TFUSE_cll_offsetVSprediction/figs/Figure_RBUSE_dataVSpred}
%{figs/Figure_TFUSE_dataVSpred}
{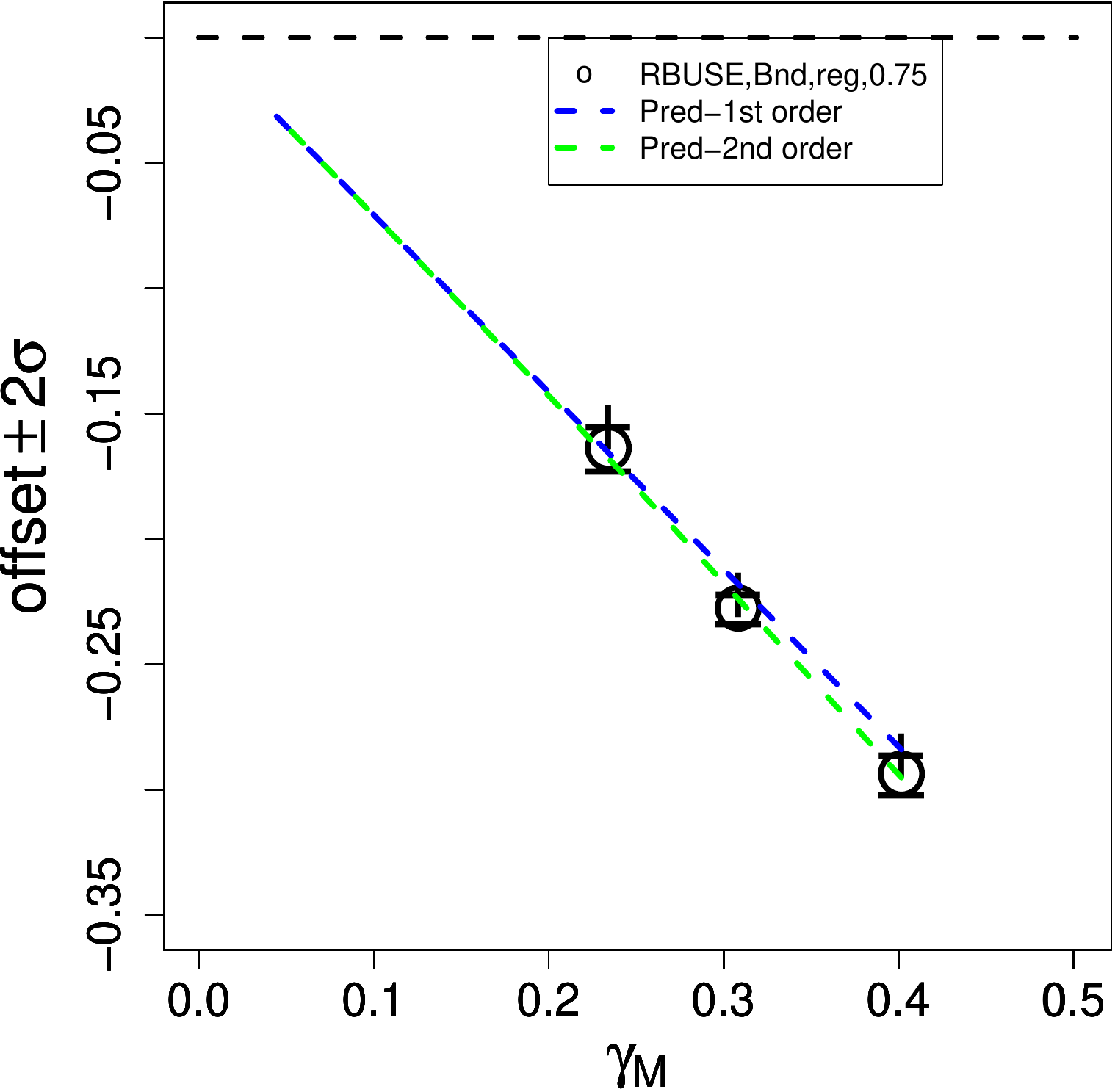}
\caption{Offset $\easy - \epreg(m,M)$ versus $\gamma_M$ for RBUSE ensemble and $\bX = [0,1]$ at $\delta = m/M = 3/4$. Problem sizes: $M = 48$, $96$ and $192$. The green and blue curves show the predictions with and without considering the second-order effects, respectively. The plus symbol $+$ locates the $+ 2 s.e.$ confidence bar and minus symbol `-' locates the $- 2 s.e.$ limit.}
\label{fig-RBUSE-bnd}
\end{figure}

\begin{figure}[htbp]
\centering
\begin{tabular}{cc}
\includegraphics[width=2.5in,angle=0]%{/Users/hatef/Dropbox/TensorCS/ANALYSIS/analysis_for_proposal_defense/RBUSE_cll_offsetVSprediction/figs/Figure_Pos_RBUSE_dataVSpred_deltaOneFourth} &
%{figs/Figure_Pos_TFUSE_dataVSpred_deltaOneFourth} &
{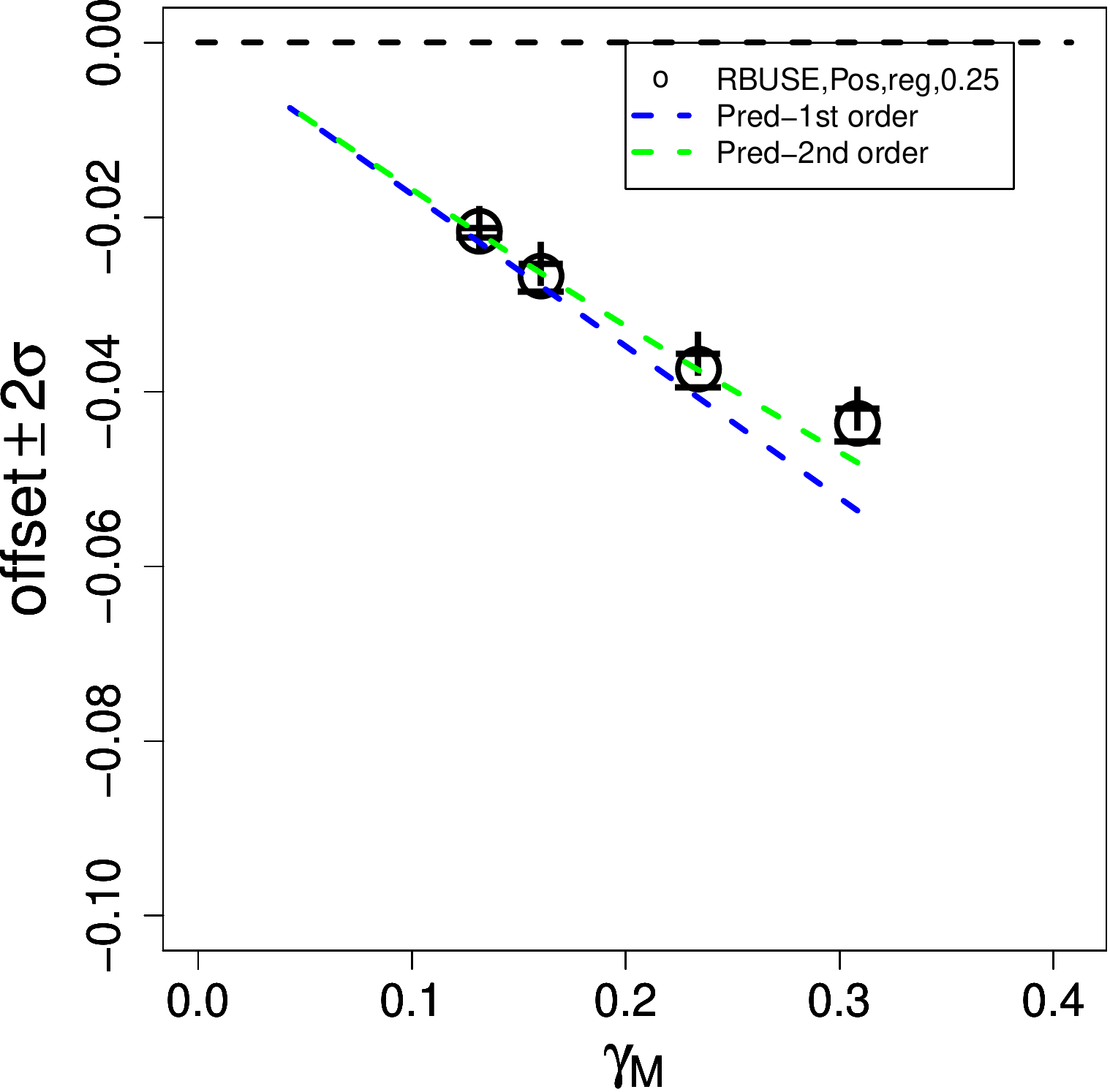} &
\includegraphics[width=2.5in,angle=0]
{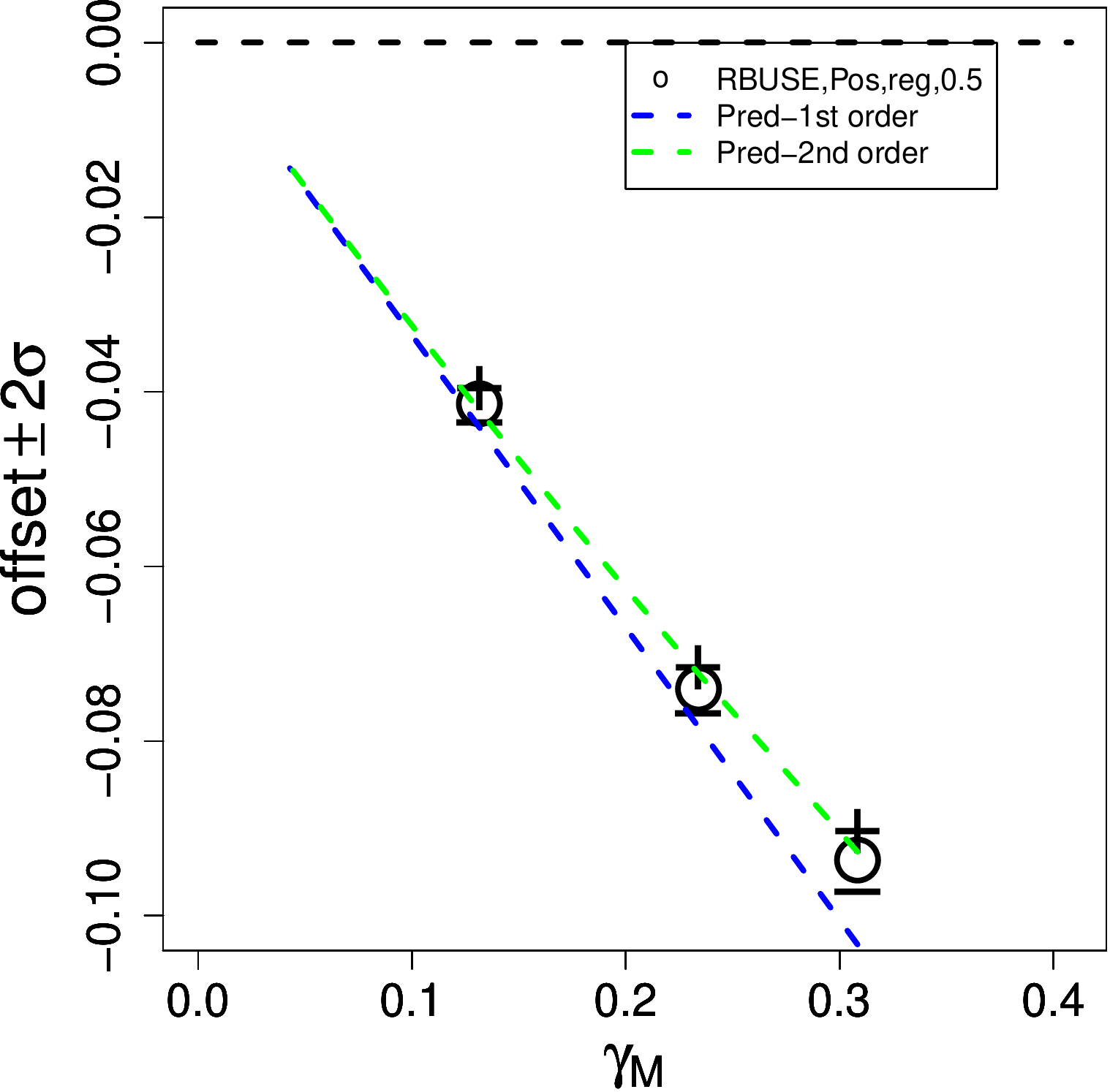} 
%{figs/Figure_Pos_TFUSE_dataVSpred_deltaHalf}
\end{tabular}
\caption{Offset $\easy - \epreg(m,{M})$ versus $\gamma_M$ for RBUSE ensemble and $\bX = \bR_+$  at $\delta = 1/2$ (left panel) and $\delta = 1/4$ (right panel). Problem sizes: $M = 96,192,480$ and $768$. The green and blue curves show the predictions with and without using the second-order term, respectively. The plus symbol $+$ locates the $+ 2 s.e.$ confidence bar and minus symbol `-' locates the $- 2 s.e.$ limit.}
\label{fig-RBUSE-pos-errorbars}
\end{figure}

\clearpage
\subsubsection*{+ Positive Coefficients}
\begin{figure}[htbp]
\centering
\begin{tabular}{cc}
\includegraphics[width=2.5in,angle=0]
{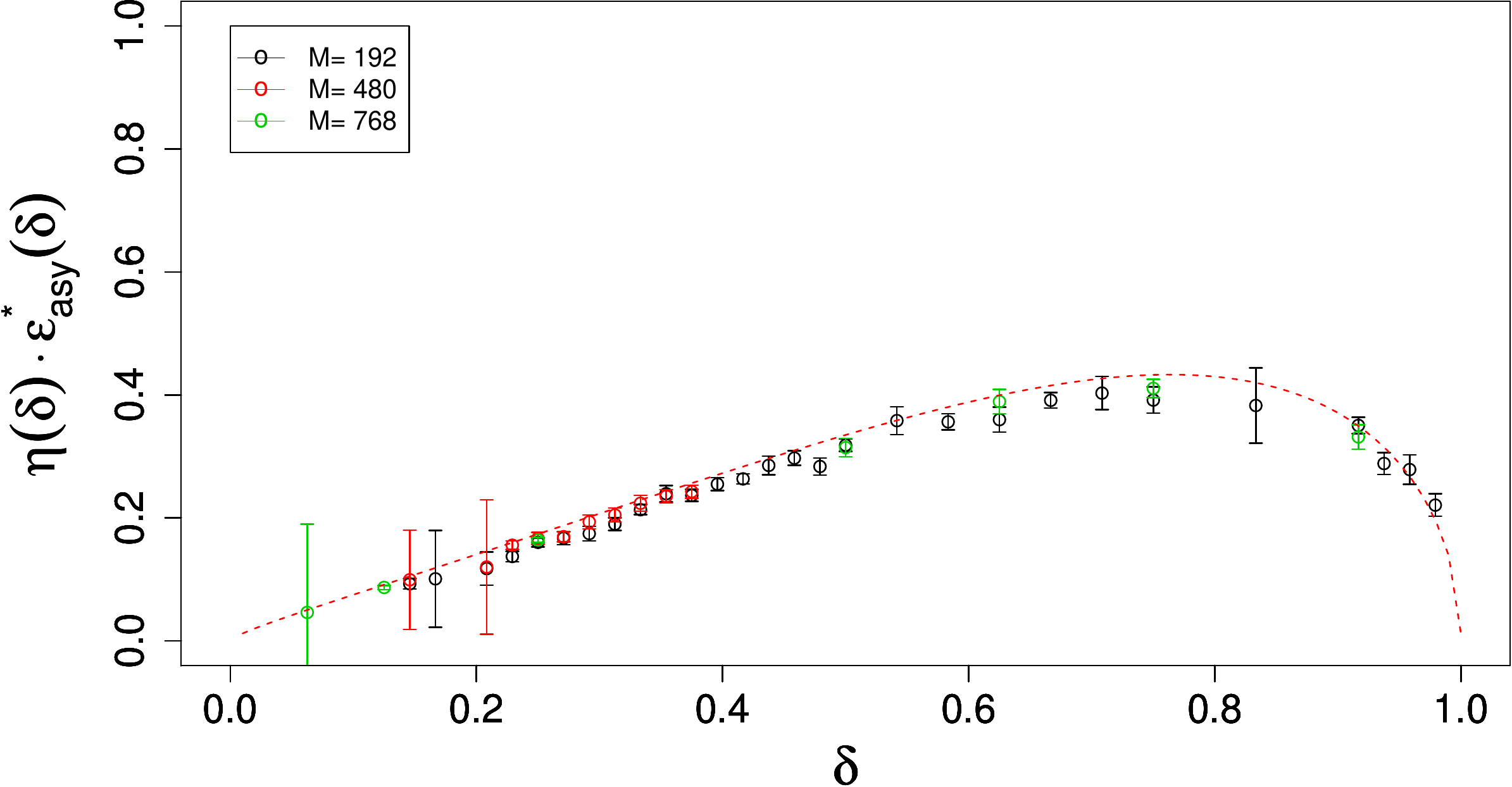} &
%{figs/finiteN_offset_vs_delta_Pos_qReg1_so_FALSE} &
\includegraphics[width=2.5in,angle=0]
{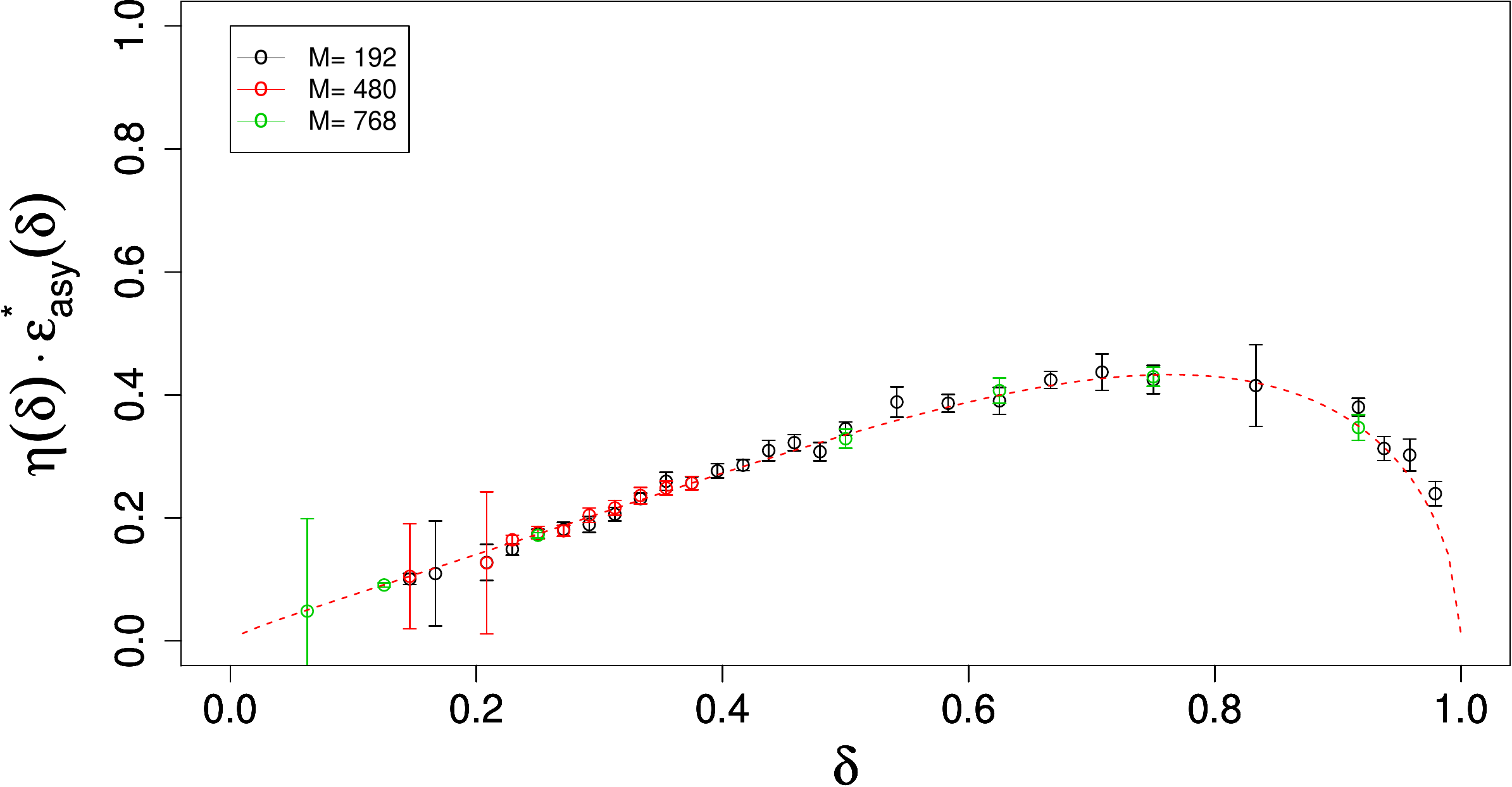} 
%{figs/finiteN_offset_vs_delta_Pos_qReg1_so_TRUE}
\end{tabular}
\caption{The ratio of offset  $\easy - \epreg(m,{M})$ to the first-order coefficient $\alpha \gamma_M$ (left panel) and  to the second-order coefficient $(\alpha \gamma_M+ \beta \gamma_M^2)$ (right panel) versus undersampling fraction $\delta$.  Here, RBUSE ensemble and $\bX = \bR_+$ coefficient set. Problem sizes: $M =192,480$ and $768$. The red dashed curve shows the predicted curves $\eta(\delta) \cdot \easy(\delta)$.}
\label{fig-RBUSE-Pos-fDeltaEpsilon}
\end{figure}

\begin{figure}[htbp]
\centering
\begin{tabular}{cc}
\includegraphics[width=2.5in,angle=0]
%{figs/finiteN_full_pt_Pos_qReg1_so_FALSE} &
{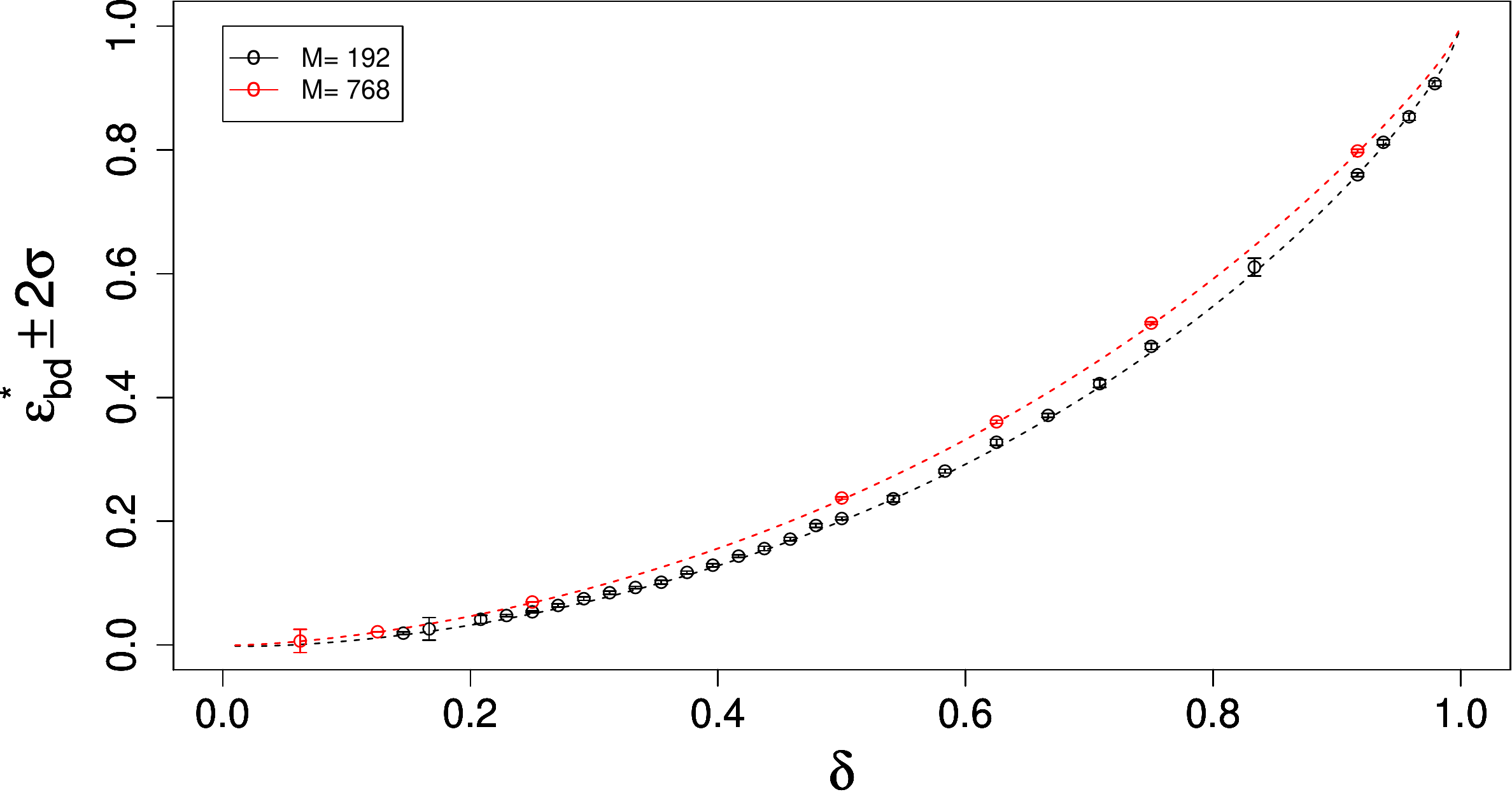} &
\includegraphics[width=2.5in,angle=0]
{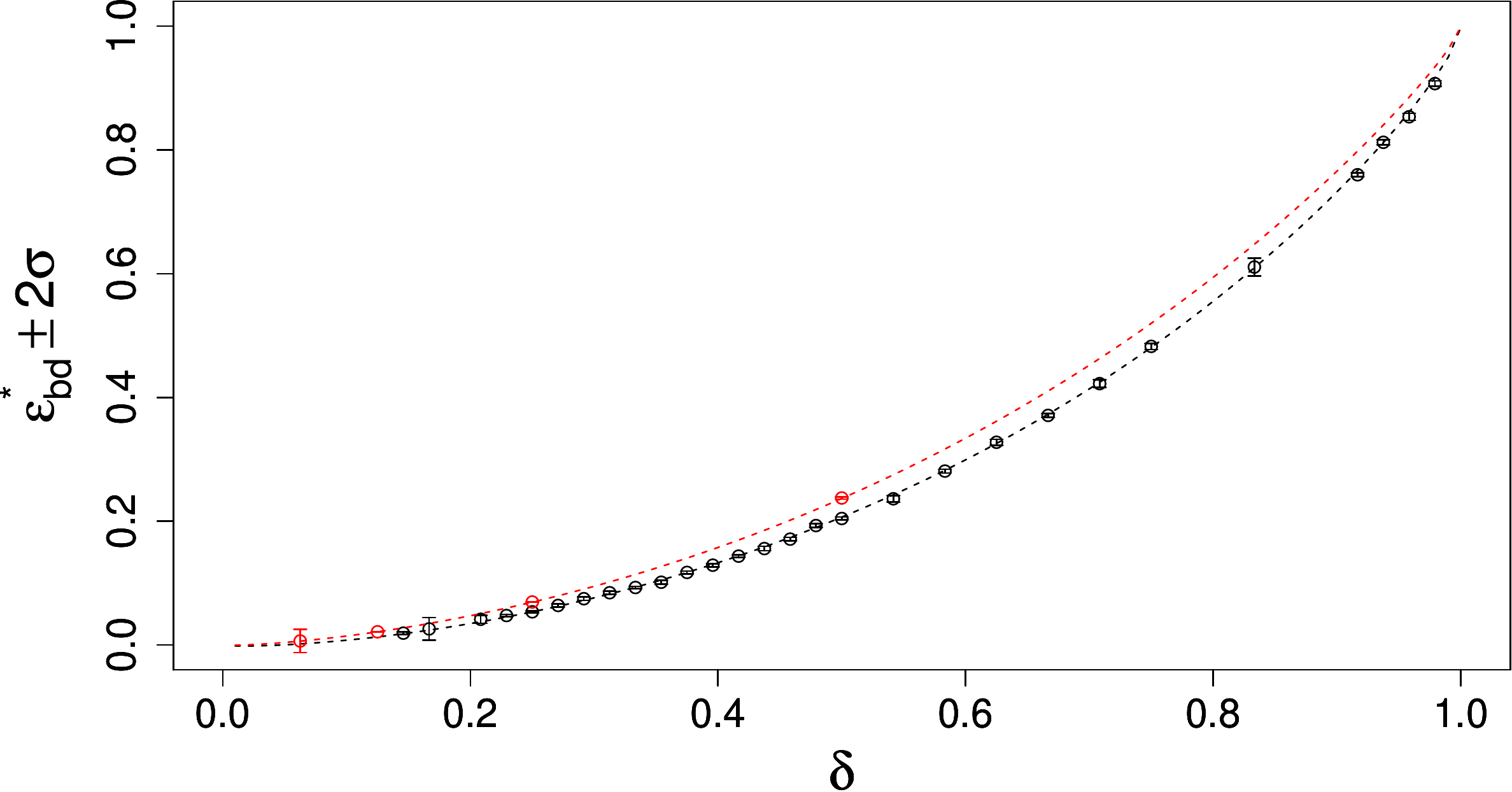}
%{figs/finiteN_full_pt_Pos_qReg1_so_TRUE}
\end{tabular}
\caption{Experimental data against first-order (left), and second-order (right) predictions of phase transition for $\bX = \bR_+$. Problem sizes: $M =192,768$. The circles show data and the dashed lines show predictions.}
\label{fig-RBUSE-R-fullPT}
\end{figure}

\begin{figure}[htbp]
\centering
\begin{tabular}{cc}
\includegraphics[width=2.5in,angle=0]
%{figs/residuals_Pos_qReg1_so_FALSE.pdf} &
{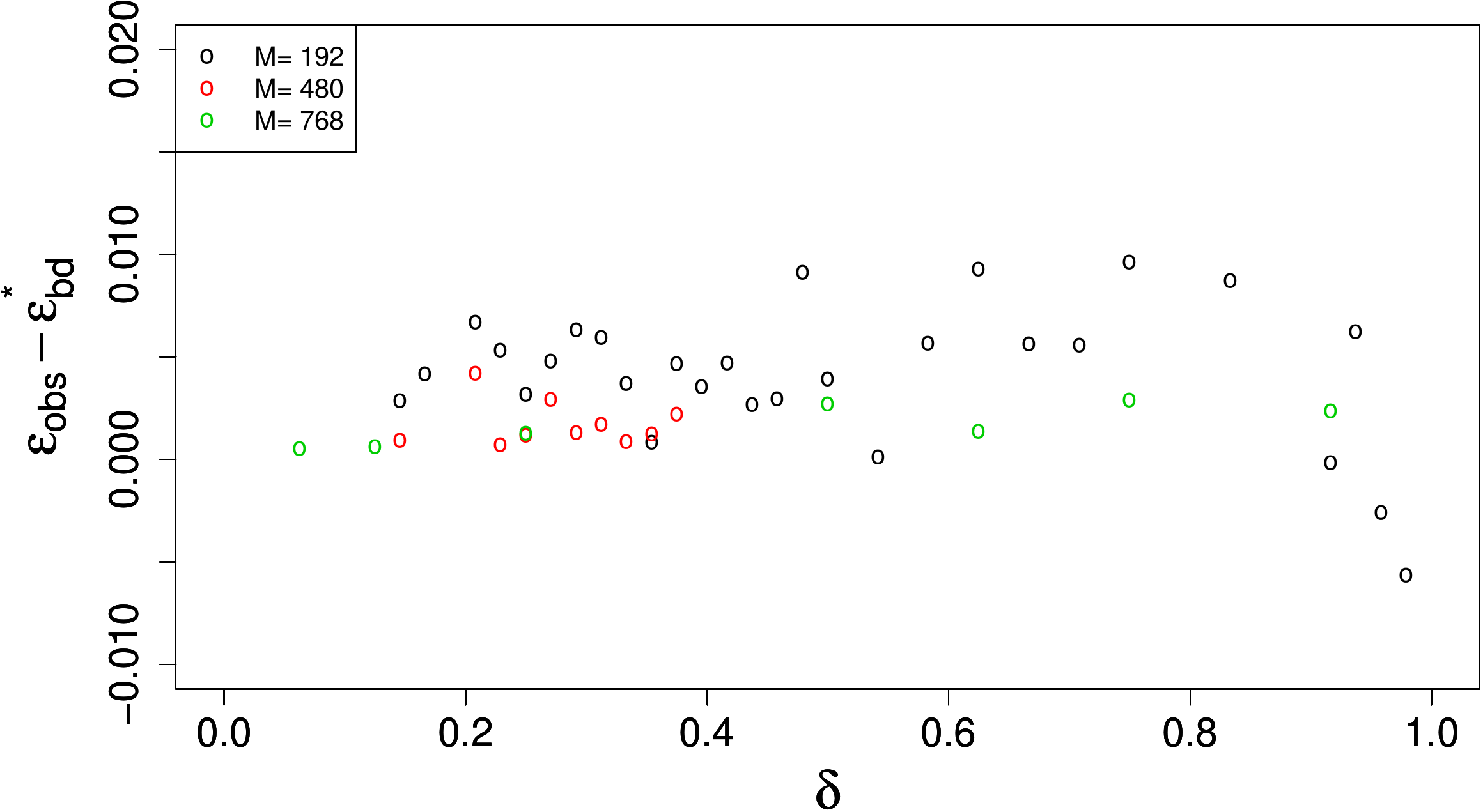} &
\includegraphics[width=2.5in,angle=0]
{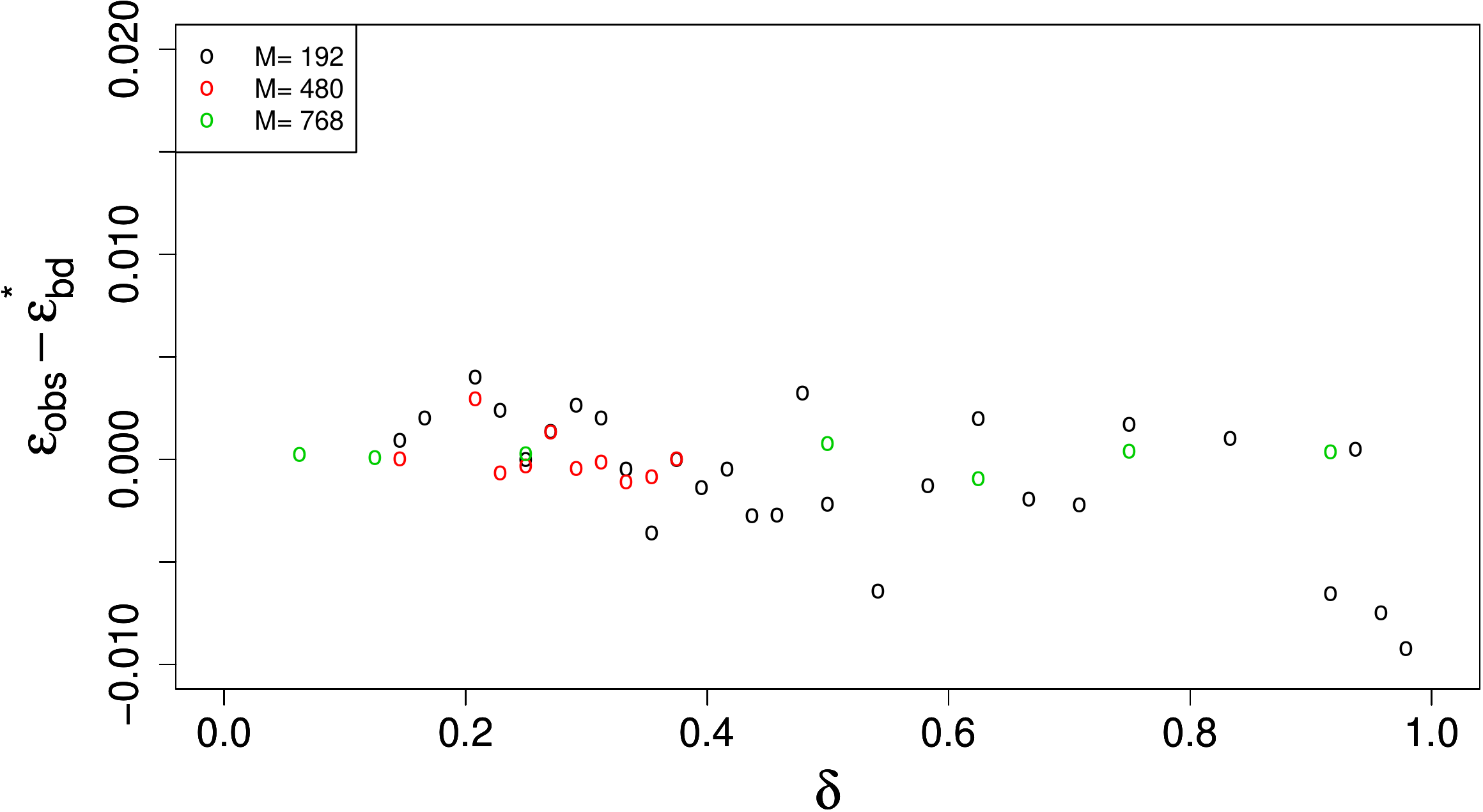}
%{figs/residuals_Pos_qReg1_so_TRUE.pdf}
\end{tabular}
\caption{Difference of predicted and experimental phase transition location for $\bX = \bR_+$
using first-order (left), and second-order (right) predictive models. Problem sizes: $M =192,480$ and $768$.
Residuals are larger near $\delta \approx 1$, 
the residuals at $M=768$ and large $\delta$ are noticeably smaller than 
those at $M=192$.}
\label{fig-RBUSE-Pos-error}
\end{figure}

% \begin{figure}[htbp]
% \centering
% \begin{tabular}{cc}
% \includegraphics[width=2.1in,angle=0]{figs/finiteN_fractionaloffset_vs_delta_Pos_qReg1_so_FALSE} &
% \includegraphics[width=2.1in,angle=0]{figs/finiteN_offset_vs_delta_Pos_qReg1_so_FALSE}
% \end{tabular}
% \caption{The ratio of fractional offset $(\tilde{\eps}-\epsilon_N)/\tilde{\eps}$ (left panel) and offset  $\tilde{\eps} - \epsilon_N$ (right panel) to $\alpha \gamma_N$ as a function of under-sampling $\delta$ for RBUSE ensemble and $X = \bR_+$ coefficient field. Problem sizes: $N =192,480$ and $768$. The red dashed curve shows the predictive curves $f(\delta)$ and $f(\delta) \cdot \epsilon(\delta)$ respectively. This plot corresponds to a first order model.}
% \label{fig-RBUSE-bnd-fDeltaEpsilon-fo}
% \end{figure}

\clearpage
\subsubsection*{+ Real Coefficients}

\begin{figure}[htp]
\centering
\begin{tabular}{cc}
\includegraphics[width=2.5in,angle=0]
{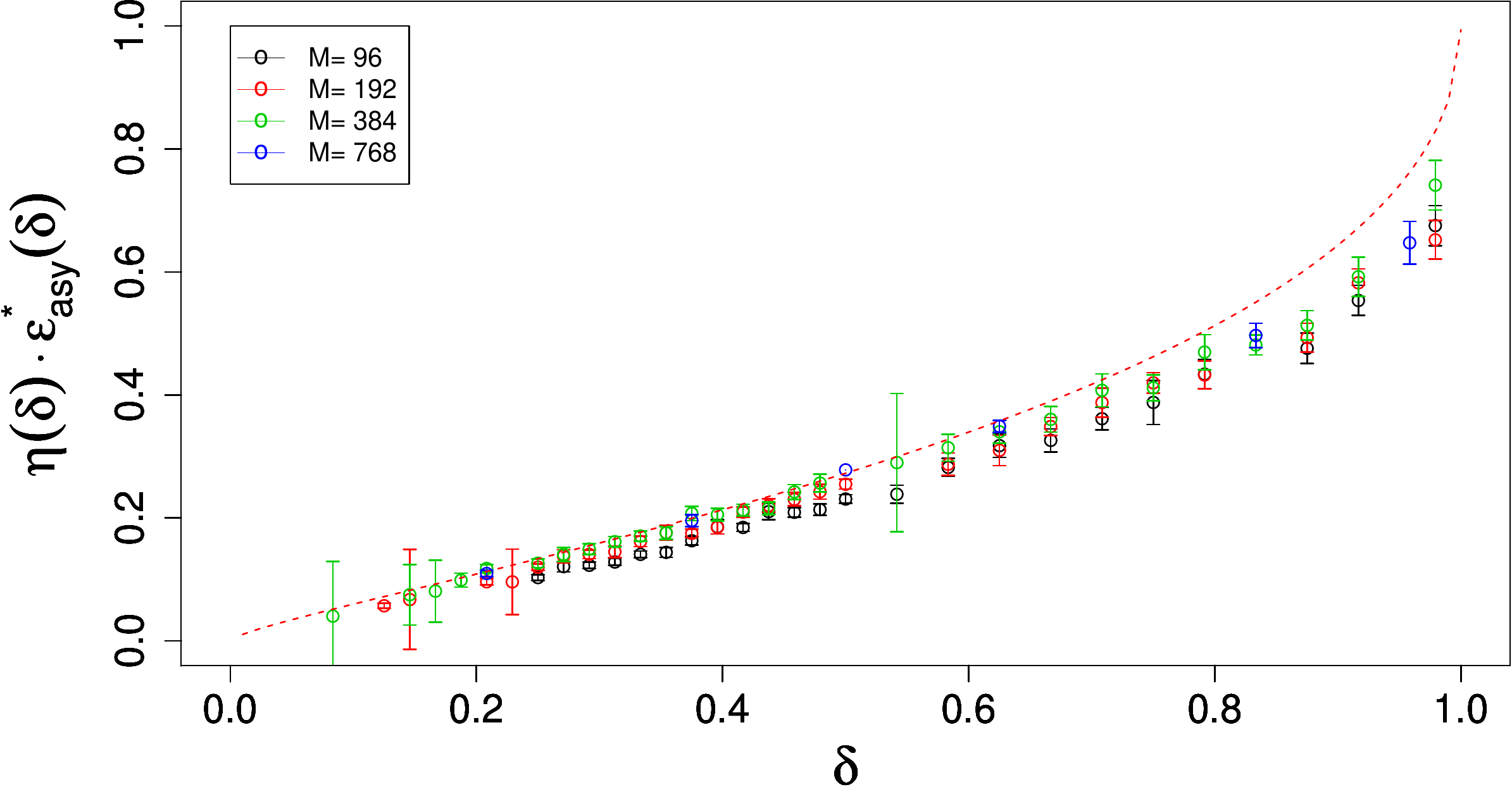} &
%{figs/finiteN_offset_vs_delta_R_qReg1_so_FALSE}&
\includegraphics[width=2.5in,angle=0]
{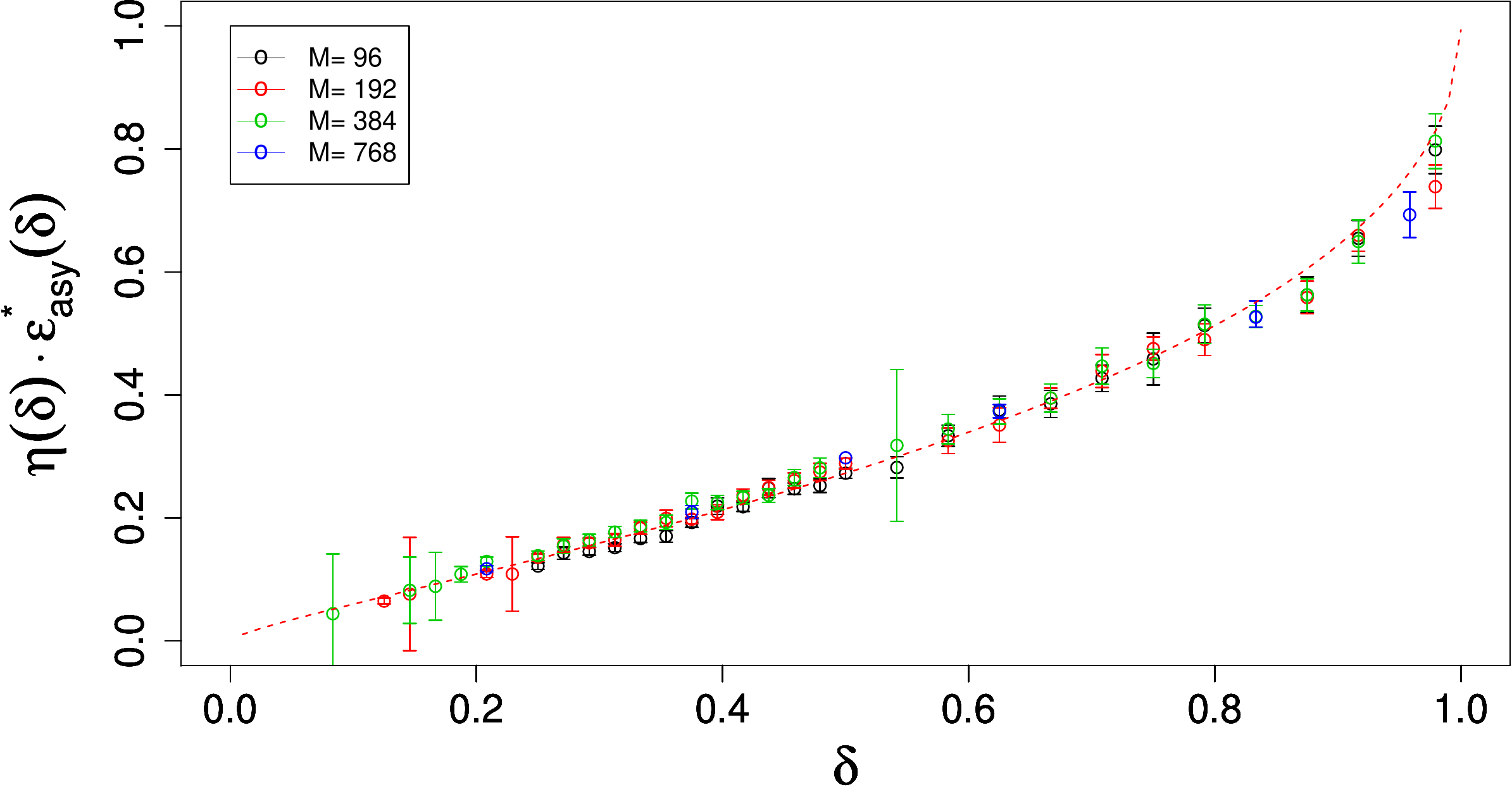} 
%{figs/finiteN_offset_vs_delta_R_qReg1_so_TRUE}
\end{tabular}
\caption{The ratio of offset  $\easy - \epreg(m,{M})$ to $\alpha \gamma_M$ (left panel-first order) and  $(\alpha \gamma_M+ \beta \gamma_M^2)$ (right panel-second order) versus undersampling $\delta$ for RBUSE ensemble and $\bX = \bR$. Problem sizes $M =96,192,384$ and $768$. The red dashed curve shows the predicted curves $\eta(\delta) \cdot \easy(\delta)$.}
\label{fig-RBUSE-R-fDeltaEpsilon}
\end{figure}

\begin{figure}[htbp]
\centering
\begin{tabular}{cc}
\includegraphics[width=2.5in,angle=0]
{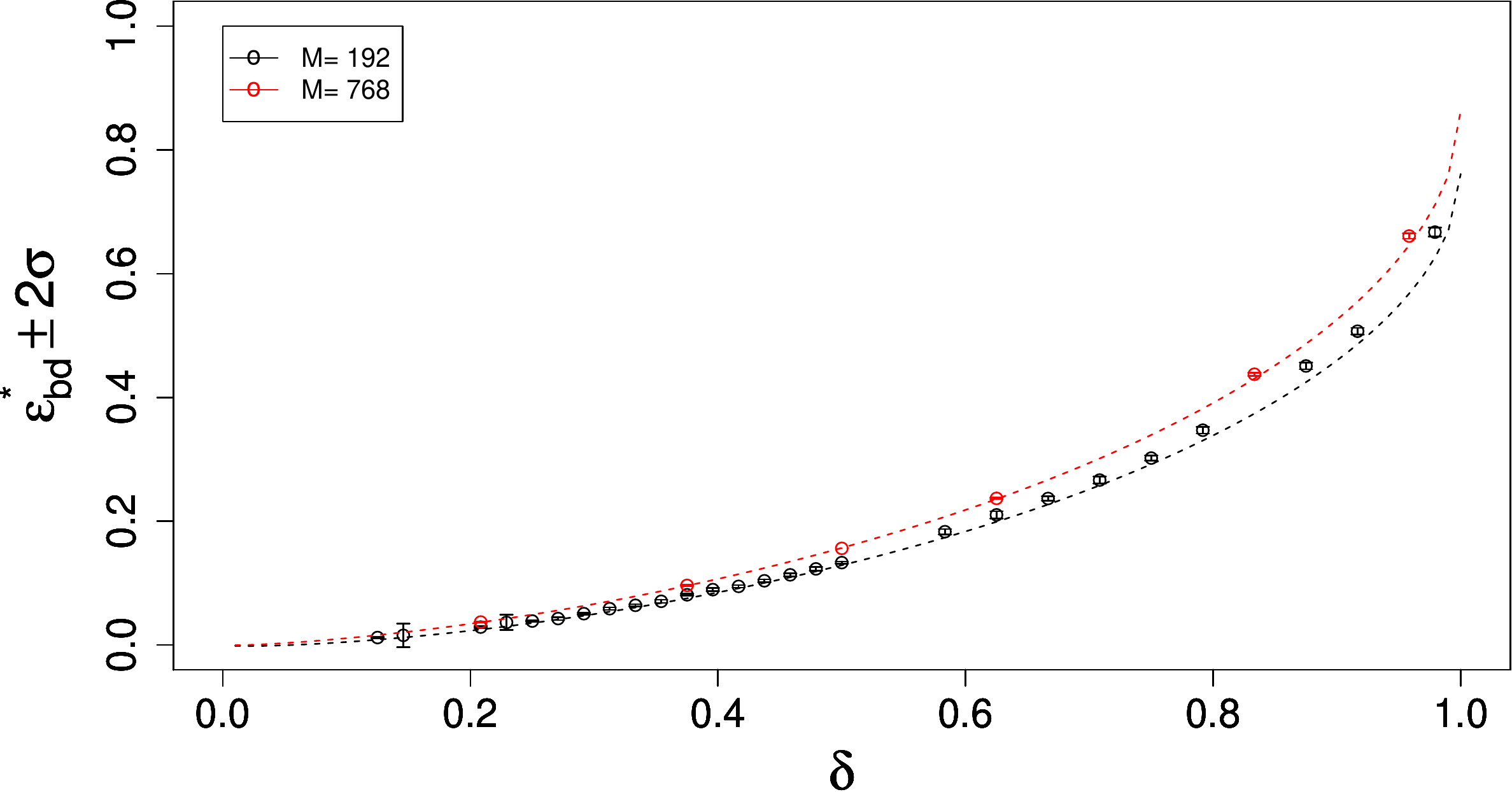} & 
%{figs/finiteN_full_pt_R_qReg1_so_FALSE} &
\includegraphics[width=2.5in,angle=0]
{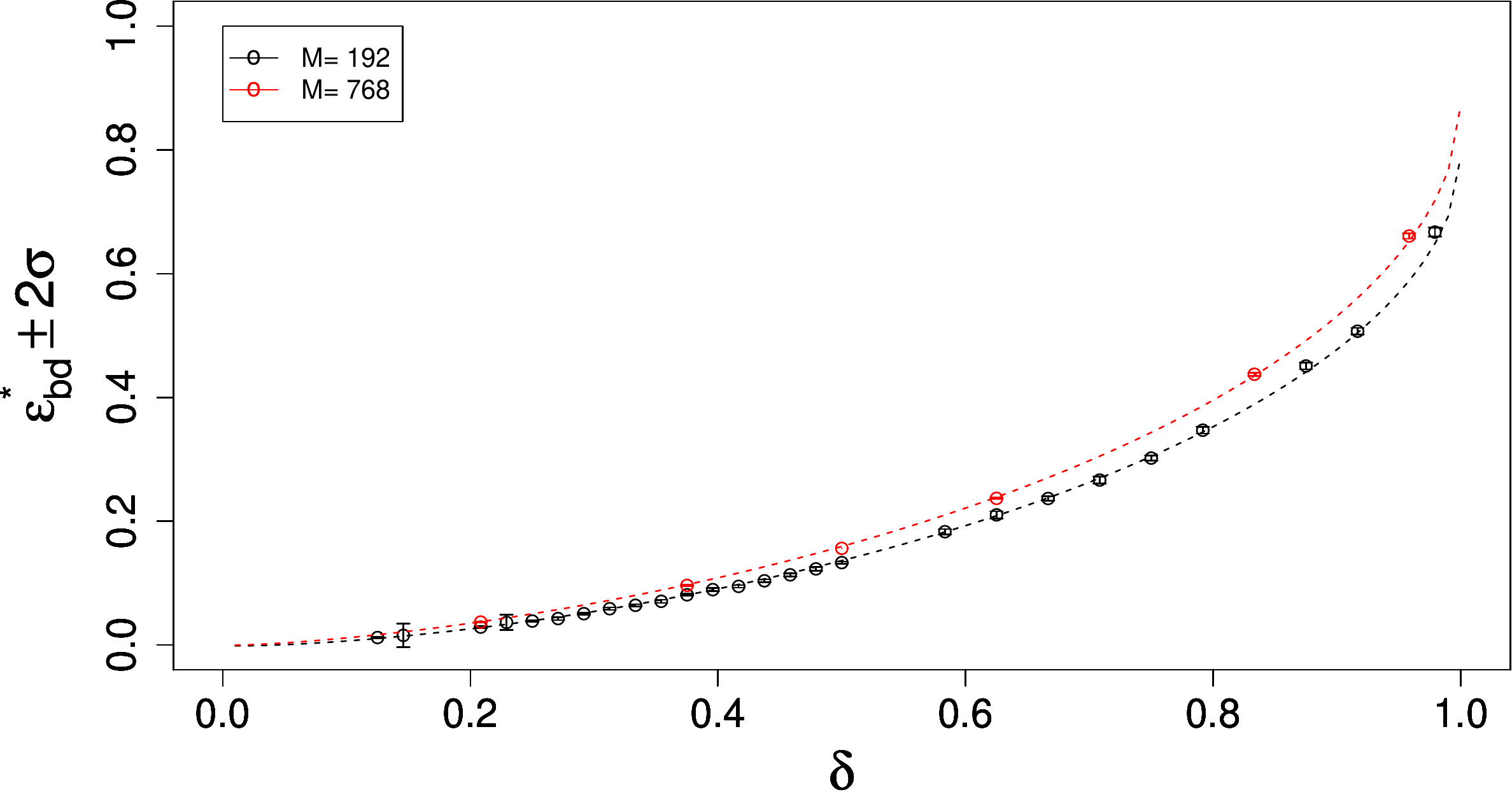}
%{figs/finiteN_full_pt_R_qReg1_so_TRUE}
\end{tabular}
\caption{Experimental data against first-order (left), and second-order (right) predictions of phase transition for $\bX = \bR$. Problem sizes: $M =192,768$. The circles are data and the dashed lines are predictions.}
\label{fig-RBUSE-R-fullPT}
\end{figure}

\begin{figure}[htbp]
\centering
\begin{tabular}{cc}
\includegraphics[width=2.5in,angle=0]
%{figs/residuals_Pos_qReg1_so_FALSE.pdf} &
{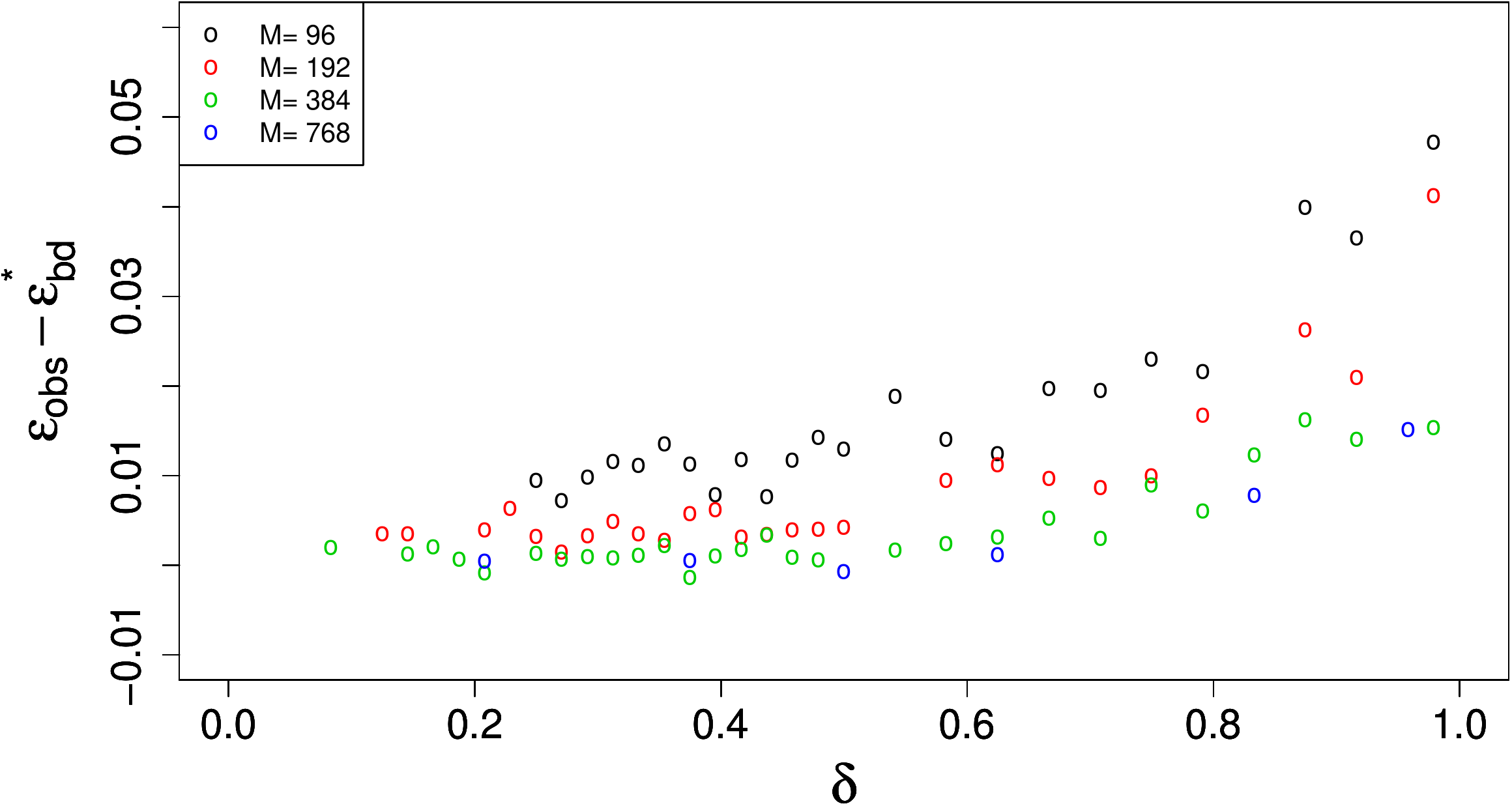} &
\includegraphics[width=2.5in,angle=0]
{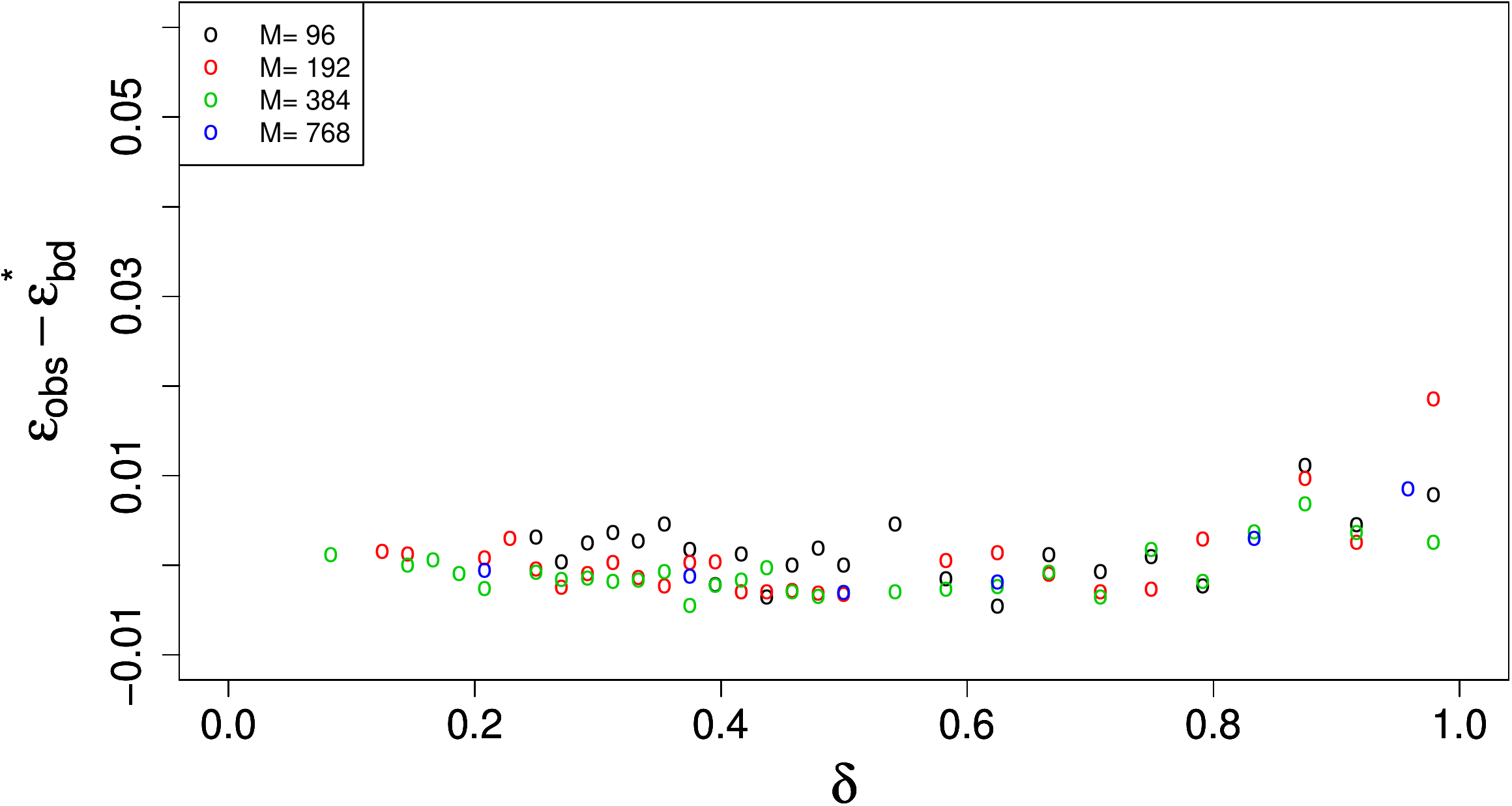}
%{figs/residuals_Pos_qReg1_so_TRUE.pdf}
\end{tabular}
\caption{Difference of predicted and experimental phase transition location for $\bX = R$ using first-order (left), and second-order (right) predictive models. Problem sizes $M =96,192,384$ and $768$.
Residuals are larger near $\delta \approx 1$, 
the residuals at $M=768$ and large $\delta$ are noticeably smaller than 
those at $M=96$.}
\label{fig-RBUSE-R-error}
\end{figure}

% \begin{figure}[htbp]
% \centering
% \begin{tabular}{cc}
% \includegraphics[width=2.1in,angle=0]{figs/finiteN_fractionaloffset_vs_delta_R_qReg1_so_FALSE} &
% \includegraphics[width=2.1in,angle=0]{figs/finiteN_offset_vs_delta_R_qReg1_so_FALSE}
% \end{tabular}
% \caption{The ratio of fractional offset $(\tilde{\eps}-\epsilon_N)/\tilde{\eps}$ (left panel) and offset  $\tilde{\eps} - \epsilon_N$ (right panel) to $\alpha \gamma_N$ as a function of under-sampling $\delta$ for RBUSE ensemble and $X = R$ coefficient field. Problem sizes $N =96,192$ and $384$. The red dashed curve shows the predictive curves $f(\delta)$ and $f(\delta) \cdot \epsilon(\delta)$ respectively. This plot corresponds to a first order model.}
% \label{fig-RBUSE-bnd-fDeltaEpsilon-fo}
% \end{figure}

\clearpage
\subsubsection*{+ Complex Coefficients}

\begin{figure}[htp]
\centering
\begin{tabular}{cc}
\includegraphics[width=2.5in,angle=0]
{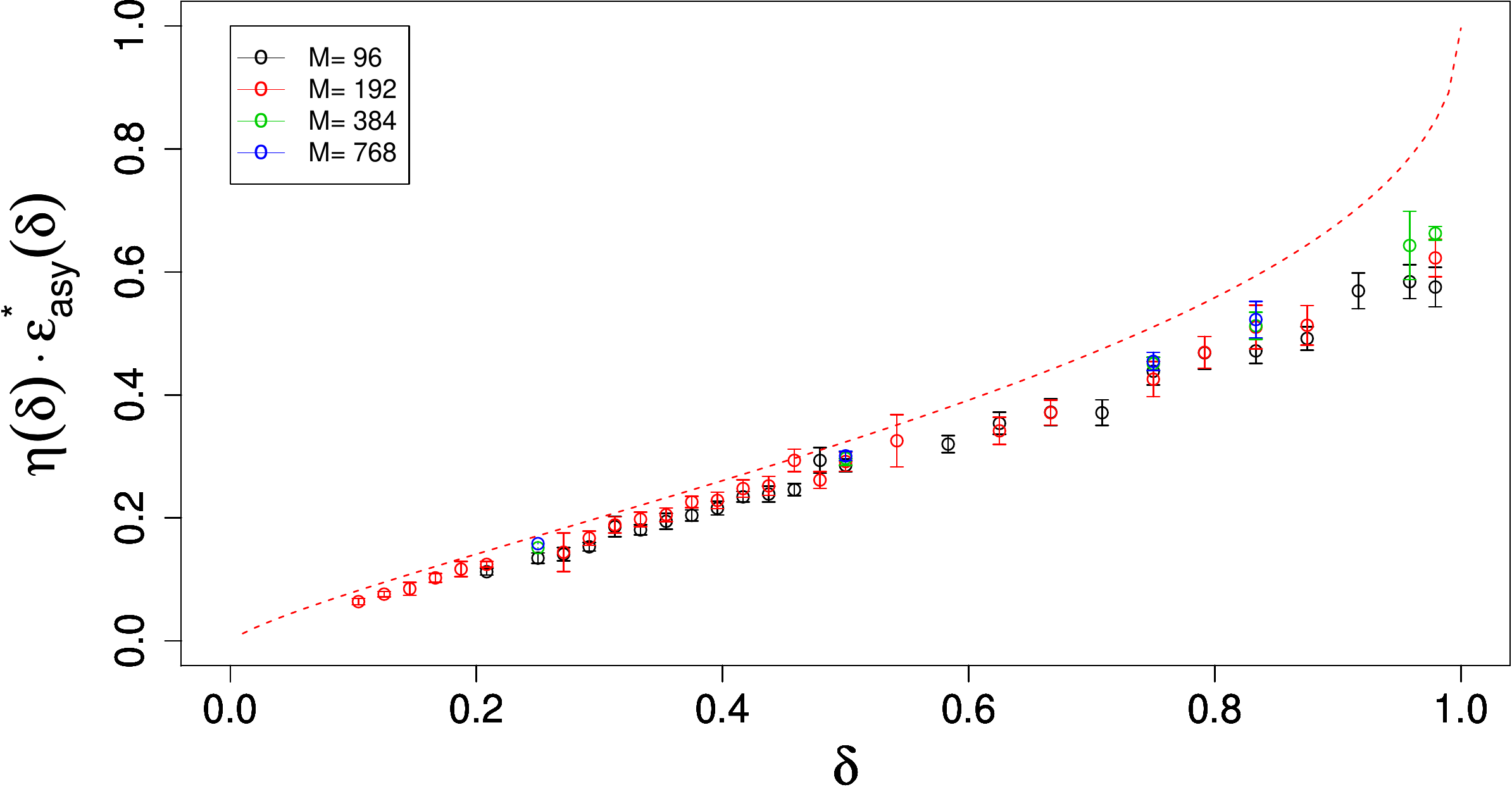} & 
%{figs/finiteN_offset_vs_delta_C_qReg1_so_FALSE}&
\includegraphics[width=2.5in,angle=0]
{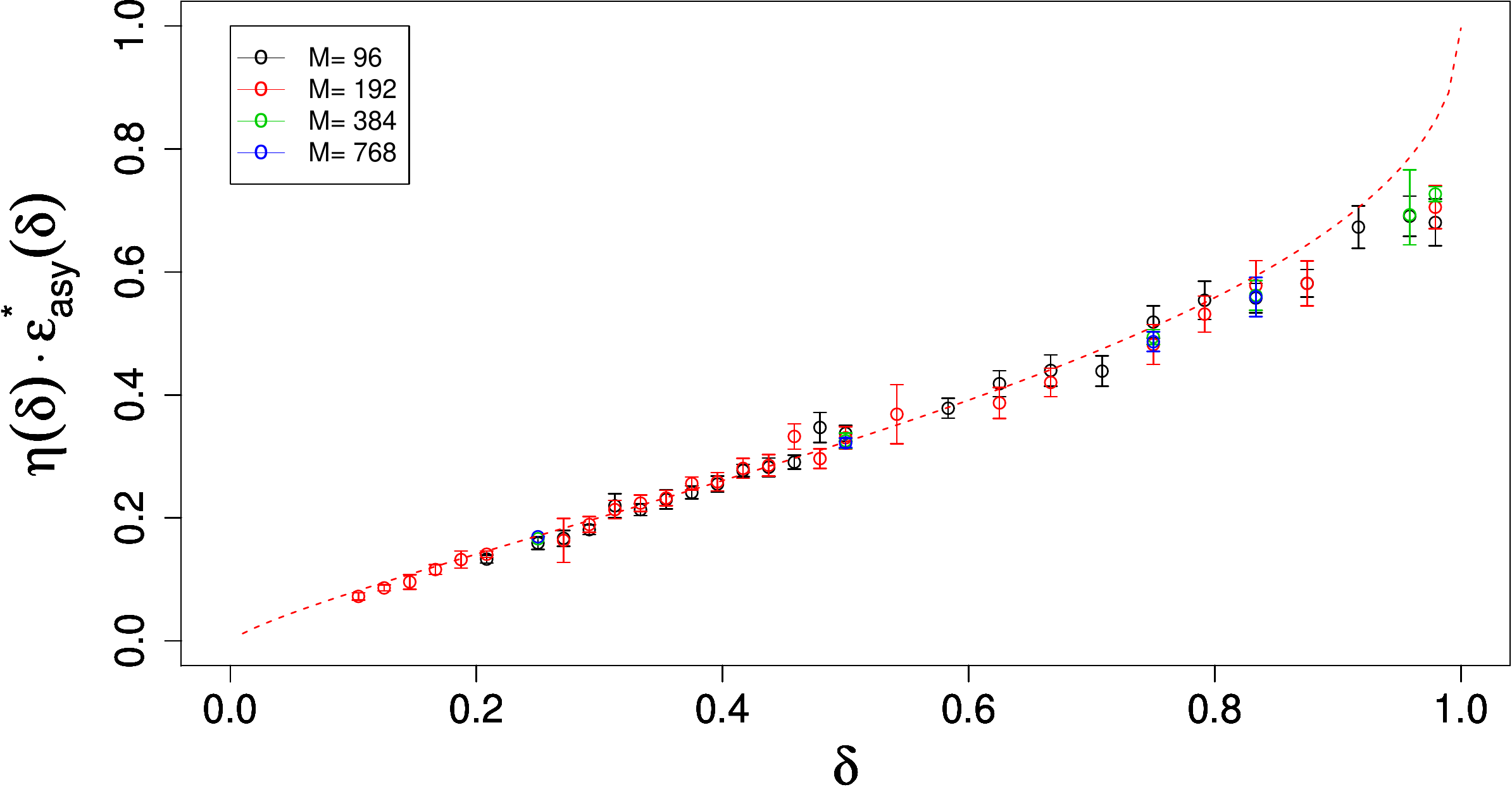} 
%{figs/finiteN_offset_vs_delta_C_qReg1_so_TRUE}
\end{tabular}
\caption{The ratio of offset  $\easy - \epreg(m,{M})$ to $\alpha \gamma_M$ (left panel-first order) and  $(\alpha \gamma_M+ \beta \gamma_M^2)$ (right panel-second order) versus undersampling $\delta$ for RBUSE ensemble and $\bX = \bC$. Problem sizes $M =96,192,384$ and $768$. The red dashed curve shows the predictive curves $\eta(\delta) \cdot \epsilon(\delta)$.}
\label{fig-RBUSE-R-fDeltaEpsilon}
\end{figure}%
\begin{figure}[htbp]
\centering
\begin{tabular}{cc}
\includegraphics[width=2.5in,angle=0]
{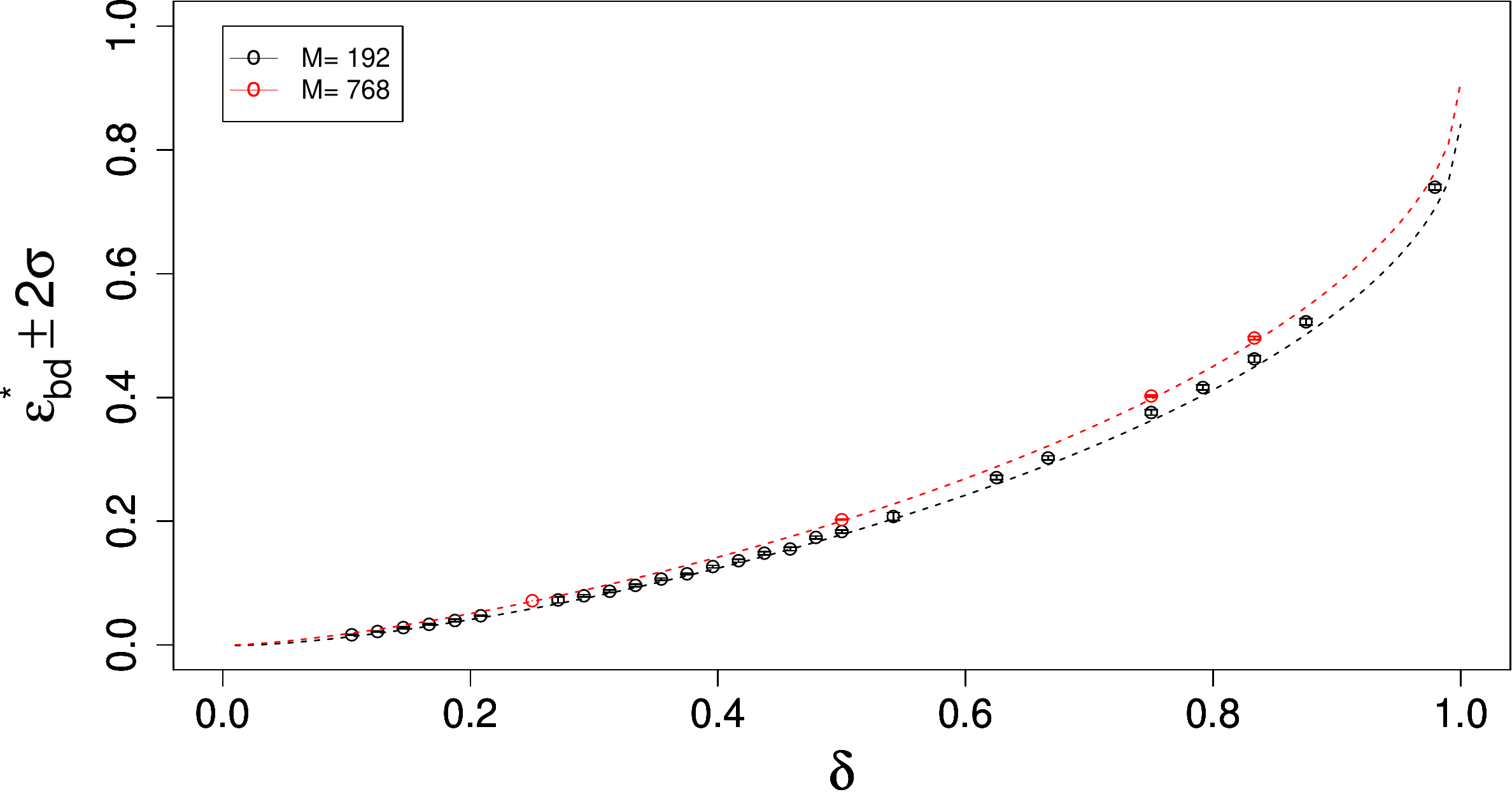} &
%{figs/finiteN_full_pt_C_qReg1_so_FALSE} &
\includegraphics[width=2.5in,angle=0]
{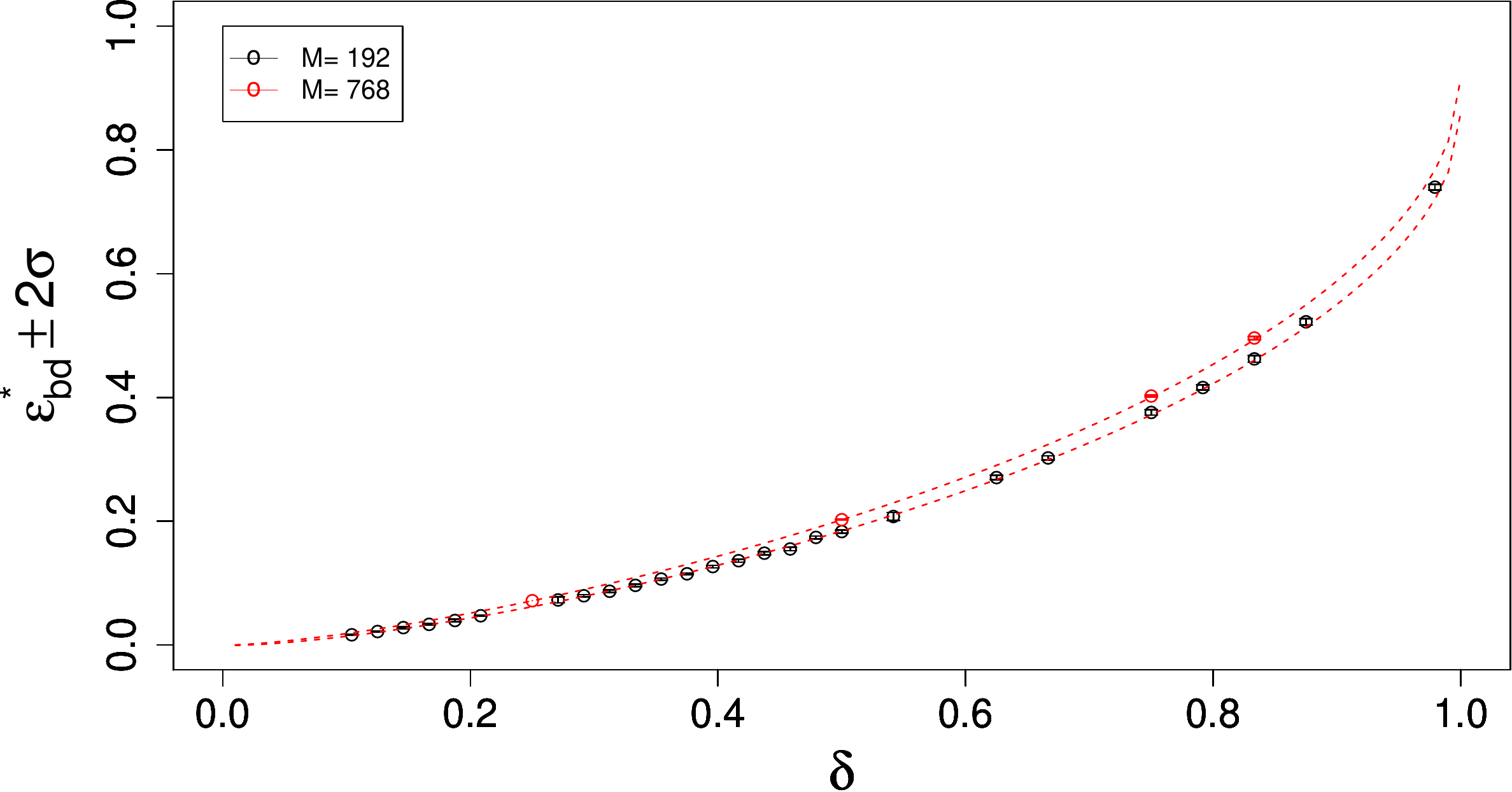}
%{figs/finiteN_full_pt_C_qReg1_so_TRUE}
\end{tabular}
\caption{Experimental data against first-order (left), and second-order (right) predictions of phase transition for $\bX = \bC$. Problem sizes: $M =192,768$. The circles are data and the dashed lines are predictions.}
\label{fig-RBUSE-R-fullPT}
\end{figure}

\begin{figure}[htbp]
\centering
\begin{tabular}{cc}
\includegraphics[width=2.5in,angle=0]
%{figs/residuals_Pos_qReg1_so_FALSE.pdf} &
{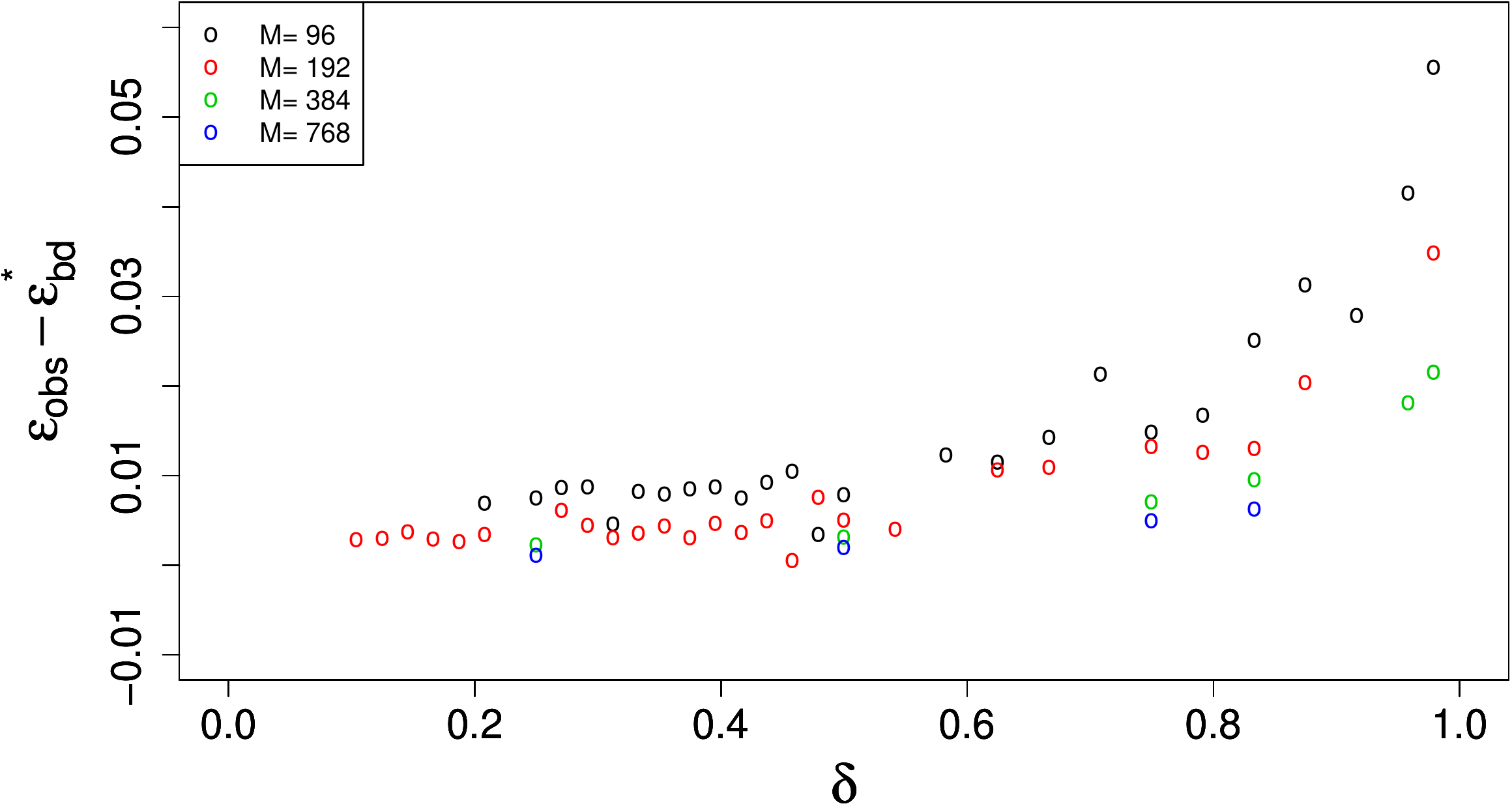} &
\includegraphics[width=2.5in,angle=0]
{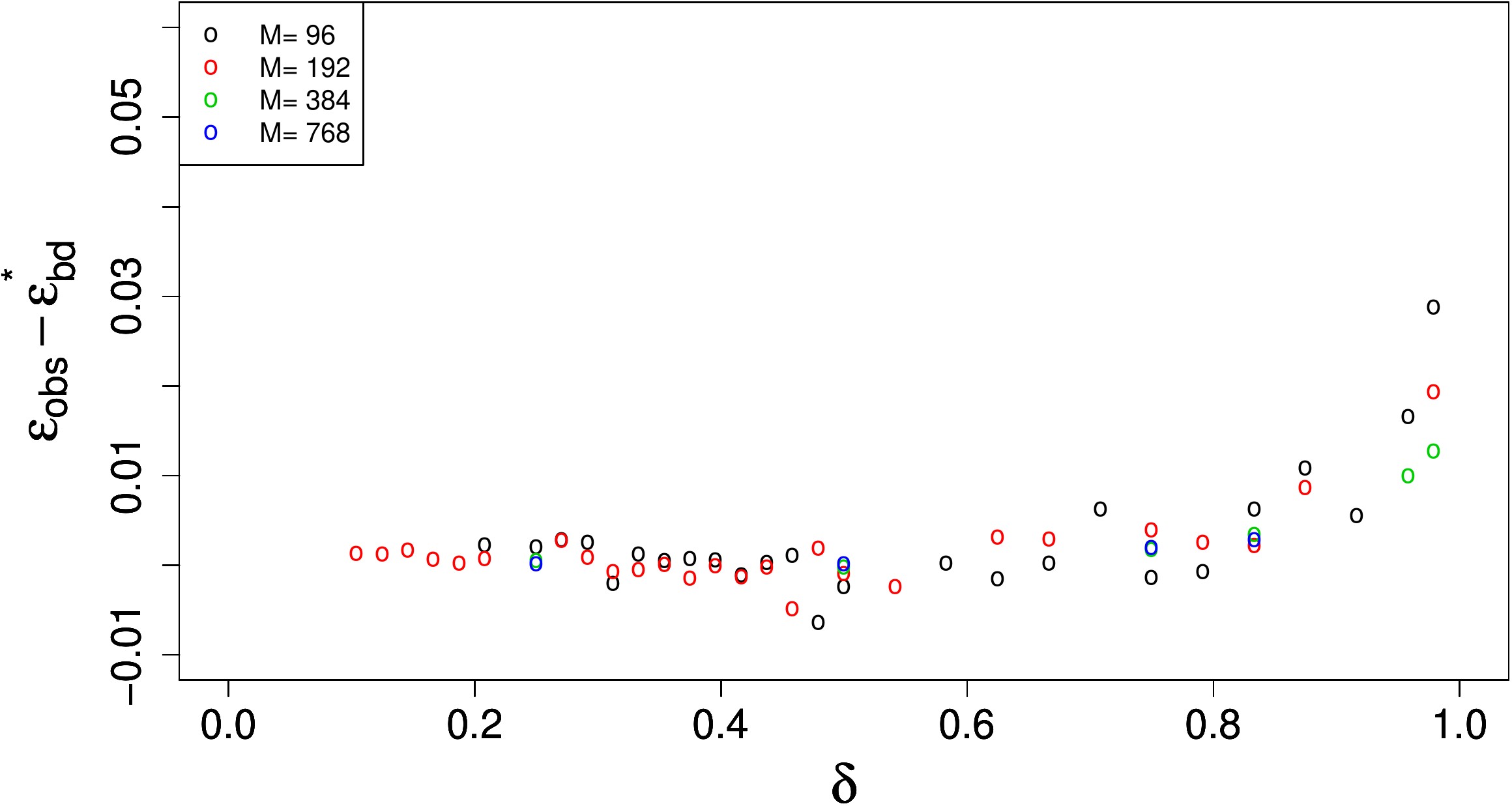}
%{figs/residuals_Pos_qReg1_so_TRUE.pdf}
\end{tabular}
\caption{Difference of experimental and predicted phase transition location for $\bX = \bC$ using first-order (left), and second-order (right) predictive models. Problem sizes $M =96,192,384$ and $768$.
Residuals are larger near $\delta \approx 1$, 
the residuals at $M=768$ and large $\delta$ are noticeably smaller than 
those at $M=96$.}
\label{fig-RBUSE-C-error}
\end{figure}

%\end{itemize}
%vasanawala2010improved,zhao2012image,zhang2015clinical,tamir2016t2,feng2016compressed
\clearpage
\section{Stylized Application to MR Imaging}

Numerous researchers \cite{vasanawala2010improved,zhao2012image,zhang2015clinical,tamir2016t2,feng2016compressed} 
have been conducting MR imaging experiments
where one dimension is sampled exhaustively
and the others are sampled at random, and in some cases
uniformly at random exactly as discussed here;
see for example  \cite{SparseMRI}. 

Theorem \ref{thm:RUID} shows that 
the $2D$ Fourier imaging with anisotropic 
undersampling is equivalent to block-diagonal 
measurements with $B = M$ and $N = M^2$. 
This equivalence is illustrated in Figure \ref{figure-2d-fft-vs-tiuse}. 
As expected, the empirical phase transition of the anisotropically-undersampled
2D FT is substantially below the transition point for
Gaussian measurement matrices. 
%At the same time, the transition point 
%matches that for the block-diagonal `RBRealDFT' matrix, 
%which consists of $2M$ repeated blocks of size $M\times 2M$.
%Each repeated block is the real-valued realization of 
%a partial 1-$d$ complex Fourier transform matrix in 
%which each complex entry is replaced by its equivalent 
%$2\times 2$ real matrix. 
%The figure also documents performance
%with block-diagonal measurements where,
%in place of each partial Fourier Matrix  we
%insert a random USE matrix,
%The results are similar; this is an instance of the 
%universality phenomenon discovered in \cite{DoTa08,MoJaGaDo2013}. 

The $2D$ imaging situation corresponds 
to the case where $\bX=\bC$
and $\gamma_M = \sqrt{2 \log(M)/M}$.
In the case of of $\| \cdot \|_{1,\bC}$ minimization 
our formulas give the following
offset between the asymptotic phase transition 
and the finite-$N$ transition:
\begin{eqnarray*}
  \mbox{offset}  
  &\sim&  \alpha_{\bC} \cdot \eta(\delta ; \bC) \cdot \gamma_M + \beta_{\bC} \cdot \zeta(\delta; \bC) \cdot \gamma_M^2  \\
  &=&  \delta^{-1/2} [  \frac{2}{3} \gamma_M - \frac{1}{3} \gamma_M^2]   \\
  &=&  \delta^{-1/2} [\frac{2\sqrt{2}}{3} \sqrt{\frac{\log(M)}{M}}  -\frac{2}{3} \frac{\log(M)}{M} ] .
\end{eqnarray*}
The experiments  reported here 
validated the formalism's predictions,
which can thus be used to gauge the amount of undersampling
required in $2D$ imaging experiments.

%\comment[HM]{Citation needed below}
Lustig and Pauly \cite{SparseMRI} also proposed 
anisotropic undersampling for $3D$ 
MR imaging,  
where one dimension is acquired  
exhaustively and the other two are acquired 
uniformly at random  \cite{SparseMRI}. 
Our formalism
applies to 3D MR imaging, 
where $\bX=\bC$, $B=M$, $N = M^3$,
and $\gamma_M = \sqrt{2\log(M)/M^2}$.
\begin{eqnarray*}
 \mbox{offset}  
  &\sim&  \alpha_{\bC} \cdot \eta(\delta ; \bC) \cdot \gamma_M + \beta_{\bC} \cdot \zeta(\delta; \bC) \cdot \gamma_M^2  \\
  &=&  \delta^{-1/2} [  \frac{2}{3} \gamma_M - \frac{1}{3} \gamma_M^2]   \\
  &=&  \delta^{-1/2} [\frac{2\sqrt{2}}{3} \frac{\sqrt{\log(M)}}{M}  -\frac{2}{3} \frac{\log(M)}{M^2} ].
\end{eqnarray*}
The leading term involves $1/M = 1/{N^{1/3}}$ in the $3D$ case, replacing the leading term $1/\sqrt{M} = 1/N^{1/4}$ from the $2D$ case. 

Note: a referee has emphasized that the model of sparsity 
entertained here is appropriate for images that look like hot spots scattered at random.
This might be appropriate for imaging with contrast agents.
Further work should study other image models and consider finite-$N$ phase transition phenomena they induce;
see also Section \ref{sec:Limitations}.

%Figure \ref{NUS-RBUSE} depicts a special case of this theorem 
%\newcommand{\bZ}{{\bf Z}}

\section{Stylized Application to MR Spectroscopy}

Jeffrey Hoch and collaborators have used anisotropic
random undersampling in multi-D NMR spectroscopy for more than two decades \cite{schmieder93}.
In MR Spectroscopy, anisotropic undersampling
is not the full story; %; there are two additional
%considerations: 
%\bitem
%\item {\it Hypercomplex nature of object $\bx$}.
we must also consider the %{\it Hypercomplex nature of object $\bx$}.
{ Hypercomplex nature of object $\bx$}.

A $d$-dimensional experiment collects measurements on
an array $x_0$ indexed by a $d$-dimensional grid of size $T_0\times \cdots \times T_{d-1}$, and having hypercomplex entries.
Each hypercomplex entry is a $2^d$-dimensional vector over the real field $\bR$
\footnote{Traditionally, the complete set of measurements in MR spectroscopy
is a set of $2^{d-1} \cdot \left( \prod_{1 \leq i < d} T_i \right) $ FID's;
different FIDs are indexed by $(\ell; k_1,\dots,k_{d-1})$.
Each FID  $F^{\ell}_{k_1,\dots,k_{d-1}}(\cdot)$ 
is a complex-valued time series $(F^{\ell}_{k_1,\dots,k_{d-1}}(k_0): 0 \leq k_0 < T_0)$
and measures two real coordinates of the hypercomplex entry associated with
 site $(k_0,k_1,\dots,k_{d-1})$ as $k_0$ varies,
effectively sampling along an axis-oriented line $u \mapsto (u,k_1,\dots,k_{d-1})$ in $\bZ_{T_0} \times \cdots \times \bZ_{T_{d-1}}$. 
Traditional full acquisition requires $2^{d-1}$ 
full passes along each line, 
 each pass \-- indexed by $\ell=0,1,\dots , 2^{d-1}-1$  \-- measuring a different
pair of coordinates of the full $2^d$-dimensional entry associated with
a given site.  In effect, the full $d$-dimensional hypercomplex transform
$\cF_{d,\bH^d} (x_0)$ is obtained at element 
$(k_0,k_1,\dots, k_{d-1})$ by gluing together the FID's
\[
\hat{x}(k_0,k_1,\dots, k_{d-1}) = \left (re(F^{0}),im(F^{0}),re(F^{1}),im(F^{1}),\dots,
  re(F^{2^{d-1}-1}),im(F^{2^{d-1}-1}) \right),
\]
where, in this display, $F^{\ell} \equiv F^{\ell}_{k_1,\dots,k_{d-1}}(k_0)$.}.

%A recent innovation in MR Spectroscopy \cite{HochRPD2011PNAS}
%is to make only one pass and, at each site visited
%measure a random projection of each $2^d$-dimensional vector;
%the object is then reconstructed by min $\ell_1$ or maximum entropy (MaxEnt). 
%This innovation was called {\it Random Phase Detection} \cite{HochRPD2011PNAS} in case the projection was 1-dimensional; and  {\it Partial Component Sampling} in case the projection was more general \cite{Schuyler13};
%see also \cite{Schuyler15,Monajemi2016ACHA},
%\eitem

In NMR spectroscopy, anisotropic undersampling is generally called {\it NUS}
(for {\it non-uniform sampling}) \cite{schmieder93}.
To carry it out, simply sample uniformly at random  
without replacement from the set of  $d-1$ tuples $(k_1,\dots,k_{d-1})$ 
and then collect $2^{d-1}$ FIDs at each such tuple \-- $\ell = 0,\dots, 2^{d-1}-1$ \-- 
i.e. collecting each $F^{\ell}_{k_1,\dots,k_{d-1}}(\cdot)$ associated to each selected tuple.
Theorem \ref{thm:RUID} can be generalized as follows, although we omit details in this article.

%at the expense of complicating our discussion somewhat.
%Suppressing details here, the resulting procedure
%acquires measurements of the form $y = A x$, where 
%$A$ is an end-to-end real acquisition matrix of size $n \times N$,
%where $N = 2^{d}\prod_{i=1}^{d}  T_i$. 
%This matrix is the net result of a pipeline 
%of simple operations which are
%linear over the real field. 
%Let $T^{(d)} \equiv  {Z_{T_1}\times \cdots Z_{T_d}}$ 
%denote the $d$-dimensional grid of interest.
%The pipeline $A = \cS\circ \cC \circ \cF_d \circ \cC^{-1} $
%involves a d-dimensional Hypercomplex Fourier transform 
%$\cF_d: (\bH_{d})^{T^{(d)}} \mapsto (\bH_d)^{T^{(d)}}$ 
%composed with real coordinatization 
%($\cC : (\bH_d)^{T^{(d)}} \mapsto \bR^{ 2^d \times T^{(d)}}$) 
%and selection 
%($\cS: \bR^{ 2^d \times T^{(d)}} \mapsto \bR^n$) operators.
%The data selection operator $\cS$ both selects 
%tuples $ (k_0; k_1,\cdots,k_{d-1}) \in \cT$,
%and selects specific projections of the
%$2^d$-dimensional-object available at each
%tuple $t = (t_1,\cdots,t_d) \in \cT$.  For more 
%details on our notation and setup see
%\cite{Monajemi2016ACHA}.

%\newcommand{\cK}{{\cal K}}
Let $N=\prod_{i=0}^{d-1} T_i$ 
and suppose that $n$ is divisible by $ T_0$. From the collection
of  $d-1$ tuples $( k_1,\dots, k_{d-1})$ in $\bZ_{T_1} \times \cdots \times \bZ_{T_{d-1}}$,
sample uniformly at random $m = n/T_0$ such tuples; 
and let $\cK$ denote the resulting set of selected tuples.
Let $\cS_d \equiv \cS_{d,T_0,\cK}$ denote the selection operator that, from 
a full array indexed by $d$-tuples $( k_0, k_1,\dots, k_{d-1})$
selects all the elements with indices in the product 
set $\{0,\dots,T_0-1\} \times \cK$.
In this setting, let $\cFu$ denote the 
linear operator defined by the 
pipeline $\cFu = \cS_{d,T_0,\cK} \circ \cF_{d, \bH^d}$.
 
For comparison, let $\cS_{d-1; \cK}$ denote a selection operator 
on $(d-1)$-dimensional arrays indexed by $d-1$ tuples $( k_1,\dots, k_{d-1})$ in $\bZ_{T_1} \times \cdots \times \bZ_{T_{d-1}}$. It selects just those entries with indices in $\cK$.
Let $\cF_{d-1,\bH^d}$ denote the $d-1$-dimensional discrete Fourier transform
{\it with scalars in the associative algebra $\bH^d$} (and {\it not} $\bH^{d-1}$),
and let $A^{(1)}$ denote the $m \times M$ matrix with 
$\bH^{d}$-valued entries representing the linear operator 
$\cS_{d-1;\cK} \circ \cF_{d-1, \bH^{d}}$.
Construct the $n \times N$ block-diagonal matrix  $A$
with $B = T_0$ identical blocks $A^{(1)}$;  
each block is a short fat  matrix with hypercomplex $\bH^d$ entries
of size $m \times M $ with 
$m = n/T_0$ and $M=N/ T_0$.

\begin{thm}{\bf (Multi-D NUS)} \label{thm:MultiD-NUS}   
Suppose that the same set $\cK$ of $d-1$ tuples
is used in defining both of the above-mentioned
selection operators $\cS_{d-1,\cK}$ and $\cS_{d,T_0,\cK}$.
Let $x_0$ be a hypercomplex array, and $\bx_0 = vec(x_0)$. 
The following two problems have identical values and isomorphic
solution sets:
$$
 (P_{1,\bX}^{\aus}) \quad \min \| x \|_{1,\bX} \quad \text{subject to} \quad \cFu(x)= \cFu(x_0),
$$
$$
{ (P_{1,\bX})} \quad \min \| \bx \|_{1,\bX} \quad \text{subject to} \quad A \bx= A \bx_0.
$$
{ where $\bX \in \{[0,1],\bR_+,\bR,\bH^d\}$ define the choice of the $\ell_1$ norm.}
\end{thm}
\begin{proof}
See Appendix \ref{app:NUS}.
\end{proof}

%The theorem below gives precise equivalence between partial-component 
%ampling schemes \cite{Schuyler13,Schuyler15} and block-diagonal matrices in the most general form.  

%\begin{thm}{\bf (multi-D Quadrature PCS)} \label{thm:Quadrature-PCS}   Suppose that the sampling set $\cT \subset T^{(d)}$, $q$ dimensions $D_i \text{\ \ for \ } i \in Q$ are sampled exhaustively and the other $(d-q)$ dimensions $D_i \text{\ \ for \ } i \in D\backslash Q$ are sampled incompletely.  There is a corresponding 
%block-diagonal matrix $A$
%with $B = (2^q \prod_{i\in Q} T_i)$ identical blocks, each block is a short fat real matrix of size $m \times M $ with $M=2^{d-q} \prod_{i\in D\backslash Q} T_i$.
%The following two problems have identical solution sets:
%$$
%{\rm (ANISO-PCS)} \quad \min \| x \|_{1} \quad \text{subject to} \quad Ax=y,
%$$
%$$
%{\rm (BD)} \quad \min \| x \|_{1} \quad \text{subject to} \quad Gx=y,
%$$
%\end{thm}
%
%\begin{proof}
%See the Appendix \ref{app:Q-PCS}.
%\end{proof}
%This equivalence is illustrated in Figures \ref{fig-NUS-RBUSE} and \ref{fig-RPD-RBUSE} for two special cases of PCS: (i) RPD in which only 2 out of 4 hypercomplex components are measured by quadrature acquisition and (ii) NUS in which all 4 components are measured \footnote{Note that NUS is a special case of partial component sampling.}. The undersampling is anisotropic meaning all the grid points along the the acquisition dimension $t_{2}$  is measured for a subset of the parametric time dimension ($t_{1}$).
{
As an example, Figure \ref{fig-NUS-RBUSE} shows such equivalence for $\bX=R$ and the special case of NUS in 2D experiments (with hypercomplex FIDs) when only half of the indirect times are sampled (i.e., $\delta=1/2$). Here, the equivalent block diagonal matrix in Theorem \ref{thm:MultiD-NUS} 
consists of $B=2T$ repeated blocks of size $m \times M = T\times 2T$. 
Each repeated block is a real-valued matrix implementing 
a partial 1D complex discrete Fourier transform 
(in this representation, each complex entry in the complex DFT
matrix is replaced by its equivalent 
$2\times 2$ real matrix). We label this block diagonal matrix by `RBRealDFT'. 
The figure documents the equivalence of NUS with  `RBRealDFT'. 
The figure also documents performance
with the RBUSE block-diagonal measurement matrix where,
in place of each partial Fourier matrix  we
insert a random USE matrix.
The results are similar; this is an instance of the 
universality phenomenon discovered in \cite{DoTa08,MoJaGaDo2013}. For more results of this kind, see \cite{Monajemi_thesis_2016}.
}

%At the same time, the transition point 
%matches that for the block-diagonal `RBRealDFT' matrix, 
%which consists of $2M$ repeated blocks of size $M\times 2M$.
%Each repeated block is the real-valued realization of 
%a partial 1-$d$ complex Fourier transform matrix in 
%which each complex entry is replaced by its equivalent 
%$2\times 2$ real matrix. 

\begin{figure*}[t] 
\centering
\includegraphics[width=3in]
{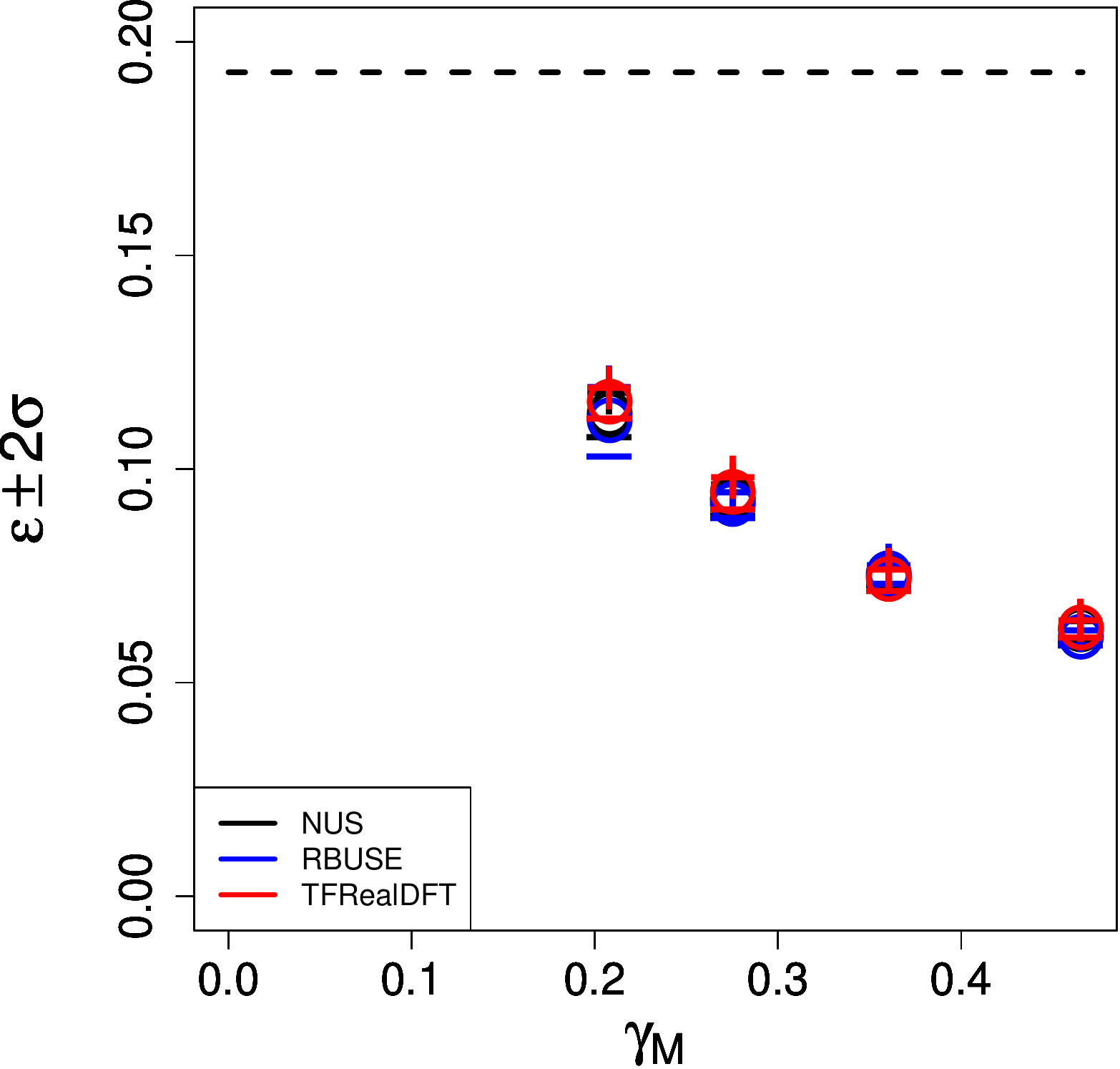}
%{figs/PT-NUSA2-TFUSE-TFRealDFT}
\caption{Equivalence of standard undersampling schemes in  2D
MR Spectroscopy with anisotropic undersampling schemes involving block-diagonal measurements.
The data come from a $T \times T$ Cartesian grid, 
amounting to  $4T^2$ real coefficients.
The undersampling fraction $\delta=1/2$. 
`NUS' models undersampling of the indirect dimension in MR Spectroscopy, 
selecting at random half of the $y=$ constant lines to be measured. 
`RBRealDFT' corresponds to block-diagonal measurements with $B=2T$ repeated blocks of size 
$T\times 2T$, each block the real representation of a partial $1D$ complex Fourier matrix, 
in which each complex entry is replaced by its equivalent $2\times 2$ real matrix.
`RBUSE' corresponds to block-diagonal measurements with $2T$ identical blocks, 
each block an $T \times 2 T$ real USE matrix. 
The dashed line represents the asymptotic Gaussian phase transition at $\delta=1/2$.  Problem sizes: $N=16,32,64,128$. }
\label{fig-NUS-RBUSE}
\end{figure*}

%\begin{figure*}[t]
%\centering
%\includegraphics[width=3in]
%{PDF/Figure13}
%%{figs/PT-RPD-TFUSE-TFRealDFT}
%\caption{
%Equivalence of standard undersampling schemes in MR Spectroscopy with anisotropic undersampling schemes
%involving  hypercomplex Fourier transform.
%Here the undersampling fraction $\delta=1/2$. 
%`RPD' corresponds to a model of undersampling
%in MR spectroscopy, acquiring an $M \times M$ image with $4M^2$ real coefficients in which only 2 out of 4 hypercomplex components are sampled by quadrature detection at every indel. `RBUSE' corresponds to block-diagonal measurements with $2M$ identical blocks, each an $M\times 2M$ USE matrix. RBRealDFT corresponds to block-diagonal measurements with $2M$ repeated blocks of size $M\times 2M$, each  the real representation of a partial Fourier matrix in which each complex entry is replaced by its equivalent $2\times 2$ real matrix. The dashed line represents the asymptotic Gaussian phase transition at $\delta=1/2$.  Problem sizes: $N=16,32,64,128$.}
%\label{fig-RPD-RBUSE}
%\end{figure*}

%We observe similar results for RPD experiments shown in Figure \ref{fig-RPD-RBUSE}, in which only 2 out of 4 hypercomplex components are measured by quadrature acquisition. %Motivated by the equivalnce of phase transition as given by Theorem \ref{thm:Quadrature-PCS}, we will study the phase transition for block-diagonal matrices in the next section.

\section{In $d>2$, how many dimensions to randomly undersample?}

We plan a full report on
the multidimensional case elsewhere,
documenting the accuracy of the 
prediction formalism developed here.
The key points can already be seen.
%Theorem \ref{thm:RUID} generalizes immediately to  
%higher-dimensional settings. Accordingly,
Suppose the object of interest is a $d$-dimensional
array, with sidelength ${T}$ on each axis,
so the total data volume $N = {T}^d$. 
In anisotropic undersampling of such an array,
let $B = {T}^{d_e}$ where $d_e$ 
is the number of exhaustively sampled dimensions; 
the individual blocks themselves 
are then of dimension $M = {T}^{d_r}$,
where $d_r = d - d_e$.
Our {\it ansatz} for the
location of the finite-N phase transition 
in Lemma \ref{Lemma-PT-MultiBlock}
translates to this special case as follows:
% For a fixed problem size $N = {T}^d$, increasing the number of exhaustively 
% sampled dimensions reduces the size of individual blocks while increasing the number of blocks. 

\begin{cor} ({\bf General $d \geq 2$}). 
\label{Lemma-PT-MultiBlock-General}
Let $d \in \{ 2,3, 4, \dots \}$,  and fix $0 \leq d_e, d_r  \leq d$,
such that $d_r+d_e=d$.
Consider a sequence of problem sizes ${T} \goto \infty$
and associated block-diagonal matrices  with
$B={T}^{d_e}$ blocks of equal size $m \times M$,
and with  $m/M \goto \delta \in (0,1)$ . 
For the offset between the 
asymptotic phase transition $\easy(\delta; [0,1])$
and the multi-block finite-$N$ phase transition { $\emb(m,M,B;[0,1])$ we have:
\[
  \easy(\delta) -  \emb  \sim \sqrt{ \frac{ 4 \frac{d_e}{d}(1-\delta) \log(N)}{N^{d_r/d}} }  , 
  \qquad  N \goto \infty.
\]
}%blue
\end{cor}

Two comments are in order:

\bitem
\item Comparing two  schemes  with  equivalent $n/N$ and $N$, 
but different $d_e$, we see that this gap is
increasing in the quantity $d_e/d$.  
In words: other things being equal, the gap
is larger  when there are more exhaustively sampled dimensions and  hence
fewer randomly sampled dimensions.
\item Comparing two problems with the same $d_e$ but different $d$,
we see that the gap is relatively less important when $d$ is larger.
For example, in multidimensional MR spectroscopy, 
the gap between the asymptotic Gaussian-measurements phase transition
and the finite-N phase transition is larger in smaller dimensions $d$ than 
in larger dimensions. The order of the gap  in $2D$-MRI -- where $d_e=1$, $d_r=1$ --
 is $O( \sqrt{\log(N)} N^{-1/4})$; while in $3D$ MRI -- where $d_e=1$, $d_r=2$ -- it is $O( \sqrt{\log(N)} N^{-1/3})$.
\eitem

\section{Limitations of Our Work}
\label{sec:Limitations}
{
There are several ways this study has been more limited than we would like.
Here are some possible variations and extensions:
\begin{enumerate}
\item {\it Pixel sparsity}. We have considered here only situations where
the object of interest is sparse in the original pixel/voxel domain. This is
a very specific assumption, and makes most sense for NMR spectroscopy when the
exponential decay times are very long. It also makes sense for MR Imaging
with contrast agents where we are looking for relatively rare `hotspots'.

\item {\it  Transform Sparsity. }
Referees suggested that sparsity in a transform basis would be more general and
more widely applicable. We leave this for further work, expecting that 
results of the precision we have been deriving 
here would require very specific assumptions.
\item {\it   Uniform Sparsity,}
Referees suggested that non-uniform sparsity -- i.e. different amounts of sparsity in
different blocks -- would be more general and more applicable. 
We agree, and have performed extensive
experiments where the sparse signals are scattered randomly,
leading to multinomial counts in the different blocks.
We also developed theory for the multinomial case, showing how
the first- and second- order correction terms will change.
Those terms are somewhat different than before,
however, the larger point remains the same: there are precise corrections
of order $polylog(M)/ \sqrt{M}$ which we can predict accurately.

We remind the reader that the regular case here is extremal --
other nonuniform sparsity cases will have phase transitions
that are even lower than this one. On the other hand, our analysis of block-diagonal systems
in Section 2 above suggests that if there are
dramatic differences in nonzeros from one column to another,
what really matters is the maximal number of nonzeros in
any column.
\end{enumerate}
All these directions of extension seem worth pursuing.

Finally we remind the reader of the existing theoretical work on
block diagonal undersampling \--  Eftekhari et al. \cite{Eftekhari2015}
and by Adcock and Chun \cite{Chun16} \--  which, taking a large-$N$ viewpoint
and thereby viewing $\log N$ factors as relatively inconsequential, explicity claims that
anisotropic undersampling is effectively just as good as dense Gaussian undersampling.
The deviations from the asymptotic model that 
we exhibit in Figure \ref{fig-NUS-RBUSE} above, and which seem practically consequential
to us, would be considered {\it de minimis} from that theoretical viewpoint.
}
\section{Conclusion}

We formalized the notion of anisotropic undersampling in multi-dimensional
Fourier imaging, and showed its mathematical equivalence with the use of 
block-diagonal measurement matrices in compressed sensing. 

We rigorously analyzed
a special case of block-diagonal measurement matrices where the object of interest 
has real coefficients bounded between $0$ and $1$ and typically at
the extreme values $0$ and $1$, and derived a precise
expression for the finite-$N$ phase transition, finding it to be
 displaced substantially from the large-$N$
phase transitions applicable fully dense Gaussian measurement schemes. 
% Indeed the offset is so great that at practical-sized problems
% the asymptotic Gaussian predictions are simply not relevant.

Massive computational experiments involving millions of CPU hours
established the empirical equivalence of random anisotropic 
Fourier undersampling with block diagonal Gaussian measurements.
The experiments showed that the phenomenon 
of substantial finite-$N$ phase transition offset from the fully dense
Gaussian measurement case \-- proven theoretically in the above special case \--
continues to hold empirically across a range of other settings, 
including the recovery of sparse objects with real coefficients, 
with real nonnegative coefficients or complex coefficients. 
The experiments allowed us to validate precise
formulas for the finite $N$-phase transitions adapted to all those cases, including 
second-order [in $\gamma$] versions of our formulas matching the experimental data closely. 

We presented formulas for the location of finite-$N$ phase transitions
in 2D and 3D Sparse MRI, where anisotropic undersampling has
a long history and has been extensively used. We briefly discussed
multi-dimensional MR spectroscopy, which involves
anisotropic undersampling of the hypercomplex Fourier transform,
and we empirically demonstrated its equivalence to block-diagonal
Gaussian measurements in the 2D hypercomplex case.
We left detailed discussion of the 
multidimensional hypercomplex  case
for future work.

\section*{Reproducible Research}
The code and data that generated the figures in this article may be found online at \url{https://purl.stanford.edu/th702qm4100} \cite{DataMonajemiAniso2017}.
\appendix

\section{Proof Sketches for Lemmas \ref{Lemma-PT-MultiBlock} and \ref{Lemma-PT-MultiBlock-General}} \label{app:PT-MultiBlock}

As the reader will see, the proof is mostly an exercise in manipulating properties of the Binomial distribution and
its normal approximation. 

\subsection{The Single-Block Problem}
The critical $\ell$-value $\ell^*(m,M)$ for the single block problem solves
\[
   Q_{sb} (\ell^*,m,M) \approx q^*.
\]
Namely, for fixed $q^*$ we find adjacent integers $\ell_\pm$ so that
$Q_{sb} (\ell_+,m,M) \geq  q^*$ and $Q_{sb} (\ell_-,m,M) \leq  q^*$.
Then $\ell^* \in \{ \ell_-,\ell_+\}$.

Recall that $Q_{sb} (\ell,m,M) = 1 - P_{M-m,M-\ell}$ with  $P_{M-m,M-\ell}$
a binomial probability defined  in Theorem \ref{thm-Donoho-Bnd}.
Since we are in the single-block problem, we take $q^* = 1/2$ as explained in Definition
\ref{dfn:q}. Hence we are trying to solve for the $\ell_\pm$ achieving
\[
P_{M-m,M-\ell_\pm} \approx \frac{1}{2}.
\]

The binomial probability $P_{k,n}$ is decreasing as $n$ increases for fixed $k$.
Moreover if $n$ is even, then $P_{n/2,n} = 1/2$ exactly.
%, while if $n$ is odd, then
%$P_{\lfloor n/2 \rfloor, n} = 1 - P_{\lceil n/2 \rceil, n} \sim \frac{1}{2} - \frac{1}{\sqrt{\pi n}} $.
We conclude that when $M-\ell = 2(M-m)$ we will exactly solve $P=1/2$. We of course do this by setting
%The binomial distribution is approximated by a Gaussian distribution for suitably large 
%problem sizes. Operating formally, we have
%\begin{align}
%\nonumber P_{M-m,M-\ell}  & \approx  \Phi \Big( \frac{(M-m) - \left(\frac{M-\ell}{2}\right)}{\frac{\sqrt{M-\ell}}{2}} \Big) & \\
%& = \Phi \Big(\frac{M-2m+\ell  }{\sqrt{M - \ell}} \Big) . &
%\label{eq-binomal} 
%\end{align}
%Since we are in the single-bloick problem,
%we take $q^* = 1/2$ and noting that $\Phi(0) = 1/2$,
%a continuum approximation  $\ell_0$ to the
%finite-$N$ phase transition location, say, is found by solving,
%$$
%\frac{M-2m+\ell_0  }{\sqrt{M - \ell_0}} =  0,
%$$
%whence
\[
   \ell^*  = 2m -M.  % \left \{ \begin{array}{rr} 2m-M & M \mbox{ even} \\ 2m-M -1 & M \mbox{ odd} \end{array} \right . .
\]
%Assuming there is a $C$ so that $|\ell_{\pm} - \ell_0| \leq C$ for all large $m,M$,
Then from $\delta \sim m/M$, we get 
\[
   \ell^*/M \sim  (2 \delta-1) ,
\]
and so 
\[
\esb(m,M; [0,1]) \mapsto (2\delta-1)_+,
\]  
as $m,M \goto \infty$ with 
$m/M \goto \delta$.

\subsection{The Multi-Block Problem}

The critical value $\ell^*$
for the multi-block problem and regular sparsity is given by

$$Q_{mb} (B\cdot\ell^*,B\cdot m, B \cdot M) = Q_{sb} (\ell^*,m,M)^B \approx q^*.$$
 Namely, either $Q_{mb} (B\cdot\ell^*,B\cdot m, B \cdot M)$ just barely
exceeds $q^*$ but $Q_{mb} (B\cdot(\ell^*+1),B\cdot m, B \cdot M)$ does not, or else 
$Q_{mb} (B\cdot\ell^*,B\cdot m, B \cdot M)$ barely is below $q^*$ but $Q_{mb} (B\cdot(\ell^*-1),B\cdot m, B \cdot M)$
is not. Let $\ell_- \equiv \ell_-(B;m,M)$ and $\ell_+ \equiv \ell_+(B;m,M)$ 
denote the two adjacent integers just identified,
namely the smallest $\ell$ where $Q \leq q^*$ and the largest $\ell$ where $Q \geq q^*$.
Then $\ell^* \in \{ \ell_-,\ell_+\}$.

We are interested in the setting where the number of blocks $B \goto \infty$; since $q^*$ is fixed,
(e.g. at $1-1/e$), it follows that the success probability for 
individual blocks obeys
\begin{equation} \label{eq:Qsandwich}
Q_{sb} (\ell_+,m,M) \geq (q^*)^{1/B} \geq Q_{sb} (\ell_-,m,M). 
\end{equation}
The last display shows that $Q_{sb} (\ell_+(B;m,M),m,M)$ tends to 1 as $B$ increases.
However, by standard properties of the Binomial probability mass function and 
the fact that $\ell_- - \ell_+ = 1$,  we also have 
$Q_{sb} (\ell_-(B;m,M),m,M) \goto 1$.
We conclude that the failure  probability for individual blocks, 
$P_{M-m,M-\ell_\pm}$, tends to zero.

We first operate purely heuristically to derive the would-be 
formula, which we then verify rigorously.
Taking logarithms of (\ref{eq:Qsandwich}), and recalling $Q_{sb} (\ell,m,M) = 1 - P_{M-m,M-\ell}$,
then from $- \log(1-p) \approx p$ for $p$ small,
we arrive at an approximation of the following form:
$$
P_{M-m,M-\ell_\pm } \approx  \frac{\log(1/q^*)}{B} .
$$
 The binomial distribution is approximated by a Gaussian distribution for suitably large 
 problem sizes:
 \begin{align}
 \nonumber P_{M-m,M-\ell_\pm}  & \approx  \Phi \Big( \frac{(M-m) - \left(\frac{M-\ell_\pm}{2}\right)}{\frac{\sqrt{M-\ell_\pm}}{2}} \Big) & \\
 & = \Phi \Big(\frac{M-2m+\ell_\pm  }{\sqrt{M - \ell_\pm}} \Big) . &
 \label{eq-binomal} 
 \end{align}
Now, let $z_B = \Phi^{-1}(\frac{\log(1/q^*)}{B})$. 
A continuum approximation  $\ell_0$ to the
finite-$N$ phase transition location, say, is found by solving,
$$
\frac{M-2m+\ell_0  }{\sqrt{M - \ell_0}} = z_B,
$$
which yields,
\beq \label{eq:continuumformula}
\ell_0 = 2m - M - \frac{1}{2} \sqrt{z_B^4 + 8 z_B^2 (M - m) } - \frac{z_B^2}{2}  .
\eeq
%One can show that $|\ell_\pm - \ell_0| \leq C$ for all large $m,M,B$, $M = B$.
Assuming $|\ell^* - \ell_0| \leq C$,
dividing both sides by $M$ and letting $\delta=m/M$:
\begin{eqnarray}
\nonumber \epsilon_{mb}^*(m,M,B,[0,1])=& 2\delta - 1 - \frac{1}{2} \sqrt{{\left(\frac{z_B^2}{M}\right)}^2 + \frac{8 z_B^2 (M - m)}{M^2} } - \frac{z_B^2}{2M} + O(\frac{1}{M}) &  \\
\nonumber  =& 2 \delta - 1 -  |z_B| \frac{\sqrt{2 (1-\delta)}}{\sqrt{M}} +O(\frac{z_B^2}{M})& \\
=& \epsilon_{sb}^*(m,M;[0,1]) - |z_B| \frac{\sqrt{2 (1-\delta)}}{\sqrt{M}} + O(\frac{z_B^2}{M}) .&
\end{eqnarray}

For $B$ large, we use the following classical approximation to $z_B$:
\[
|z_B|   =  \sqrt{2 \log(B)} \cdot (1 + o(1)), \quad B \goto \infty.
\]
Setting ${\gamma} = \sqrt{\frac{2 \log(B)}{M} }$ gives:
\begin{eqnarray}\label{eq:epsilon-tensor}
\epsilon_{mb}^*(m,M,B;[0,1])= \epsilon_{sb}^*(m,M;[0,1]) -  \sqrt{2 (1-\delta)}\gamma + o(\gamma).
\label{eq:bnd-offset}
\end{eqnarray}

To justify the above heuristic derivation rigorously, we need the following four lemmas,
which are stated in usual language familiar to probabilists.

\begin{lem} \label{lem:NormalApproxBinomial}
For $0 < k < n/2$,  let   $P_{k,n}$ be the usual binomial probability $2^{-n}  \sum_{h=0}^k {n \choose h} $
and let $\Phi_{k,n} \equiv \Phi( (2k - n)/\sqrt{n})$ be its usual normal approximation. We have
\beq \label{eq:uspensky}
      |P_{k,n} - \Phi_{k,n}| \leq  \frac{.26}{{n}} + e ^{-\sqrt{n}}.
\eeq
\end{lem}

This Lemma is effectively equation (4) in W. Feller's 1945 paper on Normal approximation to the Binomial;
he attributes this to Uspensky.

\begin{lem} \label{lem:PDropsExpon}
For $0 < k < n/2$,  again with    $P_{k,n}$ 
the usual binomial probability,
\beq \label{eq:prodrule}
      P_{k,n+h} \leq P_{k,n} \cdot 2^{-h} ( 1- k/n)^{-h}.
\eeq
\end{lem}

\begin{proof}
One computes the ratios $r_{k,n+h} = p_{k,n+h}/p_{k,n}$ of probability mass functions. 
Note that
\[
      r_{\ell,n+1}  = \frac{p_{\ell,n+1}}{p_{\ell,n}} = \frac{2^{-(n+1)} {n+1 \choose \ell}}{2^{-n} {n \choose \ell}} = \frac{1/2}{1- \frac{\ell}{n+1}}.
\]
Then from $r_{\ell,n+h} \leq 2^{-h}/(1-k/n)^{h}$ for $\ell \leq k$,
\[
      P_{k,n+h} = \sum_{\ell=0}^k p_{\ell,n+h} =  \sum_{\ell=0}^k p_{\ell,n} \prod_{g=1}^h r_{\ell,n+g}  \leq  \sum_{\ell=0}^k p_{\ell,n} \cdot 2^{-h} \cdot (1 - k/n)^{-h} = P_{k,n}  \cdot 2^{-h} \cdot (1 - k/n)^{-h} .
\]
\end{proof}
\begin{lem} \label{lem:ContinuumSolution}
Again let    $\Phi_{k,n}$  denote
the usual normal approximation to the binomial probability $P_{k,n}$.
Let $c>0$ be fixed and let $n_0(c,k)$ be the smallest {\em real value} satisfying 
\[
      \Phi_{k,n_0} \leq \frac{c}{k} , \qquad n_0 >  2k.
\]
Then  with $ 0 < c'  < c$ fixed, 
for some $C(c,c')$ made explicit below, 
\[
 \lim_{k_0 \goto \infty} \sup_{k \geq k_0} | n_0(c,k) - n_0(c',k)|/\sqrt{n_0(c,k)} \leq C.
 \]
\end{lem}

\begin{proof}
Now
\[
    \Phi_{k,n} = \Phi( \frac{ 2k-n}{\sqrt{n}} ).
\]
Let $z(c,k) = \Phi^{-1}(\frac{c}{k})$.
Then $n_0$ solves,
\[
  \frac{ 2k-n_0}{\sqrt{n_0}}  = z(c,k),
\]
and then $n_0 = 2k - z \sqrt{n_0}$ and so $\sqrt{n_0} =  (\sqrt{8 k + z^2} - z)/2$.
Now as $k \goto \infty$,  $z(c,k) = \Phi^{-1}(\frac{c}{k}) = -\sqrt{2 \log(k)} (1 + o(1))$ 
tends to infinity, in such a way that to leading order it doesn't depend on $c$.
We can say more. Suppose we wish to compare $z(c',k)$ with $z(c,k)$
precisely for large $k$, where $c' < c$ are both fixed.  This is the same thing as comparing
$\Phi^{-1}(\alpha)$ with $\Phi^{-1}(\frac{c'}{c} \alpha)$ for small $\alpha$.
Consider the difference of these two quantities,
\[
    \Psi(\beta; \alpha) = \Phi^{-1}(\alpha) - \Phi^{-1}((1-\beta) \alpha),
\]
where we introduce  $\beta = 1-c'/c \in (0,1)$.

We compare this to the $\beta$-quantile of the  conditional distribution
of the random variable $Y = z_\alpha -Z$, where $Z \sim N(0,1)$
and $Y$ is conditioned on $Z < z_\alpha$,
where $z_\alpha \equiv \Phi^{-1}(\alpha)$.
The density of the random variable $Y$ has the exact form
$f(y; \alpha)  \propto \exp( - y ( |z_\alpha| + y/2) )$ on $y \geq 0$.
Each member of this family of densities is less dispersed than the half-normal density
$\propto e^{-y^2/2}$ on $y > 0$. 
Let $F^{-1}(\beta;\alpha)$  denote the
$\beta$-th quantile of $f(y; \alpha) $.
This stays in a bounded set as $\alpha \goto 0$:
\[
 \sup_{0<\alpha < 1/2} F^{-1}(\beta;\alpha) \leq \Phi^{-1}(1/2 + \beta/2), \qquad 0 \leq \beta \leq 1.
\]
In terms of this quantile, we have the identity 
\[
\Psi(\beta; \alpha) = F^{-1}(\beta;\alpha).
\]
Hence
\[
    \sup_{\alpha < c/k_0} |\Psi(1-\frac{c'}{c}; \alpha) | < \Phi^{-1}(1 -\frac{c'}{2c}),
\]
The function $G(z;k) \equiv (\sqrt{8 k + z^2} - z)/2$ obeys $\frac{\partial}{\partial z} G(z) \leq C_1$.
Since
\[
n_0(c,k) =  [G(z(c,k); k)]^2
\]
we have
\begin{eqnarray*}
 |n_0(c,k) - n_0(c',k)| & \leq & 2 G(z(c,k);k) |G(z(c,k);k) - G(z(c',k);k)| + (G(z(c,k);k) - G(z(c',k);k))^2\\
 & \leq & 2 \cdot G(z(c,k);k)  \cdot C_1 | z(c,k) - z(c',k)|  + C_1^2  \cdot |z(c,k) - z(c',k)|^2\\
 &=&   2 \cdot G(z(c,k);k) \cdot C_1 \cdot \Phi^{-1}(1-\frac{c'}{2c}) + C_1^2 \cdot \Phi^{-1}(1-\frac{c'}{2c})^2 
\end{eqnarray*}
Hence for large $k$,
\[
 |n_0(c,k) - n_0(c',k)|  / \sqrt{n_0(c,k)} \leq C_2(c,c'),
\]
where  $C_2(c',c) \equiv  1 + 2 \cdot C_1 \cdot \Phi^{-1}(1-\frac{c'}{2c} )$ and we are assuming $c' < c$.
\end{proof}
We combine these as follows. 

\begin{lem}
Fix $\delta \in (1/2,1)$ and  consider a sequence of tuples $(m,M)$ with
$M \goto \infty$, and $m \sim \delta M$; and set $k \equiv M-m$. 
Let $\nu_k \equiv \frac{.26}{k} + e ^{-\sqrt{k}}$
denote the error term in (\ref{eq:uspensky}).
Let $n_1 = n_1(q^*,k,M)$ solve
\[
     \Phi_{k,n_1} =  \frac{\log(1/q^*)}{M} + \nu_k.
\]
Let $c = \log(1/q^*) \cdot (1 -\delta)$.
For all sufficiently large $k$, 
\beq \label{eq:sandwichA}
      n_0(c+0.27,k) > n_1(q^*,k,M)  > n_0(c,k).
\eeq
Let $n_2 (q^*,k,M)$ denote the smallest integer solving
\[
   P_{k,n_2}  \leq  \frac{\log(1/q^*)}{M} .
\]
Then for $k_0$ sufficiently large, there is an $h = h(c,k_0) > 0$ fixed independently of $k > k_0$ so that
\beq \label{eq:sandwichB}
          n_1 \leq n_2 \leq n_1 +h.
\eeq
\end{lem}

\begin{proof}
We earlier gave the formula  $n_0(c,k) = [ (\sqrt{8k + z_0^2} - z_0)/2]^2$, in terms of $k$ and $z_0 = \Phi^{-1}(c/k)$.
By inspection, $n_0$ is monotone  decreasing in $z_0$.
Similarly, we have:
\[
n_1 = [ (\sqrt{8k + z_1^2} - z_1)/2]^2, \qquad z_1 = \Phi^{-1}(\frac{\log(1/q^*)}{M} + \nu_k).
\] 
Again $n_1$ is monotone  decreasing in $z_1$. Now we observe that for large $k$, $.26/k + \exp(-\sqrt{k}) < .27/k$. For such $k$,
\[
      z_0(c,k) \geq z_1(q^*,M,k) \geq z_0(c+ 0.27,k),
\]
and (\ref{eq:sandwichA}) follows.

Now note that by (\ref{eq:uspensky})
\beq \label{eq:sandwichC}
        \frac{\log(1/q^*)}{M} =\Phi_{k,n_1} - \nu_k \leq  P_{k,n_1}  \leq 
         \Phi_{k,n_1} +  \nu_k  = \frac{\log(1/q^*)}{M}  \cdot ( 1 +  M \cdot \nu_k/\log(1/q^*))   . 
\eeq
Now from $2k = n_1 + z_1 \sqrt{n_1}$, we have $(1-\frac{k}{n_1}) = \frac{1}{2} + \frac{| z_1|}{2 \sqrt{n_1}}$.
Hence $2^{-h} (1- k/n_1)^{-h} = (1+ |z_1|/\sqrt{n_1})^{-h}$. Picking 
 $h > 2$ a positive constant, and taking into account that  $M \nu_k = O(k \exp(-\sqrt{k})) = O(1/{k})$ while
$ |z_1|/\sqrt{n_1} >  1/\sqrt{k}$ (say)  for large $k$,
we get that along our sequence $k \sim (1-\delta) M$, we have for all sufficiently large $k_0$ (say) that 
\begin{eqnarray*}
   (1+ |z_1|/\sqrt{n_1})^{-h} \cdot ( 1 +  M \cdot \nu_k/\log(1/q^*))  & \leq & (1+ k^{-1/2})^{-1} \cdot ( 1 + o({k^{-1/2}}))  \\
   &<& 1, \qquad  {k > k_0}.
\end{eqnarray*}
Applying (\ref{eq:prodrule}) and (\ref{eq:sandwichC}),  we have for $k  > k_0$:
\[
   P_{k,n_1+h}  <   \frac{\log(1/q^*)}{M}  \leq  P_{k,n_1} .
\]
It follows that $n_1 \leq  n_2 \leq n_1+h$.
\end{proof}

We apply these lemmas to our problem, in which
 $\delta \in (1/2,1)$,
$m \sim \delta M$ with $M \goto \infty$.
Setting $\ell_+ = M - n_2(q^*,M-m,M)$ yields
\[
   P_{M-m,M - \ell_+}  \leq  \frac{\log(1/q^*)}{M},
\]
and that $\ell_+$ is the largest value of $\ell$ with this property.
For $c = \log(1/q^*) \cdot (1 -\delta)$,
the previous lemma gives
\[
      M-n_0(c,M-m)  \geq M-n_2 \geq M - (n_0(c+0.27,M-m) + h),
\]
while Lemma \ref{lem:ContinuumSolution} implies that the 
two sides differ by at most  a term
$\Delta(c,m,M) = C(c,c+0.27)  \sqrt{n_0(c,M-m)} + h(c,k_0)$.  
We immediately obtain that  $|\ell_\pm - ( M - n_0) | \leq \Delta(c,m,M)$,
and of course  by our definitions, $|\ell^* - \ell_+| \leq 1$.

Finally,  the identities $2k = n_0 + z_0 \sqrt{n_0}$  and $\sqrt{n_0} = (\sqrt{8k + z_0^2} - z_0)/2$
yield
\begin{eqnarray*}
   n_0(c,k) &=&  2k - (\sqrt{8k + z_0^2} - z_0)/2 \cdot z_0(c,k)\\
                &\sim&  2k + 2 \sqrt{ k \log( {k} ) } \cdot (1 + o(1)) \\
                 &=&  2k + 2\sqrt{k} \cdot  \sqrt{ \log( M )} \cdot (1 + o(1)) .
\end{eqnarray*}

%Let $n_2 = n_2(q^*,k,M)$ denote the smallest integer solving
%\[
%     \Phi_{k,n_2} =  \frac{\log(1/q^*)}{M} + \nu_k.
%\]
%By (\ref{eq:uspensky}) we have
%\[
%     P_{k,n_2}  \geq  \Phi_{k,n_2} - \delta_k = \frac{\log(1/q^*)}{M} .
%\]
%Note that, with $k = M-m$,  $M \cdot \nu_k = 0.26/(1-\delta) + M e^{-\sqrt{(1-\delta)M}}$.
%On the other hand, 
%choose $h = h(q^*,k,M)$ so that
%\[
%        \log(1/q^*)  \leq  2^{-h} ( 1- k/n_2)^h \cdot ( \log(1/q^*) +  M \cdot \nu_k) 
%\] 
%Applying Lemma \ref{lem:PDropsExpon} gives
%\[
%         P_{k,n_2+h } \leq   P_{k,n_2}  \cdot 2^{-h} ( 1- k/n_2)^h  <  \frac{\log(1/q^*)}{M}. 
%\]
%Let $n_3 (q^*,k,M)$ denote the smallest integer solving
%\[
%   P_{k,n_3}  \leq  \frac{\log(1/q^*)}{M} .
%\]
%We conclude that
%\[
%          n_2 \leq n_3 \leq n_2 +h.
%\]
%
%We consider the solution $n_0 = n_0(c,k)$,
%and note that it obeys 
%\[
%    \Phi_{k,n_0} = \frac{\log(1/q^*)}{M}  \cdot \frac{(1-\delta) M}{k} =  \frac{\log(1/q^*)}{M} ( 1+ o(1)).
%\]
%and that furthermore, if $n_3(q^*,k,M)$ solves
%
%then $| n_3 - n_0| \leq C$, for all large $k$ and $M \asymp k/(1-\delta)$.
% $n_2 = M- \ell_0$ etc.
%we recognize that $n_2(q^*,k)$ solves
%\[
%\Phi_{k,n_2} = \frac{ \log(q^*) }{n_2},
%\]
%say. However, Lemma A.3 shows that to leading order
%\[
%n_2(q^*,k) = n_0(\log(1/q^*),k)  + O(1), \qquad  k \goto \infty.
%\]
%Lemma A.4 shows that
%\[
%n_1(\log(q^*),k) =  n_0(\log(1/q^*),k)  + O(1), \qquad  k \goto \infty.
%\]
Combining the above formulas,
\begin{eqnarray*}
 \epsilon_{mb}^*(m,M,B; [0,1]) &=& \ell^*/M = (M- n_0)/M + O(\Delta/M) \\
 &=&  (M-2k)/M   -  \frac{\sqrt{2k}}{\sqrt{M}}  \cdot \frac{\sqrt{2\log(M)} }{\sqrt{M}} (1 + o(1))   + O(\Delta/M)\\
 &=&  (M-2(M-m))/M   -  \frac{\sqrt{2(M-m)}}{\sqrt{M}}  \cdot \frac{\sqrt{2\log(M)} }{\sqrt{M}}(1 + o(1)) + O(\Delta/M)\\
 &=&  (2\delta-1)    -  \sqrt{2(1-\delta)} \cdot \frac{\sqrt{2\log(M)} }{\sqrt{M}}   + o(\gamma)\\
 &=& \epsilon_{sb}^*(m,M;[0,1]) -  \sqrt{2 (1-\delta)} \cdot \gamma + o(\gamma),
\end{eqnarray*}
where we used $O(\Delta/M) = o(\gamma)$.

\section{First Proof of Theorem \ref{thm:RUID}}\label{app:RUID}

The proof of Theorem \ref{thm:RUID} relies on three lemmas.

\begin{lem}\label{rank-deficiency}
{\bf(Rank-deficient matrix)} Consider the rank-deficient $n\times N \ (n<N)$ measurement matrix $G$ with ${\rm rank}(G)=r<n$ and  $x_0 \in \bR^N$ generating measurements $b=Gx_0$. 
The minimum-$\ell_1$ optimization problem
$$
\min \|x\|_1 \quad \text{subject to} \quad G x = b
$$
has the same solution set as the reduced-dimensional problem 
$$
\min \|x\|_1 \quad \text{subject to} \quad A x = y,
$$
where $A$ is a full-row-rank matrix of size $r\times N$ and $y = Ax_0$ . 
\end{lem}

\begin{proof} Using the SVD $G=U \Sigma V^T$, where $U \in R^{n\times r}$, $V\in R^{N\times r}$, and $\Sigma\in R^{r\times r}$. Then

\begin{eqnarray*}
Gx &=& b, \\
\Sigma^{-1} U^T G x & = & \Sigma^{-1} U^T b, \\
V^T x &= & \Sigma^{-1} U^T b.
\end{eqnarray*}
Setting $A=V^T$ and $y=\Sigma^{-1} U^T b$ completes the proof.
\end{proof}

\newcommand{\cE}{{\cal E}}
\begin{lem}\label{gram-matrix}
{\bf(Block structure of Gram matrix of anisotropically  undersampled FT)} Consider a $d$-dimensional complex-valued array $x$ defined on  
a Cartesian grid of size $N=T_0 \times T_1\times T_2\times \dots\times T_{d-1}$. 
Let $\cD = \{0,1,2,\dots,d-1\}$ denote the possible indices of the different underlying
Cartesian axes. 
Further, let $\cE \subset \cD$ denote the indices of axes 
along which exhaustive samples are taken, and 
$\cP = \cD \backslash \cE$,  the remaining indices
which are sampled partially.  Then
$\cD = \cE \cup \cP$ and, with $d_{\cE} = \# \cE$ 
exhaustively sampled dimensions and
$d_{\cP} = \# \cP$ partially sampled dimensions, 
$d = d_{\cE} + d_{\cP}$.
Let the end-to-end measurement operator 
be represented by the $n \times N$ complex-valued matrix $A$. 
Then, the complex Hermitian Gram matrix $G = A^* A \in \bC^{N\times N}$ 
is block-diagonal with $\prod_{j\in \cE} T_j$ 
identical blocks each of size $\prod_{j\in \cP} T_j$.  

The corresponding result for real-valued $A$ and
real-valued symmetric $A'A$ also holds.
\end{lem}

\begin{proof}
Let $\cK \subset \bR^d$ denote the set of all tuples $k = (k_0,\dots,k_{d-1})$
that get sampled.
 Let $e_j$ denote the $j$-th standard 
unit basis vector, $j=0,\dots,{d-1}$, let $V_{\cE} \subset \bR^d$ denote 
the linear span of $\{ e_j, j \in \cE\}$, let
$\cK_{\cE} =  {\mbox{Proj}}_{V_{\cE}} \cK$ denote the orthogonal projection
of the sampled tuples on the (span of the) exhaustively sampled dimensions.
Correspondingly let  $V_{\cP} \subset \bR^d$ denote 
the linear span of $\{ e_j, j \in \cP\}$, let
$\cK_{\cP} =  {\mbox{Proj}}_{V_{\cP}} \cT$ denote the projection
of the sampled tuples on the (span of the) partially sampled dimensions.
Then  $\cK_{\cE}$ is, speaking informally, a Cartesian product of intervals.
Formally, for each index of an exhaustively sampled dimension $j \in \cE$,
let $\cK_j = \{0,\dots, T_j-1 \}$ denote the full range of that index.
Then $\cK_{\cE}$ is an orthogonal sum $\cK_{\cE} =  \bigoplus_{j \in \cE} \cK_j  \cdot e_j$
and $\cK$ itself is an orthogonal sum
\[
\cK =  \cK_{\cP}  \bigoplus \cK_{\cE},
\]
Informally, $\cK$ is an `irregular' set of indices $\cK_{\cP}$ 
`times' a Cartesian product $\cK_{\cE}$ , and its cardinality obeys the product formula:
$\# \cK = \# \cK_{\cP} \times \# \cK_{\cE}$.
A certain multiplicative
relation generalizes the product formula.
For each tuple $k \in \cK$, let $k_{\cP}$ denote the
projection $Proj_{V_{\cP}} k$ and similarly let $k_{\cE} = Proj_{V_{\cE}} k$.
For an expression $c(k)$  obeying the factorization
$c(k) = a(k_{\cE}) b(k_{\cP})$, we have
\beq \label{eq:multiplic}
   \sum_{k \in \cK} c(k) =  \sum_{k \in \cK}  a(k_{\cE}) b(k_{\cP}) =  [ \sum_{k_{\cE} \in \cK_\cE}  a(k_{\cE})] \cdot \sum_{k_{\cP} \in \cK_\cP} b(k_{\cP}) .
\eeq

The $(k,t)$ element of the  
Fourier matrix can be written
$$
F_k(t) = \frac{1}{\sqrt{N}} \exp\left \{ 2\pi \bi (\sum_{j=0}^{d-1}  k_j  t_j/T_j)\right \},
$$
where $t = (t_0, \dots, t_{d-1})$, and $k = (k_0, \dots, k_{d-1})$.
The inner product between two distinct columns $u$ and $t$ of $A$ is 
thus given by 
\begin{eqnarray*}
G_{t,u} &=&  (A^*A)_{t,u} = \sum_{k \in \cK} F_k(t) F_{k}^* (u)  \\ 
&=&  \frac{1}{N} \sum_{k \in \cK} \exp{  \left ( 2\pi \bi  \sum_{j \in \cD} k_j (t_j-u_j)/T_j  \right) } \\
&=& \frac{1}{N}  \sum_{k \in \cK} \exp{  \left(  2\pi \bi   \cdot 
       [ \sum_{j \in \cP} k_j (t_j-u_j)/T_j   +  \sum_{j \in \cE} k_j (t_j-u_j)/T_j ]  \right)}  \\
&=& \frac{1}{\prod_{j \in \cP} T_j}   
       \sum_{t \in \cT} \exp{ \left( 2\pi \bi  \sum_{j \in \cP} k_j (t_j-u_j)/T_j  \right) } 
       \times     \left[ \frac{1}{\prod_{j \in \cE} T_j}   \exp{ \left ( 2\pi \bi  \sum_{j \in \cE} k_j (t_j-u_j)/T_j \right ) } \right]  \\
&=& \frac{1}{\prod_{j \in \cP} T_j}   
       \sum_{k \in \cK_{\cP}} \exp{ \left ( 2\pi \bi  \sum_{j \in \cP} k_j (t_j-u_j)/T_j  \right) } 
       \times    \left[   \frac{1}{\prod_{j \in \cE} T_j} \sum_{t \in \cT_{\cE}} \exp{ \left ( 2\pi \bi  \sum_{j \in \cE} k_j (t_j-u_j)/T_j \right) } \right]  \\
\end{eqnarray*}
where we used  $N = \prod_{j \in \cE} T_j \cdot \prod_{j \in \cP} T_j$ as well as the multiplicative relation (\ref{eq:multiplic})
for the multiplicative expression $c(k) =  \exp{  \left ( 2\pi \bi  \sum_{j \in \cD} k_j (t_j-u_j)/T_j  \right) } = a(k_\cE) b(k_\cP)$
with $a(k_\cE) =  \exp{  \left ( 2\pi \bi  \sum_{j \in \cE} k_j (t_j-u_j)/T_j  \right) }$ and $b(k_\cP) =  \exp{  \left ( 2\pi \bi  \sum_{j \in \cP} k_j (t_j-u_j)/T_j  \right) } $.
Recall the Dirichlet sum formula: for an integer  $u \in \{0,1, \dots T-1\}$,
\begin{equation*} \label{eq:exponsum}
     \sum_{k=0}^{T-1} \exp\left( \frac{2 \pi u}{T} k \bi \right)  = \left\{ \begin{array}{ll} T  & u=0  \\
                 0 & u \ne 0\end{array} \right . .
\end{equation*}
 Apply this to each exhaustively-sampled coordinate $j \in \cE$, obtaining:
 \[
  \frac{1}{T_j} \sum_{\cK_j} \exp{ ( 2\pi \bi k_j (t_j-u_j)/T_j ) } = \delta(t_j - u_j) , \qquad j \in \cE,
 \]
where $\delta()$ denotes the usual Kronecker symbol. We have
\beq \label{eq:blockGram}
  G_{t,u} = \frac{1}{\prod_{j \in \cP} T_j}   
       \sum_{k \in \cK_{\cP}} \exp{ \left ( 2\pi \bi  \sum_{j \in \cP} k_j (t_j-u_j)/T_j  \right) }  \times \prod_{j \in \cE} \delta( t_j-u_j) .
\eeq
We see that $G_{t,u} = 0 $ unless $t_j = u_j  \text{\ \ for all \ \ } j \in \cE $.
This indeed is the advertised block structure.
\end{proof}

\begin{lem}\label{singular-vector}
{\bf(Singular vectors of the Gram matrix)} 
Consider the Gram matrix $G= A^* A$ in a special case
of Lemma \ref{gram-matrix}, 
where  $d=2$ and $\cE =\{ 0 \}$, $\cP=\{1\}$, so $A$ implements anisotropic undersampling of 
the 2D Fourier transform on $T_0 \times T_1$ arrays. 
Namely, assume that the Fourier transform is
followed by selection of columns $k_{1,i}$,
$0 \leq k_{1,i} < T_1$ with
exhaustive sampling of all entries  $\{ (k_0,k_{1,i}): 0 \leq k_0 < T_0 \}$
in each selected column.
Necessarily $i=1,\dots, M \equiv n/T_0$.
By Lemma \ref{gram-matrix}, G is block-diagonal with $T_0$ identical blocks of size $T_1 \times T_1$. 
Let $G^{(1)}$ represent the upper left diagonal such $T_1 \times T_1$ block. 
Then, ${\rm rank} ({G}^{(1)}) = M $ and the $M$ principal eigenvectors of the $T_1$ by $T_1$ matrix
 $G^{(1)}$ are given by:
$$
V_{\ell} = (1,w_{\ell}, w_{\ell}^2, \dots, w_{\ell}^{T_1-1}), \quad \ell \in \cK_1,
$$
where $w_{\ell} = \exp(2\pi \bi \ell/T_1)$ and  $\cK_1 = (k_{1,i})_{i=1}^M$ denotes the collection of all sampled column indices.
\end{lem}

\begin{proof}
We prove that for $\ell \in \cK_1$, $V_{\ell}$ is an eigenvector by verifying $\sum_{\ell=0}^{N-1} {G}^{(1)}(t,u) V_\ell(u) = \lambda_\ell V_\ell(t)$, 
in fact by showing that $\lambda_\ell=1$. Lemma \ref{gram-matrix} -- specifically (\ref{eq:blockGram}) \-- gives us that  for $k=(k_0,k_1)$ and 
$$
G_{t,u} = \left( \frac{1}{\prod_{j \in \cP} T_j}   \sum_{k \in \cK_{\cP}} \exp{ (  2\pi \bi  \cdot \sum_{j\in \cP} k_j  (t_j-u_j)/T_j ) } \right)  \cdot \prod_{j\in \cE} \delta(t_j-u_j).
$$
Because $d=2$ and $k_0$ is sampled exhaustively,  the upper left $T_1 \times T_1$  block has the form:
$$
G_{(0,t),(0,u)} = \frac{1}{T_1}   \sum_{k \in \cK_{1}} \exp{ (2\pi \bi k  (t-u)/T_1 ) }, \quad (t,u) \in \{0,\dots,T_1-1\}^2,
$$
where now $k$,$t$, and $u$ are integers.
The matrix  $G^{(1)}$ has entries
$G^{(1)}(t,u) \equiv G_{(0,t),(0,u)} $ for $0 \leq t, u < T_1$. It has rank $M = \# \cT_1$ by inspection
of the preceding display. 
\begin{eqnarray*}
\sum_{u=0}^{T_1-1} {G}^{(1)}(t,u) V_{\ell}(u)&=&
\sum_{u=0}^{T_1-1} \left( \frac{1}{T_1}   \sum_{k \in \cK_{1}} \exp{ (2\pi \bi k  (t-u)/T_1 ) }\right) \exp\left( 2\pi \bi u \ell/T_1 \right) \\
&=&\sum_{k \in \cK_{1}} \left( \frac{1}{T_1}   \exp{ (2\pi \bi k  t/T_1 ) }\right)  \left(\sum_{u=0}^{T_1-1}  \exp\left( 2\pi \bi u (k-\ell)/T_1 \right)  \right)\\
&=&\sum_{k \in \cK_{1}} \left( \frac{1}{T_1}   \exp{ (2\pi \bi t k/T_1 ) }\right)  \left( T_1\ \delta\left(k-\ell \right) \right) \\
&=& \exp{ (2\pi \bi \ell t/T_1 ) } = V_\ell(t).
\end{eqnarray*}
\end{proof}

\begin{proof}[Proof of Theorem \ref{thm:RUID}]
Consider the two convex optimization problems
$$
{\rm (P_1)} \quad \min \| x \|_{1,\bC^{N}} \quad \text{subject to} \quad Ax=y,
$$
$$
{\rm (P_2)} \quad  \min \| x \|_{1,\bC^{N}} \quad \text{subject to} \quad A^* Ax=A^* y,
$$
where $A$ is an $n\times N \ (n<N)$ matrix having $n$ nonzero singular values (i.e., $A$ has full row rank). 
Problem {\rm $(P_1)$} is equivalent to {\rm $(P_2)$} because $A^*$ has full column rank $n$; hence their solution sets match. 
By Lemma \ref{gram-matrix}, $G=A^* A$ is block-diagonal. 
By separability of $\ell_1$ minimization, we can solve the  $T_0$ block subproblems 
each of size $T_1\times T_1$ individually. 
Because blocks are identical and ${\rm rank}(G) = {\color{blue}\lfloor\delta T_1\rfloor T_1}$, 
${\rm rank}(G^{b}) = {\color{blue}\lfloor\delta T_1}\rfloor \ \text{for blocks} \ b=1,\dots,N$. 
By Lemma \ref{rank-deficiency}, we know that we can solve equivalent full-row-rank problems of 
size ${\color{blue}\lfloor\delta T_1}\rfloor \times {\color{blue}T_1}$ as long as we find the right singular vectors. 
By Lemma \ref{singular-vector} we know that right singular vectors are defined by the partial Fourier matrix. 
\end{proof}

\section{Second Proof of Theorem \ref{thm:RUID}} \label{sec:AltProofRUID}
	 \newcommand{\Vmm}{V_{c;m_0,m_1}}
	 \newcommand{\VMm}{V_{c;M,m}}
	 \newcommand{\VMM}{V_{c;M,M}}	

We begin with terminology.
For an array $x = ( x(t_0,t_1), 0 \leq t_i \leq m_i)$, we call the collection of entries $x(\cdot, t_1)$ a {\it row}
and a collection $x(t_0,\cdot)$ a {\it column}. This is consistent with our depiction in Figure 3
of the main paper.

	Let $\Vmm$ denote the $vec$ operation taking 
        arrays in $\bC^{m_0 \times m_1}$ into vectors $\bC^{m_0 \cdot m_1}$ in {\it column-major} order;
         \[
	           (\Vmm x)(i_0 \cdot m_0  + i_1) = x(i_0,i_1), \qquad 0 \leq i < m_0; 0 \leq j < m_1.
	 \]
	 Thus $(\Vmm x)(0) = x(0,0)$, $(\Vmm x)(1) = x(0,1)$, $(\Vmm x)(2) = x(0,2)$, etc. 
	 
	 In the first half of the proof we will need $\VMM$ exclusively
	 and denote this simply $V$ for short.
	 Of course $V$ is an
	 $\ell_2$ isometry  which is also an $\ell_1$ isometry:
\begin{eqnarray*}
            \| x \|_{2,\bC^{M^2}} &=& \| V(x) \|_{2,\bC^{M \times M}}, \\
            \| x \|_{1,\bC^{M^2}} &=& \| V(x) \|_{1,\bC^{M \times M}}.
\end{eqnarray*}
       
\begin{lem}
        There is an $\ell_2$ isometry $T$  from $\bC^{Mm} \mapsto \bC^{M \times m}$ so that
	\beq \label{eq:isometry}
	T \circ A \circ V  = \cF_{\aus}.
	\eeq
\end{lem}

   \begin{proof}
       	
	We explicitly construct the isomorphism $T$ in (\ref{eq:isometry}). Let $\cF_{c} \equiv \cF_{c;m_0,m_1}$ denote the
	operator on $m_0 \times m_1$ arrays that applies the 1D discrete Fourier transform to each {\it column} separately,
	returning an $m_0 \times m_1$ array. Let $\cF_{r} \equiv \cF_{r;m_0,m_1}$ denote the
	operator on $m_0 \times m_1$ arrays that applies the 1D discrete Fourier transform to each {\it row} separately,
	returning an $m_0 \times m_1$ array.
	
	It is well-known that the $2D$ Discrete Fourier transform on $M \times M$ arrays has the factorization
	\[
	     \cF_2 = \cF_r \cF_c  =  \cF_{r; M,M} \cF_{c; M,M},
	\]
	the $1D$ Fourier transform of columns followed by the $1D$ Fourier transform of rows.
	 Let $\cK$ denote a collection of $m$ row indices $0 \leq k_i < M$
	 and let $\cS_{r,\cK}$ denote the operator from $M \times M$ arrays to $M \times m$ arrays that
	 simply selects those rows with indices in $\cK$. We observe the identity 
	\beq \label{eq:switchA}
	         \cS_{r,\cK} \cF_{r; M,M}  = \cF_{r; M,m}   \cS_{r,\cK}.
	 \eeq
	 In words, we can either first $1D$ Fourier transform each row individually, and then select certain rows,
	 or else we can select those same rows and then Fourier transform them; either way we get the same
	 outcome. Note that the two Fourier transform operators in this relation have different domains;
	 one operates on $M \times m$ arrays and one operates on $M \times M$ arrays.
	 
	 Our anisotropic undersampling operator has been defined by:
	 \[
	        \cF_{\aus} = \cS_{r,\cK} \cF_2.
	 \]
	 Based on the previous paragraph, we can equivalently write
	 \beq \label{eq:switchB}
	         \cF_{\aus} = \cF_r \cS_{r,\cK} \cF_c = \cF_{r; M,m}  \cS_{r,\cK} \cF_{c; M,M}.
	 \eeq
	 
	 Let now $\VMm$ be a $vec$ operator that maps from $M \times m$ arrays to $M \cdot m$ vectors,
	again by vectorizing in {\it column-major} order; namely, 
	 \[
	           \VMm(y)(i_0 M + i_1) = y(i_0,i_1) \qquad 0 \leq i_0, i_1  < M.
	 \]
	 Thus $(\VMm y)(0) = y(0,0)$, $(\VMm y)(1) = y(0,1)$, $(\VMm y)(2) = y(0,2)$, etc. Then of course $\VMm$ is an
	 isometry between $\bC^{M \times M}$ and $\bC^{M m}$, and so one-one.
	 
	 From now on the operator $\VMM$ previously denoted $V$, will always be spelled out as $\VMM$,
	 to keep domains and ranges unambigious.
	 
Now define  $T : \bC^{M m} \mapsto \bC^{M \times m}$ by
\beq \label{eq:Tdef}
	 T   = \cF_{r; M,m} \VMm^{-1}.
\eeq
	 In words, $T$ builds an $M \times m$ array and then applies the $1D$ Fourier transform
	 to each resulting row.
	 We now make the key observation:
\beq \label{eq:Aequivdef}
	       A =  \VMm \cS_{r,\cK} \cF_c \VMM^{-1} 
\eeq 

To check this, note first that the domain is indeed $\bC^{M^2}$ and the range is indeed $\bC^{Mm}$.
We  previously defined $A$ as a block diagonal operator $ I_{M} \otimes A^{(1)}$,
where $A^{(1)} : \bC^M \mapsto \bC^m$ is the pipeline $ A^{(1)} = \cS_{1,\cK}  \cF_{1}$ of two operators:
$\cF_1$ , a $1D$ Fourier transform of $M$-vectors followed by $\cS_{1,\cK}$ a selection
of certain elements out of those $M$ vectors. Checking definitions we see that
\[
      \cS_{r,\cK}  = \VMm^{-1}  (I_M \otimes \cS_{1,\cK} ) \VMM
\]
and
\[
   \cF_{c;M,M} = \VMm^{-1}  (I_M \otimes \cF_{1} ) \VMM.
\]
Hence
\begin{eqnarray*}
\VMm \cS_{r,\cK} \cF_c \VMM^{-1} &=&  \VMm  \left( \VMm^{-1}  (I_M \otimes \cS_{1,\cK} ) \VMM \right) \left(    \VMm^{-1}  (I_M \otimes \cF_{1} ) \VMM \right) \VMM^{-1} \\
 &=&    (I_M \otimes \cS_{1,\cK} )   (I_M \otimes \cF_{1} ) \\
 &=&    I_M \otimes (\cS_{1,\cK}  \cF_{1}) \\
 &=&   I_M \otimes  A^{(1)} \\
 &=&   A,
\end{eqnarray*}
which proves (\ref{eq:Aequivdef}).

We now verify (\ref{eq:isometry})
	\begin{eqnarray*}
	\cF_{\aus} &=& \cS_{r;\cK} \cF_2 \\
	               &=& \cS_{r;\cK} \cF_{r;M,M} \cF_{c;M,M} \\
	               &=& \cF_{r;M,m}  \cS_{r;\cK} \cF_{c;M,M}  \qquad \mbox{by } (\ref{eq:switchA})-(\ref{eq:switchB})\\
	               &=& \left( \cF_{r;M,m}   \VMm^{-1} \right) \left(  \VMm \cS_{r;\cK} \cF_{c;M,M} \VMM^{-1} \right)  \VMM  \\
	               &=& T A \VMM   \qquad \mbox{by } (\ref{eq:Tdef})-(\ref{eq:Aequivdef})
	\end{eqnarray*}
	where, as remarked earlier, both $T$ and $V$ are both isometries.
\end{proof}

We now use the representation $\cF_{\aus} = T AV$ to prove our main result.

\begin{proof} (of Theorem \ref{thm:RUID})
	
	 Fix $x_0$, generating undersampled measurements $\xaus = \cF_{\aus}(x_0)$.
        Consider the instance of $(P_{\aus})$ based on measurements vector
        $\xaus$. Let $x_1$ denote some specific solution of $(P_{\aus})$ .
        As a solution, it must obey the feasibility condition
        \[
           \cF_{\aus}(x_1) = \cF_{\aus}(x_0).
        \]
        
        Let $\bx_0 = V(x_0)$ and $\by_0 = A  \bx_0$ and 
        consider $\bx_1 = V(x_1)$
        as a candidate solution for $(P_{1,\bC})$ with data $\by_0$.
         We need to check that $\bx_1$ is feasible for
        $(P_{1,\bC})$  i.e. that $\by_0 = A \bx_1$.
       \begin{eqnarray*}
           \by_0                           &=& A \bx_0 \\
                          &=&  A V (x_0) \\
                        &=&  T^{-1} \cF_{\aus} ( x_0)  \\ 
                        &= & T^{-1} \cF_{\aus} ( x_1) \\
                         &=&  A V(x_1) \\
                        &=& A \bx_1 .
        \end{eqnarray*}
	So $\bx_1$ is indeed feasible for $(P_{1,\bC})$. It follows that
	\[
	  val(P_{1,\bC}) \leq  \| \bx_1 \|_{1,\bC^{M^2}} = \| x_1 \|_{1,\bC^{M \times M}} = val(P_{\aus}).
	\]
	Arguing in the other direction, let $\bx_1$ denote some solution of $(P_{1,\bC})$.
	We consider  $x_1 \equiv V^{-1}(\bx_1)$ as a candidate solution of  $(P_{\aus})$.
	From the feasibility of $\bx_1$ for $(P_{1,\bC})$ we have $A \bx_0 = A \bx_1  = \by_1$, say.
	We check  the feasibilty $\cF_{\aus}(x_1) = \cF_{\aus}(x_0)$:
      \begin{eqnarray*}
           \cF_{\aus}(x_0)    &=&   T A V (x_0) \\
                        &=&  T  \by_0 \\ 
                        &= & T \by_1 \\ 
                         &= & T A \bx_1 \\ 
                         &=&  T A V \cdot  V^{-1}(\bx_1)\\
                         &=&  T A V  x_1\\
                         &=& \cF_{\aus}(x_1) .
        \end{eqnarray*}
 We conclude that 	
	\[
  val(P_{1,\bC}) \geq  val(P_{\aus}).	
	\]
	Hence, $  val(P_{1,\bC}) =  val(P_{\aus})$.
	So the two problems have identical optimal values 
	and their solution sets are isomorphic
	under the vec mapping $V( \cdot )$.
\end{proof}

\section{Proof of Theorem \ref{thm:MultiD-NUS}}\label{app:NUS}

For $\bX = \bH_d$, the arguments of Appendix \ref{app:RUID} can all be redone, step-by-step, replacing 
the field $\bC$ by the hypercomplex algebra $\bH_d$. 
The notation and basic pattern of argument are given in \cite{Monajemi2016ACHA} and we won't repeat them.
The basic idea is as follows.
Let $n' = n/2^d$ and $N' = N/2^d$.
The matrix $A$ belongs to $\bH_d^{n' \times N'}$, the matrix $G = A^\#A$ belongs to $\bH_d^{N' \times N'}$ (here $\#$ denotes hypercomplex conjugation; again, see \cite{Monajemi2016ACHA} for details).  
The hypercomplex entries $x(i)$ can be viewed as $2^d$ dimensional
real vectors. The $\ell_1$ norm can then be written:
\[
  \| \bx \|_{1,H^d} = \sum_{i=1}^{N'} \| x(i) \|_{\ell_1^{2^d}(\bR)} .
\]
The arguments of the preceding section go through without essential changes; the Dirichlet sum
has this direct analog: 

\begin{equation*} \label{eq:exponsum}
     \sum_{t=0}^{T-1} \exp_{\bH_d}\left( \frac{2 \pi t}{T} u \bi \right)  = \left\{ \begin{array}{ll} T  & u =0  \\
                 0 & u \ne 0\end{array} \right . ,
\end{equation*}
where $u \in \{0,\dots, T-1\}$, and $\exp_{\bH_d}$ denotes the exponential function defined by the usual power series
within the associative algebra $\bH_d$. { For other choices of $\bX$, the theorem can be proved by realizing that the hypercomplex algebra $\bH_d$ is isomorphic to a subalgebra of the algebra of $2^d\times2^d$ matrices with real entries. The reader is referred to \cite{Monajemi2016ACHA,Monajemi_thesis_2016} for the details.}

\section{Comparison to exponential bounds by Donoho and Tanner}\label{app:exp-bound}
{
Donoho and Tanner \cite{DoTa10} give exponential bounds for the 
finite-$N$ probability of successful reconstruction for the coefficient fields
$\bR_+$and $\bR$. They consider the following condition on $\epsilon$ at certain $\delta$, 
$$
0 \le \epsilon \le {\easy(\delta)}{(1-R_\tau)}
$$
where $R_\tau$ is a certain multiplicative term having a real parameter $\tau \in (0,1)$
which, by their
bounds, implies
$$
P \{ {\bx_1 = \bx_0 }\} \geq 1-\tau.
$$
Taking $\tau = 1/M$ 
and $m \sim \delta M$ we get:
% Here,
% $$
% R_w (N, \bX, \delta, \tau ) = \left[ \frac{\log\left( c_w (\bX) \cdot (N+2)^6 / \tau \right)} {\delta \cdot N \cdot \Omega_w(\delta)}\right]^{1/2}
% $$
% where $c_w(\bR_+) = \frac{375 \sqrt{2}}{512\pi}$, and     $c_w(\bR) = \frac{625 \sqrt{2}}{512\pi^{3/2}}$. The exponent $\Omega_w (\delta) \ge 1/4$ can be numerically calculated. Using these results, and by setting appropriate values for $\tau$, we can find a \emph{lower bound} for the break-down point ${\epsilon}_N$, namely,

% $$
% {\epsilon}_N(\delta) \ge  {\epsilon^*(\delta)}{(1-R_w)}
% $$

$$
\frac{ {\easy(\delta)} - \esb(m,M)  } {\easy(\delta)} \le R_{1/M}
$$
where 
$$
R_{1/M} \simeq  c \cdot  \delta^{-1/2} \gamma_M
$$
Figure \ref{fig-exponential-bounds} depicts the lower bounds on $\esb(m,M)$ based on these bounds. In the case of real signals (cross-polytope), the formula obtained from the exponential bounds agrees, up to a proportionality constant, to our formula for $\eta$ following 
this article's (\ref{eq:PTforTensor}).
}
\begin{figure}[h]
\centering
\begin{tabular}{cc}
\includegraphics[width=2.5in,angle=0]%{/Users/hatef/Dropbox/TensorCS/ANALYSIS/analysis_for_proposal_defense/finite_bounds/finiteN_exponential_bounds_Pos} &
{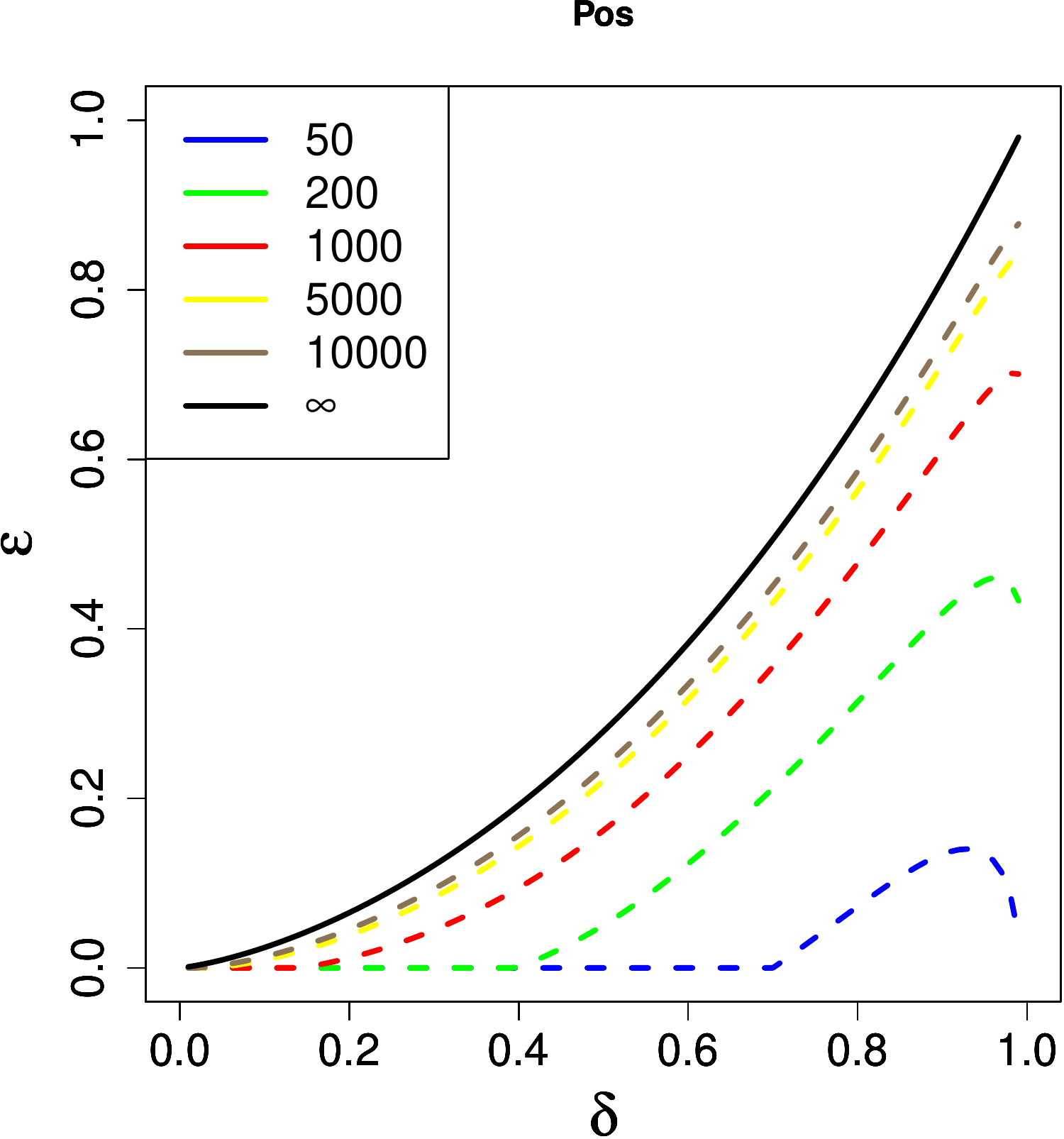} &
\includegraphics[width=2.5in,angle=0]{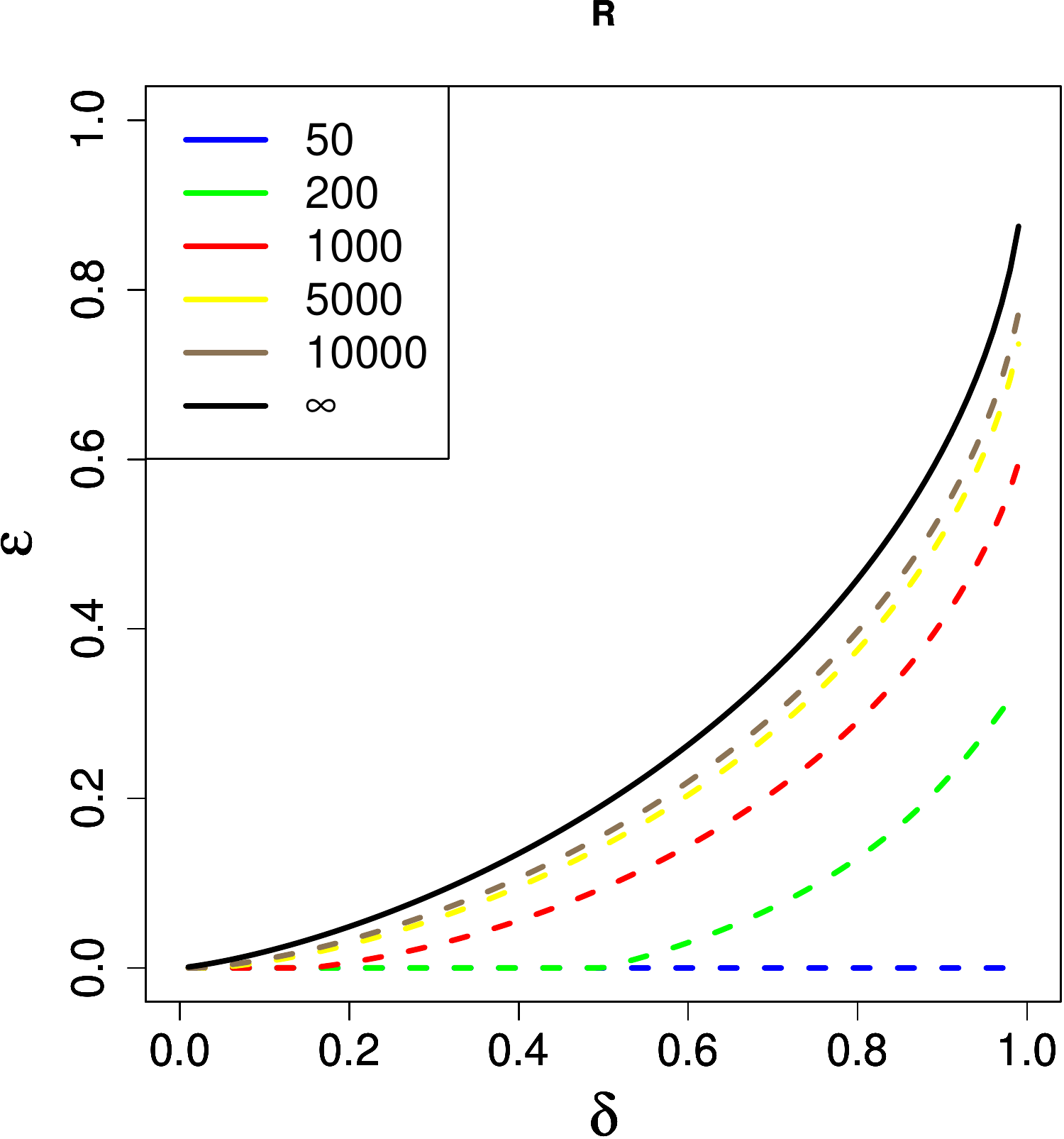}  \\
(a)  $\bR_+$  & (b)  $\bR$
\end{tabular}
\caption{Lower bound on $\esb(m,M)$ based on Donoho-Tanner exponential bounds. (left) Simplex ($\bR_+$), and (right) Cross-polytope ($\bR$)}
\label{fig-exponential-bounds}
\end{figure}

{
\bibliographystyle{plain}
\bibliography{refs,NMR-refs}
}

\end{document}